\documentclass{amsart}

\usepackage{graphicx} 
\usepackage{bbm}
\usepackage[lite]{amsrefs}
\usepackage{amssymb}
\usepackage[all,cmtip]{xy}
\usepackage{amsmath}
\usepackage{amsthm}
\usepackage{amsfonts}
\usepackage{hyperref}
\usepackage{enumerate}
\usepackage[titletoc]{appendix}
\usepackage{fancyhdr}
\usepackage[dvipsnames]{xcolor}
\usepackage{comment}
\usepackage{tikz-cd}
\usetikzlibrary{babel}
\usepackage{tikzpagenodes}
\usepackage{caption}
\usetikzlibrary[patterns]
\usepackage[margin = 1in]{geometry}
\usepackage{setspace}
\usepackage[T1]{fontenc}

\usepackage{booktabs}
\usepackage{array}

\providecommand{\cprime}{'}
\allowdisplaybreaks

 \newtheorem{thm}{Theorem}[section]

 \newtheorem{cor}[thm]{Corollary}
 \newtheorem{lem}[thm]{Lemma}
 \newtheorem{prop}[thm]{Proposition}

  \newtheorem{defn-thm}[thm]{Definition-Theorem}
   \newtheorem{ex}[thm]{Example}
   
   \newtheorem{rem}[thm]{Remark}

 \theoremstyle{definition}
 \newtheorem{defn}[thm]{Definition}
 
   \renewcommand\[{\begin{equation}}
\renewcommand\]{\end{equation}}
\numberwithin{equation}{section}

\usepackage{amsxtra}
\usepackage{graphicx}
\usepackage{newlfont}
\usepackage{amscd}
\usepackage{hyperref}
\usepackage[german,english]{babel}
\usepackage{cancel}
\usepackage{amsmath,amsthm,amssymb}
\usepackage{enumerate}
\usepackage{color}
\usepackage{varwidth}
\usepackage{mathrsfs}  
\usepackage{setspace}
\usepackage{collectbox}
\usepackage[utf8]{inputenc} 
\usepackage[T1]{fontenc}


\usepackage{graphicx}

\DeclareFontFamily{U}{mathx}{}
\DeclareFontShape{U}{mathx}{m}{n}{<-> mathx10}{}
\DeclareSymbolFont{mathx}{U}{mathx}{m}{n}
\DeclareMathAccent{\widehat}{0}{mathx}{"70}
\DeclareMathAccent{\widecheck}{0}{mathx}{"71}
\DeclareMathOperator{\Mat}{Mat}

\newcommand{\ZZ}{\mathbb{Z}}

\newcommand{\CC}{\mathbb{C}}

\newcommand{\n}{\mathfrak{n}}

\def \CC {\mathbb{C}}

\def \ZZ {\mathbb{Z}}

\def \Ical {\mathcal{I}}
\def \Jcal {\mathcal{J}}

\def \Vcal {\mathcal{V}}
\def \Wcal {\mathcal{W}}

\def \Sfr {\mathfrak{S}}

\def \dfr {\mathfrak{d}}

\def \gfr {\mathfrak{g}}
\def \hfr {\mathfrak{h}}

\def \lfr {\mathfrak{l}}

\def \nfr {\mathfrak{n}}

\def \pfr {\mathfrak{p}}

\def \sfr {\mathfrak{s}}

\def \hbar {\bar{h}}

\def \Vbf {\mathbf{V}}
\def \Wbf {\mathbf{W}}

\newcommand{\tr}{\text{tr}}

\allowdisplaybreaks
 
\begin{document}

\title[Poisson structures of deformed $W$-algebras and classical $W$-algebras]{Poisson structures of \protect \\
deformed $W$-algebras and classical $W$-algebras}

\author[D.J.Choi]{Dong Jun Choi}
\address[D.J.Choi]{Department of Mathematical Sciences, Seoul National University, Gwanak-ro 1, Gwanak-gu, Seoul 08826, Korea}
\email{djchoi9696@snu.ac.kr}

\author[U.R.Suh]{Uhi Rinn Suh}
\address[U.R.Suh]{Department of Mathematical Sciences and Research institute of Mathematics, Seoul National University, Gwanak-ro 1, Gwanak-gu, Seoul 08826, Korea}
\email{uhrisu1@snu.ac.kr}
\thanks{U.R. S. is supported by National Research Foundation of Korea (NRF) Grant No. 2022R1C1C1008698 and the Global-LAMP Program of NRF Grant No. RS-2023-00301976 funded by the Ministry of Education.
}

\date{} 

\begin{abstract}
In this paper, we explain how classical $W$-algebras can be obtained as limits of $q$-deformed $W$-algebras in type $A$. To this end, we give a detailed examination of the construction of deformed $W$-algebras via Hamiltonian reduction of deformed affine Poisson algebras. We first derive an affine Poisson vertex algebra (PVA) from the Poisson structure of the deformed affine Poisson algebra. This construction identifies a specific set of generators of the deformed $W$-algebra with generators of the classical $W$-algebra.
\end{abstract}

\maketitle

\section{Introduction}
\label{sec:int}

This paper develops an algebraic framework for the type $A$ $q$-difference Drinfeld-Sokolov reduction introduced by \cite{FRS98}, and establishes its relationship with classical affine $W$-algebras.

\subsection{History} \label{subsec:history}
The mathematical framework for classical $W$-algebras originally emerged from the study of Poisson structures associated with differential operators. Specifically, the space of differential operators of the form $\partial^N + w_{2}\partial^{N-2}+\dots +w_{N-1}\partial+ w_{N}$ is equipped with a well-known Poisson structure, widely recognized as the Adler-Gelfand-Dickey bracket in the theory of integrable systems \cites{A79, GD78}. In their seminal work, Drinfeld and Sokolov \cite{DS84} provided an interpretation of this Poisson structure through the framework of Hamiltonian reduction of the dual space of the affine Lie algebra $\widehat{\gfr}$ for every semisimple Lie algebra $\gfr$. The resulting Poisson algebra of functionals on the reduced space is now known as the classical $W$-algebra. On the other hand, Zamolodchikov \cite{Z85} introduced extensions of the Virasoro algebra by adding fields of higher conformal dimension, and Fateev and Lukyanov \cite{FL88} constructed the $W_N$-algebras corresponding to $\mathfrak{sl}_N$. Later, Feigin and Frenkel \cite{FF90a} provided a rigorous foundation via the quantum Drinfeld-Sokolov reduction, which for a simple Lie algebra $\mathfrak{g}$ produces the $W$-algebra $W^k(\mathfrak{g})$ at level $k$, generalizing $W_N$-algebras; see also \cites{F91, DSK06, FB}. Feigin and Frenkel \cite{FF92} proved that in the limit $k \to \infty$ the algebra $W^k(\mathfrak{g})$ becomes the classical $W$-algebra.

The study of two-parameter deformations of the $W_N$-algebras was initiated by Awata, Kubo, Odake and Shiraishi \cites{SKAO96, AKOS96} based on a bosonic system, independently of and in parallel with Feigin and Frenkel \cite{FF96}. This framework was further generalized by Frenkel and Reshetikhin \cite{FR98}, who introduced the two-parameter deformed $W$-algebra $\mathbf{W}_{q, t}(\mathfrak{g})$ associated with any simple Lie algebra $\mathfrak{g}$. It is a deformation of the affine $W$-algebra in the sense that, for a fixed $\beta$, taking the conformal limit $q \to 1$ with $t=q^\beta$ recovers the affine $W$-algebras. The algebra $\mathbf{W}_{q, t}(\mathfrak{g})$ has since been studied from various perspectives, including Langlands duality and string theory \cites{FH11a, FH11b, AFO18, T18}, four-dimensional supersymmetric gauge theories \cites{KP1, KP2, EP19, LNZ18, NZ20, HKc22}, representation theory and algebraic combinatorics \cites{M07, N22, FJM25, S02}, and quantum integrable systems \cites{FHR22, K22, AFR20}. Moreover, $\mathbf{W}_{q, t}(\mathfrak{g})$ retains a remarkably rich structure even under various limits, which are connected with integrable models associated with the quantum affine algebras $U_q(\widehat{\mathfrak{g}}),\; U_t({}^L(\widehat{{}^L\mathfrak{g}})) \text{, and }U_t({}^L\widehat{\mathfrak{g}})$. Among these, we focus in particular on the classical limit $q \to 1$, and study the resulting Poisson algebra $\mathbf{W}_{1, t}(\mathfrak{g})$.

When $\mathfrak{g}$ is simply-laced, the algebra $\mathbf{W}_{1, t}(\mathfrak{g})$ is isomorphic to $\mathbf{W}_{q, 1}(\mathfrak{g})$, whose structure is relatively well understood in terms of the center of a completion of the quantum affine algebra $U_q(\widehat{\mathfrak{g}})$ \cites{RS90, DP94, FR96}. The corresponding commutative algebra of generating fields can be identified with the Grothendieck ring of representations of $U_q(\widehat{\mathfrak{g}})$. Under this identification, the free field realization of $\mathbf{W}_{q, 1}(\mathfrak{g})$ coincides with the $q$-character homomorphism. For general types, it was conjectured in \cite{FR98}*{Conjecture 3} that $\mathbf{W}_{1, t}(\mathfrak{g})$ can be obtained via a $t^{r^\vee}$-deformed Drinfeld-Sokolov reduction of the loop group associated with $G$, introduced in \cites{FRS98, SS98}. The construction of \cites{FRS98, SS98} is a difference operator analogue of the classical theory \cite{DS84}. Frenkel, Reshetikhin, Semenov-Tian-Shansky and Sevostyanov \cites{FRS98, SS98} introduced an appropriate $q$-deformed Poisson structure on the loop group and carried out a Hamiltonian reduction with respect to a $q$-gauge action, showing that the resulting space is identified with a space of $q$-difference operators of the form $D^N+T_1D^{N-1}+\dots + T_{N-1}D+1$. It should be noted that the parameter $q$ in this context is distinct from the $q$ used in the notation $\mathbf{W}_{q, t}(\mathfrak{g})$.

The aim of this paper is to investigate the deformed $W$-algebra $\Wbf_{1,t}(\sfr\lfr_N)$ (equivalently, $\Wbf_{q,1}(\sfr\lfr_N)$), which is equipped with a difference algebra structure, and its relationship with the classical $W$-algebra $\Wbf(\sfr\lfr_N)$ endowed with a differential algebra structure.
From this perspective, a series of works by De Sole, Kac, and Valeri is particularly relevant. In the differential setting, they introduced a Poisson vertex algebra (PVA) framework for the Drinfeld--Sokolov reduction in \cite{DSKV13}, which was subsequently further developed in \cites{DSKV16a, DSKV16b, DSKV18}. Parallel to this, they also developed the theory of multiplicative Poisson vertex algebras (mPVAs) and the associated theory of integrable systems in the difference setting \cites{DSKVW19, DSKVW20}. Although the present paper is formulated in the language of \cites{FRS98, SS98}, we expect that the constructions and results can be translated naturally into the framework of mPVAs.

\subsection{Algebraic structures of deformed \texorpdfstring{$W$}{W}-algebras} 
In this paper, we focus on $q$-deformed $W$-algebras $\Wbf_q(LSL_N)$ of type $A$, which is isomorphic to $\mathbf{W}_{1,t}(\sfr\lfr_N)$ or $\mathbf{W}_{q,1}(\sfr\lfr_N)$ introduced in Section~\ref{subsec:history}. In Sections~\ref{sec:affine} and~\ref{sec:W-alg}, we provide a detailed algebraic description of the deformed $W$-algebras via Hamiltonian reduction.

The Drinfeld--Sokolov reduction for $\Wbf_q(LSL_N)$ introduced in \cite{FRS98} plays a central role in our approach. In this framework, $\Wbf_q(LSL_N)$ is realized as a Hamiltonian reduction of the deformed affine Poisson algebra $\Vbf_q(LSL_N)$, which is generated by the formal series $\mu_{ij}(z)$ for $i,j=1,\dots,N$. The algebra is equipped with the difference operator
$D\mu_{ij}(z)=\mu_{ij}(qz),$ 
together with the relation $\det(z)=1$. The Poisson structure on $\Vbf_q(LSL_N)$ is defined through an operator satisfying the modified classical Yang--Baxter equation (see Appendices~\ref{AppendixA} and~\ref{AppendixB}). While it is generally understood that this construction ensures the well-definedness of the Poisson bracket, we were unable to find a complete proof in the existing literature. We therefore provide a detailed proof in the appendices for the reader's convenience. Moreover, we derive in Section~\ref{subsec:deform_affine} the explicit formula of Poisson brackets between the generating series : 
\begin{equation} \label{eq:PoissonBracketDefinition-intro}
    \begin{aligned}
        \{\mu_{ij}(z), \mu_{kl}(w)\}
        &=C_{ijkl}\left(\frac{w}{z}\right) \mu_{ij}(z)\mu_{kl}(w) +\frac{1}{2}(\varepsilon_{ik}+\varepsilon_{lj})\delta\left(\frac{w}{z}\right)\mu_{kj}(z)\mu_{il}(w) \\
        &-\delta_{jk}\sum_{\alpha > j} \delta\left(\frac{qw}{z}\right)\mu_{i\alpha}(z)\mu_{\alpha l}(w) + \delta_{il} \sum_{\alpha > i} \delta\left(\frac{w}{qz}\right) \mu_{\alpha j}(z) \mu_{k\alpha}(w),
    \end{aligned}
\end{equation}
where $C_{ijkl}\big(\frac{w}{z}\big)$ is given in \eqref{eq:C}. 

The deformed $W$-algebra $\Wbf_q(LSL_N)$ was defined in \cite{FRS98} as the subspace of $\Vbf_q(LSL_N)$ consisting of elements invariant under the $q$-gauge action of the unipotent subgroup $N_+$. In Definition \ref{def:W}, we introduced an equivalent definition of $\Wbf_q(LSL_N)$ in terms of a Lie algebra action. Using these two constructions, we showed that the elements $T_r(z)$, for $r=1,\dots,N-1$, form a generating set of $\Wbf_q(LSL_N)$, as expected in \cite{FRS98}, where 
\begin{equation} \label{eq:cdet-intro}
    \text{cdet}\begin{pmatrix}
    & & \\
    & \delta_{ij} D+ \pi(\mu_{ij}(z))& \\
    & & 
\end{pmatrix}
= D^{N}+ T_1(z) D^{N-1}+\dots+T_{N-1}(z) D+1,
\end{equation}
for $\pi$ given in \eqref{eq:pi}. Again, it seems that the well-definedness of Poisson bracket on $\Wbf_q(LSL_N)$, induced from that of $\Vbf_q(LSL_N)$, is already known. However, we were unable to find a complete proof in the literature, we include one in Appendix~\ref{AppendixC}.

The deformed $W$-algebra $\Wbf_q(LSL_N)$ can be embedded into the $q$-difference subalgebra $\Vbf_q(H)$ of $\Vbf_q(LSL_N)$, which is generated by $\mu_{rr}(z)$ for $r=1,2,\dots, N.$ The embedding
\begin{equation} \label{eq:q-miura-intro}
m_q: \Wbf_q(LSL_N) \hookrightarrow \Vbf_q(H)
\end{equation}
is called the deformed Miura transformation. Here we note that while $\Vbf_q(H)$ is a $q$-difference subalgebra of $\Vbf_q(LSL_N)$, it is not a Poisson subalgebra since the bracket \eqref{eq:PoissonBracketDefinition-intro} is not closed in $\Vbf_q(H)$. However, by defining the Poisson bracket on $\Vbf_q(H)$ as
\begin{equation*}
     \{\mu_{ii}(z), \mu_{jj}(w)\}
        =C_{iijj}\left(\frac{w}{z}\right) \mu_{ii}(z)\mu_{jj}(w),
\end{equation*}
one can make \eqref{eq:q-miura-intro} into a Poisson algebra homomorphism. In Section~\ref{subsec:q-miura}, we included detailed proofs of the facts on the Miura transformation.


\subsection{Classical \texorpdfstring{$W$}{W}-algebras from deformed \texorpdfstring{$W$}{W}-algebras}

By construction, $\Wbf_q(LSL_N)$ is naturally viewed as a deformation of the classical $W$-algebra $\Wbf(\sfr\lfr_N)$. However, the precise relationship between $\Wbf_q(LSL_N)$ and $\Wbf(\sfr\lfr_N)$, particularly the mechanism by which the Poisson structure of the latter can be recovered from that of the former, has not been thoroughly investigated. In Sections~\ref{sec:affinelimit} and ~\ref{sec:W-w}, we address this question and, in particular, establish Theorem~\ref{thm:main}, inspired by the approach in \cite{FR96}*{Section 10.2}.

In Section~\ref{sec:affinelimit}, we begin by investigating the affine case. Since the Poisson bracket given in \eqref{eq:PoissonBracketDefinition-intro} on $\Vbf_q(LSL_N)$ is completely determined by the Poisson brackets among its Fourier coefficients, we consider the $\mathbb{C}(q)$-algebra $\Vcal_q(LSL_N)$ generated by the Fourier modes of the elements in $\Vbf_q(LSL_N)$. We then construct a Poisson $\mathbb{C}(\!(h)\!)$-algebra $\widetilde{\Vcal}_h(\sfr\lfr_N)$, setting $q=e^{-h}$, which captures the behavior of the Poisson structure of $\Vcal_q(LSL_N)$ as $h\to 0$. More precisely, $\widetilde{\Vcal}_h(\sfr\lfr_N)$ is generated by the Fourier coefficients of the formal series in $\mathbb{C}[\mathsf{E}_{ij}(z) | i,j=1,\ldots,N]$, and its Poisson bracket is defined by
\begin{equation*}
   \{\mathsf{E}_{ij}[m],\mathsf{E}_{kl}[n]\}_h
= \frac{1}{h}\,\iota \, \left\{\mu_{ji}[m],\mu_{lk}[n]\right\},
\end{equation*}
where the map $\iota:\Vcal_q(LSL_N)\to \widetilde{\Vcal}_h(\sfr\lfr_N)$, describing the correspondence between the generating series $\mu_{ji}(z)$ and $\mathsf{E}_{ij}(z)$, is given by
$
    \iota\bigl(\mu_{ji}[m]\bigr)
= \delta_{ij}\delta_{m0} + h\mathsf{E}_{ij}[m].
$


To take the limit $h \to 0$, we introduce an $A_1$-subalgebra $\Vcal_h(\sfr\lfr_N)$ of $\widetilde{\Vcal}_h(\sfr\lfr_N)$ consisting of elements that admit a well-defined limit as $h \to 0$, where $A_1$ is the ring of rational functions in $q$ regular at $q=1$. Especially, $\Vcal_h(\sfr\lfr_N)$ contains $\mathsf{E}^{0,l}_{ij}[m]=\mathsf{E}_{ij}[m]$, 
    \begin{equation} \label{eq:generator_intro}
            \mathsf{E}^{k,l}_{ij}[m]=\prod_{s=0}^{k-1} \left(\frac{1-q^{(-m-1)l}}{lh}-s\right)\mathsf{E}_{ij}[m]
    \end{equation}
  for $1 \le i, j \le N$, $l \in \ZZ_{>0}, k \in \ZZ_{\ge 0}$, and $m \in \ZZ$, and the elements
  \begin{equation} \label{eq:generator_2_intro}
      \sum_{m_1+\cdots + m_k = M} c_{m_1m_2\cdots m_k}(h)\cdot \mathsf{E}_{i_1j_1}[m_1] \dots \mathsf{E}_{i_kj_k}[m_k]
  \end{equation}
  for $c_{m_1m_2\cdots m_k}(h)\in h\mathbb{C}[\![h]\!].$ Note that the coefficients appearing on the right-hand side of \eqref{eq:generator_intro} give rise, in the limit $h\to 0$, to the differential operator $z^k\partial_z^k$ acting on the generating series (see \eqref{eq:motiv_generator}). Furthermore, the expressions given in \eqref{eq:generator_2_intro} vanish in the limit $h\to 0$. Now we obtain the first main theorem in this paper.
  
  \begin{thm} [Theorem~\ref{thm:main_affine}] \label{thm:main-1,intro}
  Let $\Vcal(\sfr\lfr_N)$ be the $\mathbb{C}$-algebra generated by the Fourier series of elements in
  \begin{equation*}
      \Vbf(\sfr\lfr_N) = \mathbb{C}[\partial^n e_{ij}(z)|i,j=1,\cdots, N, \, n\in \mathbb{Z}_{\ge 0}] \big/ \big< e_{11}(z) + \dots + e_{NN}(z) \big>.
  \end{equation*} 
  Then the limit 
  \begin{equation*}
      \lim_{h\to 0} : \Vcal_h(\sfr\lfr_N) \to \Vcal(\sfr\lfr_N)
  \end{equation*}
  is a well-defined surjective map, and the Poisson structure on $ \Vcal_h(\sfr\lfr_N)$ induces the affine PVA structure on $\Vcal(\sfr\lfr_N)$ by identifying $\mathsf{E}_{ij}(z)$ with the field $e_{ij}(z)-\frac{\delta_{ij}}{N}I_N(z)$ corresponding to the element $e_{ij}-\frac{\delta_{ij}}{N}I_N\in \sfr\lfr_N.$ Moreover, the analogous statement holds when $\sfr\lfr_N$ is replaced by its Cartan subalgebra $\hfr.$
\end{thm}

\noindent
The following diagram summarizes the limit process from $\Vbf_q(G)$ to $\Vcal(\gfr)$ for $(G,\gfr)=(LSL_N,\sfr\lfr_N)$ or $(H,\hfr)$:
\begin{equation*}
    \begin{tikzcd}[column sep=large]
    \Vbf_q(G) \arrow[d, dashed, "\text{\scriptsize Fourier modes}"'] & & \Vbf(\gfr) \arrow[d, dashed, "\text{\scriptsize Fourier modes}"]\\
    \Vcal_q(G) \arrow[r, hook, "\iota"] & \widetilde{\Vcal}_h(\gfr) \supset \Vcal_h(\gfr) \arrow[r, two heads, "\lim_{h\to 0}"] & \Vcal(\gfr)
    \end{tikzcd}
\end{equation*}

\vskip 2mm

Another main theorem of this paper is the $W$-algebra analogue of Theorem~\ref{thm:main-1,intro}. Recall that $\Wbf_q(LSL_N)$ is the $q$-difference algebra generated by the coefficients appearing in \eqref{eq:cdet-intro}. Likewise, the generators of the classical $W$-algebra can be obtained from the column determinant of a certain matrix (see \eqref{eq:cdetnondeformed}). Since the Miura transform in each case is an injective homomorphism, we can apply it to the determinant form and obtain the generators in the image of the Miura transform. In other words, the coefficients $T_1(z), \cdots, T_{N-1}(z)\in \Vbf_q(H)$ of
\begin{equation*}
(D+\mu_{11}(z))(D+\mu_{22}(z))\cdots (D+\mu_{NN}(z))= D^N+T_1(z)D^{N-1}+\cdots + T_{N-1}(z)D+1, 
\end{equation*}
generate $\Wbf_q(LSL_N)$ and the coefficients $w_2(z), \cdots, w_N(z)\in \Vbf(\hfr)$ of 
\begin{equation*}
(\partial+e_{11}(z))(\partial+e_{22}(z))\cdots (\partial+e_{NN}(z)) = \partial^N+ w_2(z) \partial^{N-2}+\cdots + w_{N-1}(z)\partial+ w_{N}(z),
\end{equation*}
generate $\Wbf(\sfr\lfr_N).$ Our strategy is to establish a connection between the generators of $T_1(z), \cdots, T_{N-1}(z)$ and $w_2(z), \cdots, w_N(z).$

As in the affine case, we consider the $\mathbb{C}(q)$-algebra $\Wcal_q(LSL_N)$ generated by the Fourier modes of $\Wbf_q(LSL_N)$, together with the $\mathbb{C}(\!(h)\!)$-algebra $\widetilde{\Wcal}_h(\sfr\lfr_N)$, defined as the image of $\Wbf_q(LSL_N)$ under the map $\iota$. Recall that the Poisson bracket on $\widetilde{\Wcal}_h(\sfr\lfr_N)$ is inherited from that on $\widetilde{\Vcal}_h(H)$, while the Poisson structure of the affine PVA $\Vcal(\hfr)$ is induced by the corresponding Poisson structure on $\widetilde{\Vcal}_h(H)$. Therefore, if every element of the classical $W$-algebra can be realized as the limit of elements in $\widetilde{\Wcal}_h(\sfr\lfr_N)$, it follows immediately that the Poisson structure on the classical $W$-algebra also arises as the limit of the Poisson structure on the deformed $W$-algebra.

To identify elements whose limits produce the generators of the classical $W$-algebra, we introduce the elements $\mathfrak{T}_{r}[m] \in \widetilde{\Wcal}_h(\sfr\lfr_N)$ for $1\le r\le N-1$ and $m\in\mathbb{Z}$, defined by
\begin{equation*}
\begin{aligned}
   \mathfrak{T}_{r}[m]\textstyle &:=\frac{1}{h^{r+1}}\Big(\mathsf{T}_r[m] - q^{-1}\binom{N-r}{1}_{q^{-1}} \mathsf{T}_{r-1}[m]+ \\
   &\quad \dots + (-1)^{r-1}q^{-\frac{(r-1)r}{2}}\binom{N-2}{r-1}_{q^{-1}}\mathsf{T}_1[m] + (-1)^r q^{-\frac{r(r+1)}{2}}\binom{N-1}{r}_{q^{-1}} \delta_{m, 0} - \delta_{m, 0} \Big),
    \end{aligned}
\end{equation*}
where $T_r(z)=\sum_{m\in \mathbb{Z}}T_r[m] z^{-m-r}$ and $\mathsf{T}_r[m]= \iota(T_r[m])$. Defining the formal generating series as
$\mathfrak{T}_r(z) := \sum_{m \in \ZZ} \mathfrak{T}_r[m]z^{-m-r-1}$,
we can state the second main theorem of this paper.

\begin{thm}[Theorem \ref{thm:main}]\label{thm:main_intro}
Let $\Wcal(\sfr\lfr_N)$ be the $\mathbb{C}$-algebra generated by the Fourier modes of elements in 
  \begin{equation*}
      \Wbf(\sfr\lfr_N) = \mathbb{C}[\partial^n w_i(z)|i=2,\cdots, N, \, n\in \mathbb{Z}_{\ge 0}].
  \end{equation*}
Then each $\mathfrak{T}_r[m]$ admits a well-defined limit as $h\to 0$, that is, $
     \mathfrak{T}_r[m] \in \Wcal_h(\sfr\lfr_N):=\widetilde{\Wcal}_h(\sfr\lfr_N) \cap \Vcal_h(\hfr)$. Furthermore,
\begin{equation*}
\begin{aligned}
\lim_{h \to 0} : \Wcal_h(\sfr\lfr_N) \to \Wcal(\sfr\lfr_N)
\end{aligned}
\end{equation*}
is a well-defined surjective homomorphism such that 
$\lim_{h \to 0} \mathfrak{T}_r[m] = -w_{r+1}[m]$, or equivalently, $\lim_{h \to 0} \mathfrak{T}_r(z) = -w_{r+1}(z)$. Moreover, the Poisson structure on $\widetilde{\Wcal}_h(\sfr\lfr_N)$ induces the Poisson structure on $\Wcal(\sfr\lfr_N).$
\end{thm}

\noindent Finally, the following diagram summarizes the discussion above:

\begin{equation*}
\begin{tikzcd}[column sep=large, row sep=large]
\Wbf_q(LSL_N) \arrow[d, dashed, "\text{\scriptsize Fourier modes}"'] & & \Wbf(\sfr\lfr_N) \arrow[d, dashed, "\text{\scriptsize Fourier modes}"]\\
\Wcal_q(LSL_N)
  \arrow[r, hook, "\iota"]
  \arrow[d, hook]
&
\widetilde{\Wcal}_h(\sfr\lfr_N)
  \supset
\Wcal_h(\sfr\lfr_N)
  \arrow[r, two heads, "\lim_{h\to 0}"]
  \arrow[d, hook]
&
\Wcal(\sfr\lfr_N)
  \arrow[d, hook]
\\
\Vcal_q(H)
  \arrow[r, hook, "\iota"]
&
\widetilde{\Vcal}_h(\hfr)
  \supset
\Vcal_h(\hfr)
  \arrow[r, two heads, "\lim_{h\to 0}"]
&
\Vcal(\hfr).
\end{tikzcd}
\end{equation*}

\vskip 3mm

\subsection{Organization of the paper}
This paper is organized as follows. In Section~\ref{sec:affine}, we review the affine PVA $\Vbf(\sfr\lfr_N)$ and the deformed affine Poisson algebra $\Vbf_q(LSL_N)$. In Section~\ref{subsec:affine PVA}, we recall the definition and basic properties of affine PVAs, while Section~\ref{subsec:deform_affine} introduces the deformed affine Poisson algebra and presents explicit formulas for its Poisson brackets. Although these brackets have already appeared in other literature (for example, \cites{FRS98, SS98}) in a geometric form based on the $R$-matrix formalism, we explain how to derive the Poisson brackets among generators from this description. Furthermore, since we could not find a complete proof in the existing literature, we provide a proof of the well-definedness of the Poisson bracket in Appendix~\ref{AppendixB}. In Section~\ref{sec:W-alg}, we study the algebraic structures of the classical $W$-algebra $\Wbf(\sfr\lfr_N)$ and the deformed $W$-algebra $\Wbf_q(LSL_N)$. Section~\ref{subsec:classical W} reviews the structure of $\Wbf(\sfr\lfr_N)$ and its Miura transformation. In Sections~\ref{subsec:deformed W}--\ref{subsec:q-miura}, we carefully examine the Drinfeld--Sokolov construction for $\Wbf_q(LSL_N)$, together with explicit formulas for its generators and Miura transformation. In addition, Appendix~\ref{AppendixC} contains a complete proof of the well-definedness of the Poisson bracket on $\Wbf_q(LSL_N)$. Finally, in Sections~\ref{sec:affinelimit} and ~\ref{sec:W-w}, we explain how the affine PVA $\Vbf(\sfr\lfr_N)$ and the classical $W$-algebra $\Wbf(\sfr\lfr_N)$ arise as limits of the deformed structures $\Vbf_q(LSL_N)$ and $\Wbf_q(LSL_N)$, respectively.

\subsection*{Acknowledgments}
The first author thanks Sin-Myung Lee for a helpful discussion in the early stages of this work.

\section{Affine Poisson vertex algebra and deformed affine Poisson algebra} \label{sec:affine}
In this section, we introduce affine Poisson vertex algebras and deformed affine Poisson algebras. They provide the appropriate setting for Hamiltonian reduction. The affine Poisson vertex algebras are equipped with differential and Poisson structures, while the deformed affine Poisson algebras are equipped with $q$-difference and Poisson structures.

\subsection{Affine Poisson Vertex Algebra} \label{subsec:affine PVA}
Let $\gfr$ be a finite-dimensional simple Lie algebra over $\CC$ and fix a symmetric invariant bilinear form $\langle \cdot, \cdot \rangle$. Consider the corresponding affine Lie algebra $\widehat{\gfr} = \gfr \otimes \CC[t, t^{-1}] \oplus \CC K$. For $X, Y \in \gfr$ and $m, n \in \ZZ$, the commutation relations are given by
\begin{equation} \label{eq:affine bracket}
    [X[m], Y[n]] = [X, Y][m+n] + m\delta_{m+n, 0} \langle X, Y \rangle K,
\end{equation}
where $X[m] := X \otimes t^m$ and $K$ is a central element. The bracket \eqref{eq:affine bracket} can be rewritten in terms of the formal generating series 
\begin{equation}
    X(z) := \sum_{n \in \ZZ} X[n]z^{-n-1},
\end{equation}
for $X \in \gfr$ as follows: 
\begin{equation}\label{eq:affinePoisson}
    [X(z), Y(w)] = [X, Y](w) \delta(z, w) + \langle X, Y \rangle \partial_w \delta(z, w).
\end{equation}
Here
\begin{equation*}
    \delta(z, w) := \sum_{n \in \ZZ} z^{-n-1}w^n
\end{equation*}
denotes the formal delta distribution. The \emph{affine Poisson vertex algebra (affine PVA)} is the commutative differential algebra
\begin{equation*}
\Vbf(\gfr)=\mathbb{C}[\partial^nX(z)\mid X\in \gfr, \, n\in \mathbb{Z}_{\ge 0}]
\end{equation*}
endowed with the Poisson bracket induced from \eqref{eq:affinePoisson}. More precisely, $\{X(z), Y(w)\} = [X, Y](w) \delta(z, w) + \langle X, Y \rangle \partial_w \delta(z, w)$ for $X,Y\in \gfr$, and the bracket between any elements $X_1(z), X_2(z), X_3(z) \in \Vbf(\gfr)$ is defined by the following properties: 
\begin{enumerate}
    \item (Leibniz rule) \quad 
    $ \left\{ X_1(z), X_2(w)X_3(w)\right\} = \left\{ X_1(z), X_2(w) \right\} X_3(w) + \left\{X_1(z), X_3(w) \right\} X_2(w)$;
    \item (Compatibility with derivation) 
    \begin{equation*}\left\{ \partial X_1(z), X_2(w)\right\} = \partial_z \left\{X_1(z), X_2(w)\right\}, \quad \left\{ X_1(z), \partial X_2(w) \right\} = \partial_w \left\{ X_1(z), X_2(w)\right\}.\end{equation*}
\end{enumerate}
In addition, one can easily check the following properties 
\begin{enumerate}
    \item (Skew-symmetry) 
    $ \left\{X_1(z), X_2(w)\right\} = -\left\{X_2(w), X_1(z)\right\};$
    \item (Jacobi identity) 
    $\left\{X_1(z), \left\{X_2(w), X_3(u)\right\} \right\} - \left\{ X_2(w), \left\{ X_1(z), X_3(u)\right\}\right\} = \left\{ \left\{ X_1(z), X_2(w)\right\}, X_3(u)\right\},$
\end{enumerate}
from the corresponding properties for the brackets between generators.

This framework is formally captured in the $\lambda$-bracket language of PVAs, extensively developed in a series of papers \cites{BDSK09, DSK13, DSKV13, DSKV14, DSKV15, DSKV16a, DSKV16b, DSKV18}. For example, by encoding $\partial_w^n \delta(z, w)$ as the formal variable $\lambda^n$, the compatibility with the derivation is expressed by the sesquilinearity
\begin{equation*}
    \left\{ \partial X {}_\lambda Y \right\} = -\lambda \left\{ X {}_\lambda Y \right\}, \left\{ X {}_\lambda \partial Y\right\} = (\partial + \lambda) \left\{ X {}_\lambda Y \right\}.
\end{equation*}
of the $\lambda$-bracket.
In this paper, we use the formal generating series notation instead of $\lambda$-brackets, since it is more convenient for comparison with the deformed Poisson algebras.

\subsection{Deformed affine Poisson algebra} \label{subsec:deform_affine}

In this subsection, we introduce the deformed affine Poisson algebra associated with the loop Lie group $LSL_N$, originally introduced in \cite{FRS98} and \cite{SS98}. Its Poisson structure is induced from the Poisson--Lie group structure on $LSL_N$ associated with the classical $r$-matrix of $L\sfr\lfr_N$. Based on these results, we provide an explicit algebraic formalism.

\vskip 2mm
Throughout this paper, $N \ge 2$ denotes an integer and $q$ denotes a nonzero complex number that is not a root of unity. For an element $A(t) \in LSL_N$, we view each of its entries as an element of the Laurent polynomial ring  $\CC[t, t^{-1}]$. Consider the functionals $\mu_{ij}[n] : LSL_N \to \CC$, which extract the Fourier mode of the $(i, j)$-th entry:
\begin{equation*}
    \mu_{ij}[m](A(t)) = A_{ij}[m], \quad \text{ where } \quad A_{ij}(t) = \sum_{m \in \ZZ} A_{ij}[m]t^{-m}.
\end{equation*}
Let us define the formal generating series
\begin{equation*}
\mu_{ij}(z) := \sum_{m \in \ZZ}\mu_{ij}[m]z^{-m-1}
\end{equation*}
for $1\le i, j \le N$ and denote 
\begin{equation*}
    \det(z) = \sum_{n \in \ZZ} \det[m] z^{-m-N} := \sum_{\sigma \in \Sfr_N} (-1)^{\mathrm{sgn}(\sigma)} \mu_{1 \sigma(1)}(z) \mu_{2 \sigma(2)}(z) \cdots \mu_{N \sigma(N)}(z).
\end{equation*}
We now introduce one of the main objects of this paper.

\begin{defn}\label{def:affine} \cite{FRS98}
    The \emph{deformed affine Poisson algebra} associated with $LSL_N$ is the commutative algebra
    \begin{equation*}
        \Vbf_q(LSL_N) := \Vbf_q(L\Mat_N) \left/\big\langle {\det}(q^mz) -1 \bigm| m \in \ZZ \big\rangle\right. ,
    \end{equation*}
    where $\Vbf_q(L\Mat_N)$ is the free commutative algebra generated by the formal generating series $\mu_{ij}(q^mz)$:
    \begin{equation*}
    \Vbf_q(L\Mat_N) := \CC\left[\mu_{ij}(q^mz) \mid m \in \ZZ, 1 \le i, j \le N\right].
    \end{equation*}
\end{defn}

Let $D$ be the $q$-difference operator acting on a formal generating series 
    \begin{equation*}
        f(z) = \sum_{m\in \ZZ} f[m]z^{-m-1}
    \end{equation*}
via
    \begin{equation*}
        (Df)(z) := f(qz) = \sum_{m\in \ZZ} f[m]q^{-m-1}z^{-m-1}.
    \end{equation*}
A commutative algebra closed under the action of $D$ is called a \emph{$q$-difference algebra}. Both $\Vbf_q(LSL_N)$ and $\Vbf_q(L\Mat_N)$ are $q$-difference algebras. For convenience, we denote the image of $\mu_{ij}(z) \in \Vbf_q(L\Mat_N)$ under the quotient map
    \begin{equation*}
        \Vbf_q(L\Mat_N) \to \Vbf_q(LSL_N)
    \end{equation*}
simply by $\mu_{ij}(z)$.

For the remainder of this section, we define the Poisson bracket on $\Vbf_q(LSL_N)$. Note that the Poisson bracket given in \eqref{eq:PoissonBracketDefinition} is motivated by the geometry of Poisson-Lie groups; see Section~\ref{subsec:B1}. To this end, we introduce the following notation:
\begin{equation*}
\begin{aligned}
    &\delta_{\alpha>\beta}= \begin{cases}
        1 \quad &\text{if } \alpha>\beta \\ 0 \quad &\text{otherwise}
    \end{cases}, \quad \varepsilon_{\alpha\beta} = \delta_{\alpha>\beta} - \delta_{\beta>\alpha} = \begin{cases}
        1 \quad &\text{if  } \alpha>\beta,\\
        0 \quad &\text{if  } \alpha=\beta,\\
        -1 \quad &\text{if  } \alpha<\beta.
    \end{cases}\\
    &\Psi_{\alpha \beta}[r] = \frac{q^{r|\alpha-\beta|_N}}{1-q^{Nr}}, \quad \Omega_{\alpha\beta}[r]=\Psi_{\alpha\beta}[r]-\delta_{\alpha\beta}, \\
    &\Phi_{\alpha \beta}[r] = \Psi_{\alpha \beta}[r] + \Omega_{\alpha \beta}[r] = \begin{cases}
        \frac{2q^{(\alpha-\beta)r}}{1-q^{Nr}} \quad &\text{if } \alpha > \beta, \\
        \frac{1+q^{Nr}}{1-q^{Nr}} \quad &\text{if } \alpha = \beta,\\
        \frac{2q^{(N+\alpha-\beta)r}}{1-q^{Nr}} \quad &\text{if } \alpha < \beta,
    \end{cases}
    \end{aligned}
\end{equation*}
where $|\alpha-\beta|_N \in \{0, 1, \dots, N-1\}$ denotes $\alpha-\beta$ if $\alpha \ge \beta$ and $N+\alpha-\beta$ if $\alpha<\beta$. Furthermore, we define
\begin{equation} \label{eq:C}
    C_{ijkl}[r] = \frac{1}{2}\Phi_{ik}[r] + \frac{1}{2}\Phi_{jl}[r] - q^{r}\Psi_{jk}[r] - q^{-r}\Omega_{il}[r].
\end{equation}
Now consider the bracket on the generators given by
\begin{equation} \label{eq:PoissonBracketDefinition}
    \begin{aligned}
        \{\mu_{ij}(z), \mu_{kl}(w)\}&= \sum_{m, n \in \ZZ} \{\mu_{ij}[m], \mu_{kl}[n]\}z^{-m-1}w^{-n-1}\\
        &=C_{ijkl}\left(\frac{w}{z}\right) \mu_{ij}(z)\mu_{kl}(w) +\frac{1}{2}(\varepsilon_{ik}+\varepsilon_{lj})\delta\left(\frac{w}{z}\right)\mu_{kj}(z)\mu_{il}(w) \\
        &\quad -\delta_{jk}\sum_{\alpha > j} \delta\left(\frac{qw}{z}\right)\mu_{i\alpha}(z)\mu_{\alpha l}(w) + \delta_{il} \sum_{\alpha > i} \delta\left(\frac{w}{qz}\right) \mu_{\alpha j}(z) \mu_{k\alpha}(w),
    \end{aligned}
\end{equation}
where $C_{ijkl}(x) = \sum_{r \in \ZZ} C_{ijkl}[r]x^{r}$ and $\delta\left(\frac{w}{z}\right)=\sum_{r\in \mathbb{Z}}\left( \frac{w}{z}\right)^r$ and we extend this bracket to the entire algebra $\Vbf_q(L\Mat_N)$ by the Leibniz rule:
\begin{equation*}
    \{f_1(z), f_2(w)f_3(w)\} = \{f_1(z), f_2(w)\}f_3(w) + \{f_1(z), f_3(w)\}f_2(w).
\end{equation*}
Note that the property
\begin{equation*}
    \{D_zf(z), g(w)\} = D_z\{f(z), g(w)\}, \quad \{f(z), D_wg(w)\} = D_w\{f(z), g(w)\}
\end{equation*}
follows immediately from the definition of the Poisson bracket on Fourier modes.
In the rest of this subsection, we show that the bracket \eqref{eq:PoissonBracketDefinition} induces a well-defined Poisson structure on $\Vbf_q(L\Mat_N)$ and $\Vbf_q(LSL_N)$.

Note that this Poisson bracket is slightly different from that of \cites{FRS98, SS98}; the difference merely comes from the convention for $\nfr_+$ and $\nfr_-$, namely, which one is taken to be upper triangular or lower triangular. See Example~\ref{ex:N=2} and Example~\ref{ex:N=3} below.

\begin{ex}\label{ex:N=2}
    The Poisson bracket on $\Vbf_q(L\Mat_2)$ is defined as follows: 
    \begin{equation*}
        \begin{aligned}
            \{\mu_{11}(z), \mu_{11}(w)\}&= \sum_{r\in\ZZ}\frac{1-q^r}{1+q^r}\left(\frac{w}{z}\right)^r \mu_{11}(z)\mu_{11}(w)-\delta\left(\frac{qw}{z}\right)\mu_{12}(z)\mu_{21}(w)+\delta\left(\frac{w}{qz}\right)\mu_{21}(z)\mu_{12}(w),\\
            \{\mu_{11}(z), \mu_{12}(w)\}&=-\delta\left(\frac{qw}{z}\right) \mu_{12}(z)\mu_{22}(w),\\
            \{\mu_{11}(z), \mu_{21}(w)\}&=\delta\left(\frac{w}{qz}\right) \mu_{21}(z)\mu_{22}(w),\\
            \{\mu_{11}(z), \mu_{22}(w)\}&=-\sum_{r \in \ZZ} \frac{1-q^r}{1+q^r}\left(\frac{w}{z}\right)^r \mu_{11}(z)\mu_{22}(w),\\
            \{\mu_{12}(z), \mu_{12}(w)\}&=0,\\
            \{\mu_{12}(z), \mu_{21}(w)\}&=-\delta\left(\frac{w}{z}\right)\mu_{11}(z)\mu_{22}(w)+\delta\left(\frac{w}{qz}\right)\mu_{22}(z)\mu_{22}(w),\\
            \{\mu_{12}(z), \mu_{22}(w)\}&= -\delta\left(\frac{w}{z}\right)\mu_{12}(z)\mu_{22}(w),\\
            \{\mu_{21}(z), \mu_{21}(w)\}&= 0,\\
            \{\mu_{21}(z), \mu_{22}(w)\}&= \delta\left(\frac{w}{z}\right) \mu_{21}(z)\mu_{22}(w),\\
            \{\mu_{22}(z), \mu_{22}(w)\}&= \sum_{r\in\ZZ} \frac{1-q^r}{1+q^r} \left(\frac{w}{z}\right)^r \mu_{22}(z)\mu_{22}(w).
        \end{aligned}
    \end{equation*}
\end{ex}
    \begin{ex}\label{ex:N=3}
        We also provide a few non-trivial Poisson brackets for the case $\Vbf_q(L\Mat_3)$:
    \begin{equation*}
        \begin{aligned}
            &\{\mu_{12}(z), \mu_{23}(w)\}= -\sum_{r\in\ZZ} \frac{1+2q^r}{1+q^r+q^{2r}}\left(\frac{w}{z}\right)^r \mu_{12}(z)\mu_{23}(w) - \delta\left(\frac{qw}{z}\right) \mu_{13}(z)\mu_{33}(w),\\
            &\{\mu_{11}(z), \mu_{31}(w)\}\\
            &\quad = \sum_{r\in\ZZ} \frac{q^r-q^{2r}}{1+q^r+q^{2r}}\left(\frac{w}{z}\right)^r \mu_{11}(z)\mu_{31}(w) + \delta\left(\frac{w}{qz}\right)\mu_{21}(z)\mu_{32}(w) + \delta\left(\frac{w}{qz}\right) \mu_{31}(z)\mu_{33}(w),\\
            &\{\mu_{22}(z), \mu_{22}(w)\}\\
            &\quad = \sum_{r \in \ZZ} \frac{1-q^{2r}}{1+q^r+q^{2r}} \left(\frac{w}{z}\right)^r \mu_{22}(z)\mu_{22}(w) - \delta\left(\frac{qw}{z}\right) \mu_{23}(z)\mu_{32}(w) +\delta\left(\frac{w}{qz}\right) \mu_{32}(z)\mu_{23}(w).
        \end{aligned}
    \end{equation*}
    \end{ex}

\begin{rem}\label{rem:indep_d}
    Observe that for any $d \in \ZZ$, the choice
    \begin{equation}
        \mu_{ij}^{(d)}(z) = \sum_{m \in \ZZ} \mu_{ij}[m]z^{-m-d}
    \end{equation}
    does not affect the Poisson structure. For example, in $\Vbf_q(L\Mat_2)$,
    \begin{equation*}
    \begin{aligned}
        \{\mu_{11}^{(d)}(z), \mu_{22}^{(d)}(w)\} &= \sum_{m, n \in \ZZ}\{\mu_{11}[m], \mu_{22}[n]\}z^{-m+d}w^{-n-d} \\
        &= -\sum_{m, n, r \in \ZZ} \frac{1-q^r}{1+q^r}\left(\frac{w}{z}\right)^r\mu_{11}[m-r]\mu_{22}[n+r]z^{-m+r+d}w^{-n-r+d}\\
        &= -\sum_{r\in\ZZ} \frac{1-q^r}{1+q^r}\left(\frac{w}{z}\right)^r \mu_{11}^{(d)}(z)\mu_{22}^{(d)}(w),
    \end{aligned}
    \end{equation*}
    which coincides with the computation in Example~\ref{ex:N=2}. 
\end{rem}

To see the well-definedness of the Poisson bracket on $\Vbf_q(L\Mat_N)$, we need to show the following properties: 
\begin{enumerate}
    \item (Skew-symmetry)
    $\{f_1(z), f_2(w)\} = -\{f_2(w), f_1(z)\}$;
    \item (Jacobi identity)
    $\{f_1(z), \{f_2(w), f_3(u)\}\} - \{f_2(w), \{f_1(z), f_3(u)\}\} = \{\{f_1(z), f_2(w)\}, f_3(u)\},$
\end{enumerate}
for any $f_1(z), f_2(z), f_3(z) \in \Vbf_q(L\Mat_N)$.
We will show in Proposition~\ref{prop:Poisson} that our bracket \eqref{eq:PoissonBracketDefinition} exactly matches the Poisson bracket \eqref{eq:PoissonBracketDef} evaluated on $\mu_{ij}[m]$ and $\mu_{kl}[n]$, which is derived from the Poisson-Lie group structure of $LSL_N$. In addition, we prove in Proposition~\ref{prop:Jacobi} the bracket \eqref{eq:PoissonBracketDef} satisfies skew-symmetry and the Jacobi identity. We also note that by the Leibniz rule and compatibility with the $q$-difference operator, skew-symmetry and the Jacobi identity for the brackets between generators imply those in the full algebra.

To compare the brackets \eqref{eq:PoissonBracketDefinition} and  \eqref{eq:PoissonBracketDef}, we first need to establish the following setup. Let $L\gfr = L\sfr\lfr_N = \sfr\lfr_N \otimes \CC[t, t^{-1}]$ be the loop Lie algebra endowed with the nondegenerate symmetric invariant bilinear form
\begin{equation*}
    \langle X[m], Y[n] \rangle := \tr(XY) \delta_{m+n, 0}
\end{equation*}
where $X[m] := X \otimes t^{m}$, $Y[n] := Y \otimes t^{n}$. We decompose $L\gfr$ as
\begin{equation*}
L\gfr = L\nfr_+ \oplus L\hfr \oplus L\nfr_-,
\end{equation*}
where $L\nfr_+$ and $L\nfr_-$ denote the subalgebras of strictly upper and strictly lower triangular matrices, respectively, and $L\hfr$ denotes the subalgebra of diagonal matrices. Let $e_{ij}$ denote the matrix with $1$ in the $(i,j)$-entry and $0$ otherwise. We define a linear map $\theta : L\hfr \to L\hfr$ by
\begin{equation*}
\theta\left(\sum_{i=1}^Nc_ie_{ii}[n]\right) = q^n\left(\sum_{i=1}^{N-1} c_{i+1}e_{ii}[n] + c_1e_{NN}[n]\right), \quad \sum_{i=1}^N c_i = 0.
\end{equation*}
It is straightforward to verify that $\theta$ satisfies the following properties for any $X, Y \in L\hfr$:
    \begin{equation*}
    \begin{aligned}
        &\det(I-\theta) \ne 0;\\
        &\langle\theta(X), \theta(Y) \rangle = \langle X, Y\rangle.
    \end{aligned}
    \end{equation*}
By Lemma~\ref{lemma:AppendixA1} and Lemma~\ref{lemma:AppendixA2}, the operators $R_+, R_-$ satisfy the classical Yang-Baxter equation, and $R$ satisfies the modified classical Yang-Baxter equation, where
\begin{equation*}
    R = P_+ - P_- + \frac{I+\theta}{I-\theta}P_0, \quad R_+ = P_+ + \frac{I}{I-\theta}P_0, \quad R_- = -P_- + \frac{\theta}{I-\theta}P_0.
\end{equation*}
Here, $P_+, P_0, P_-$ are the projections from $L\gfr$ onto $L\nfr_+, L\hfr, L\nfr_-$, respectively. 

For later use, we summarize the following explicit computations. Let $X = \sum_{i=1}^N c_ie_{ii}[n] \in L\hfr$. Since the images of $X$ under the operators are also diagonal matrices, it suffices to describe their $i$-th diagonal entries:
    \begin{equation*}
        \left(\frac{I}{I-\theta}X\right)_{ii} = \frac{1}{1-q^{Nn}}\sum_{j=1}^N c_j q^{n|j-i|_N}t^n,
    \end{equation*}
    \begin{equation*}
        \left(\frac{\theta}{I-\theta}X\right)_{ii} = \frac{1}{1-q^{Nn}} \sum_{j=1}^N c_j\left( q^{n|j-i|_N} - \delta_{ij}(1-q^{Nn})\right) t^n,
    \end{equation*}
    and
    \begin{equation*}
        \left(\frac{I+\theta}{I-\theta}X\right)_{ii} = \frac{1}{1-q^{Nn}} \sum_{j=1}^N c_j\left( 2q^{n|j-i|_N} - \delta_{ij}(1-q^{Nn}) \right)t^n.
    \end{equation*}
    In addition, the explicit formulas for the left and right gradients of the functionals $\mu_{ij}[m]$, corresponding to the matrix expressions in \eqref{eq:nabla}, are given below. We regard these gradients $L\gfr$-valued formal series whose coefficients are Fourier modes of elements of $\Vbf_q(L\Mat_N)$:
\begin{equation*}
\begin{aligned}
    \nabla\mu_{ij}[m] &:= \sum_{\alpha=1}^N \sum_{r \in \ZZ} \left(-\frac{1}{N}\right) \mu_{ij}[m-r] \otimes e_{\alpha \alpha}[r] + \sum_{\alpha=1}^N \sum_{r \in \ZZ}\mu_{\alpha j}[m-r] \otimes e_{\alpha i}[r],\\
    \nabla'\mu_{ij}[m] &:= \sum_{\alpha=1}^N \sum_{r \in \ZZ} \left(-\frac{1}{N}\right) \mu_{ij}[m-r] \otimes e_{\alpha \alpha}[r] + \sum_{\alpha=1}^N \sum_{r \in \ZZ}\mu_{i\alpha}[m-r] \otimes e_{j \alpha}[r].
\end{aligned}
\end{equation*}
We naturally extend the linear operators $R, R_\pm$ and the bilinear form $\langle \cdot, \cdot \rangle$ on $L\gfr$ to those on the $L\gfr$-valued formal series. Let $\tau : L\gfr \to L\gfr$ be the Lie algebra automorphism given by $\tau(X[n]) = q^nX[n]$ for $X \in \sfr\lfr_N$. With this setup, the Poisson bracket originally defined in \eqref{eq:PoissonBracketDef} takes the following form:
\begin{equation} \label{eq:Poisson_R}
\begin{aligned}
    \{\mu_{ij}[m], \mu_{kl}[n] \} &= \frac{1}{2}\langle R\nabla \mu_{ij}[m], \nabla\mu_{kl}[n] \rangle +\frac{1}{2}\langle R\nabla' \mu_{ij}[m], \nabla' \mu_{kl}[n]\rangle\\
    &\quad -\langle (\tau \circ R_+)\nabla' \mu_{ij}[m], \nabla \mu_{kl}[n]\rangle-\langle (R_- \circ \tau^{-1})\nabla \mu_{ij}[m], \nabla' \mu_{kl}[n]\rangle.
    \end{aligned}
\end{equation}
As we mentioned, the skew-symmetry and Jacobi identity of the bracket \eqref{eq:Poisson_R} were established in Proposition~\ref{prop:Jacobi}. Since the following proposition shows that \eqref{eq:PoissonBracketDefinition} coincides with \eqref{eq:Poisson_R}, it follows that \eqref{eq:PoissonBracketDefinition} also satisfies skew-symmetry and the Jacobi identity.

\begin{prop}\label{prop:Poisson}
    Let $R, R_+, R_-$ and $\tau$ be defined as above. The Poisson bracket of $\mu_{ij}[m]$ and $\mu_{kl}[n]$ defined in \eqref{eq:Poisson_R} coincides with that given in \eqref{eq:PoissonBracketDefinition}. Hence \eqref{eq:PoissonBracketDefinition} is a well-defined Poisson bracket. 
\end{prop}

\begin{proof}
    The non-zero contributions to the evaluation of the Poisson bracket arise from the pairings $\langle e_{\alpha \beta}[r], \allowbreak e_{\beta \alpha}[-r]\rangle$ for $\alpha \ne \beta$, and from $\langle e_{\alpha \alpha}[r], e_{\alpha \alpha}[-r]\rangle$, which arise from the gradients separately. The former case arising in the computation of $\{\mu_{ij}[m], \mu_{kl}[n]\}$ in \eqref{eq:Poisson_R} yields the following:
\begin{equation*}
    \begin{aligned}
        &\frac{1}{2} (\varepsilon_{ik}+\varepsilon_{lj})\sum_{r \in \ZZ} \mu_{kj}[m-r]\mu_{il}[n+r]\\
        & -\delta_{jk} \sum_{\alpha>j} \sum_{r\in \ZZ} q^{r} \mu_{i\alpha}[m-r]\mu_{\alpha l}[n+r] + \delta_{il} \sum_{\alpha > i} \sum_{r \in \ZZ} q^{-r} \mu_{\alpha j}[m-r]\mu_{k \alpha}[n+r].
    \end{aligned}
\end{equation*}
    For the latter case, let $X = \sum_{\alpha=1}^N \left( \delta_{i\alpha} - \frac{1}{N} \right) e_{\alpha \alpha}[r]$. Evaluating $X$ on the linear operators yields the following identities:
    \begin{equation*}
    \begin{aligned}
        \frac{I}{I-\theta} X &= \sum_{\alpha=1}^N\left(-\frac{1+q^r+\dots+q^{(N-1)r}}{N(1-q^{Nr})} + \Psi_{i\alpha}[r] \right) e_{\alpha \alpha}[r],\\
        \frac{\theta}{I-\theta} X &= \sum_{\alpha=1}^N\left(- \frac{q^r+\dots+q^{(N-1)r} + q^{Nr}}{N(1-q^{Nr})} + \Omega_{i\alpha}[r] \right)e_{\alpha\alpha}[r],\\
        \frac{I+\theta}{I-\theta}X &=\sum_{\alpha=1}^N\left(- \frac{1+2q^r+\dots+2q^{(N-1)r}+q^{Nr}}{N(1-q^{Nr})} + \Phi_{i\alpha}[r] \right) e_{\alpha\alpha}[r].
    \end{aligned}
    \end{equation*}
Taking the bilinear form and summing over $r \in \ZZ$, we conclude that the latter case gives rise to the term
\begin{equation*}
    \sum_{r \in \ZZ}C_{ijkl}[r]\mu_{ij}[m-r] \mu_{kl}[n+r].
\end{equation*}
Hence \eqref{eq:PoissonBracketDefinition} coincides with \eqref{eq:PoissonBracketDef}.
\end{proof}

Finally, we show the following lemma, which shows that the Poisson bracket $\Vbf_q(L\Mat_N)$ induces a well-defined Poisson bracket on $\Vbf_q(LSL_N).$

\begin{lem}\label{lem:determinant}
For $1\le i,j\le N$, we have  
    $\{\mu_{ij}(z), \det(w)\} = 0$.
\end{lem}

\begin{proof}
Expanding the Poisson bracket using the Leibniz rule, for each $\sigma \in \Sfr_N$ we obtain
    \begin{align}
        &\{\mu_{ij}(z), \mu_{1\sigma(1)}(w)\mu_{2\sigma(2)}(w)\cdots\mu_{N\sigma(N)}(w)\}\\
        &=\sum_{x=1}^NC_{ijx\sigma(x)}\left(\frac{w}{z}\right) \mu_{ij}(z)\mu_{1\sigma(1)}(w) \cdots \mu_{N\sigma(N)}(w) \label{eq:determinant1}\\
        &\quad +\frac{1}{2}\sum_{x=1}^N\varepsilon_{ix}\delta\left(\frac{w}{z}\right) \mu_{xj}(z)\mu_{i\sigma(x)}(w)\mu_{1\sigma(1)}(w)\cdots \mu_{i\sigma(i)}(w) \cdots \widehat{\,\mu_{x\sigma(x)}\,}(w)\cdots \mu_{N\sigma(N)}(w)\label{eq:determinant2}\\
        &\quad +\frac{1}{2}\sum_{x=1}^N \varepsilon_{\sigma(x)j} \delta\left(\frac{w}{z}\right) \mu_{xj}(z)\mu_{i\sigma(x)}(w)\mu_{1\sigma(1)}(w)\cdots\mu_{\sigma^{-1}(j)j}(w) \cdots\widehat{\,\mu_{x\sigma(x)}\,}(w)\cdots \mu_{N\sigma(N)}(w)\label{eq:determinant3}\\
        &\quad -\sum_{\alpha > j} \delta\left(\frac{qw}{z}\right)\mu_{i\alpha}(z)\mu_{\alpha \sigma(j)}(w)\mu_{1\sigma(1)}(w) \cdots \mu_{\alpha \sigma(\alpha)}(w) \cdots \widehat{\,\mu_{j\sigma(j)}\,}(w) \cdots \mu_{N\sigma(N)}(w)\label{eq:determinant4} \\
        &\quad +\sum_{\alpha>i} \delta\left(\frac{w}{qz}\right) \mu_{\alpha j}(z) \mu_{\sigma^{-1}(i)\alpha}(w) \mu_{\sigma^{-1}(1)1}(w) \cdots \mu_{\sigma^{-1}(\alpha) \alpha}(w) \cdots \widehat{\,\mu_{\sigma^{-1}(i)i}\,}(w) \cdots \mu_{\sigma^{-1}(N)N}(w)\label{eq:determinant5}.
    \end{align}
    Here, $\widehat{\,\mu_{ij}\,}(w)$ denotes the omission. The term \eqref{eq:determinant1} vanishes because $\sigma$ is a permutation of $\{1, 2, \dots, N\}$, which implies that summing over $\sigma(x)$ is equivalent to summing over $x$. Thus,
    \begin{equation*}
    \begin{aligned}
            \sum_{x=1}^N C_{ijx\sigma(x)}[r] &=\frac{1}{2}\sum_{x=1}^N \Phi_{ix}[r] + \frac{1}{2}\sum_{x=1}^N \Phi_{jx}[r] -q^{r}\sum_{x=1}^N \Psi_{jx}[r] - q^{-r}\sum_{x=1}^N \Omega_{ix}[r] \\
            &\notag= \frac{1+2q^r+\dots+2q^{(N-1)r}+q^{Nr}}{2(1-q^{Nr})}+\frac{1+2q^r+\dots+2q^{(N-1)r}+q^{Nr}}{2(1-q^{Nr})}\\
            &\notag \ -\frac{q^r+\dots+q^{(N-1)r}+q^{Nr}}{1-q^{Nr}}-\frac{1+q^r+\dots+q^{(N-1)r}}{1-q^{Nr}} = 0.
    \end{aligned}
    \end{equation*}
    We now show that the remaining terms cancel pairwise after summing over $\sigma \in \Sfr_N$.
    
    For the terms of type \eqref{eq:determinant2}, fix $x \neq i$ and set $\tau = \sigma \circ (ix)$ where $(ix)$ is the transposition exchanging $i$ and $x$. The term \eqref{eq:determinant2} corresponding to $\tau$ is
    \begin{equation*}
        \frac{1}{2}\sum_{y=1}^N\varepsilon_{iy}\delta\left(\frac{w}{z}\right) \mu_{yj}(z)\mu_{i\tau(y)}(w)\mu_{1\tau(1)}(w)\cdots \mu_{i\tau(i)}(w) \cdots \widehat{\,\mu_{y\tau(y)}\,}(w)\cdots \mu_{N\tau(N)}(w)
    \end{equation*}
    and taking the term with $y = x$, we obtain
    \begin{equation}\label{eq:determinant22}
        \frac{1}{2}\varepsilon_{ix}\delta\left(\frac{w}{z}\right) \mu_{xj}(z)\mu_{i\sigma(i)}(w)\mu_{1\sigma(1)}(w)\cdots \mu_{i\sigma(x)}(w) \cdots \widehat{\,\mu_{x\sigma(i)}\,}(w)\cdots \mu_{N\sigma(N)}(w).
    \end{equation}
    Since $\sigma$ and $\tau$ have opposite sign, \eqref{eq:determinant2} and \eqref{eq:determinant22} cancel. The same argument applies to the terms \eqref{eq:determinant3}, using $\tau = \sigma \circ (\sigma^{-1}(j)x)$.
    
    To show that \eqref{eq:determinant4} cancels, fix $\alpha > j$ and consider $\tau = \sigma \circ (j \alpha)$. The permutations $\tau$ and $\sigma$ have opposite sign and the term \eqref{eq:determinant4} corresponding to $\tau$ is
    \begin{equation*}
    \begin{aligned}
        &-\sum_{\alpha > j} \delta\left(\frac{qw}{z}\right) \mu_{i\alpha}(z)\mu_{\alpha \tau(j)}(w) \mu_{1\tau(1)}(w) \cdots \mu_{\alpha \tau(\alpha)}(w) \cdots \widehat{\,\mu_{j\tau(j)}\,}(w) \cdots \mu_{N \tau(N)}(w)\\
        &=-\sum_{\alpha > j} \delta\left(\frac{qw}{z}\right) \mu_{i\alpha}(z)\mu_{\alpha \sigma(\alpha)}(w) \mu_{1\sigma(1)}(w) \cdots \mu_{\alpha \sigma(j)}(w) \cdots \widehat{\,\mu_{j\tau(j)}\,}(w) \cdots \mu_{N \sigma(N)}(w).
    \end{aligned}
    \end{equation*}
    Thus, this term cancels with the term \eqref{eq:determinant4}. The same argument applies to the terms \eqref{eq:determinant5}, using $\tau = (i\alpha)\circ \sigma$. Hence $\{\mu_{ij}(z), \det(w)\}=0.$
\end{proof}

The following theorem summarizes the result of this section.

\begin{thm}
    The deformed affine Poisson algebra $\Vbf_q(LSL_N)$ is a $q$-difference Poisson algebra with the Poisson bracket induced from the bracket \eqref{eq:PoissonBracketDefinition}.
\end{thm}

\section{Classical and Deformed \texorpdfstring{W}{W}-algebra} \label{sec:W-alg}
In this section, we provide an algebraic construction of the deformed $W$-algebra, whose geometric construction was studied in \cites{FRS98, SS98}.
We demonstrate that the $q$-difference and Poisson structures are naturally induced from the deformed affine Poisson algebra, and that the Miura map is well-defined.

\subsection{Classical \texorpdfstring{$W$}{W}-algebra} \label{subsec:classical W}
We first recall the definition of classical $W$-algebras associated with the principal nilpotent element of $\sfr\lfr_N$. For a detailed description, we refer the reader to \cites{DS84, MR15, DSKV13}.

Let $\sfr\lfr_N = \nfr_+ \oplus \hfr \oplus \nfr_-$ be the standard triangular decomposition, and set $\pfr_- = \hfr \oplus \nfr_-$. Define a surjective differential algebra homomorphism
\begin{equation*}
    \pi_{\pfr_-} : \Vbf(\sfr\lfr_N) \to \Vbf(\pfr_-) := \CC\big[\partial_z^n e_{ij}(z) \mid i \ge j, \ n \in \ZZ_{\ge 0}\big]\Big/\Big\langle \partial_z^n \sum_{i=1}^N e_{ii}(z) \Bigm| n \in \ZZ_{\ge 0} \Big\rangle
\end{equation*}
by
\begin{equation*}
    \pi_{\pfr_-} (e_{ij}(z)) = \begin{cases}
        e_{ij}(z) \quad &\text{if } i \ge j,\\
        -1 \quad &\text{if } i+1=j,\\
        0 \quad &\text {if } i+1<j.
    \end{cases}
\end{equation*}
The \emph{classical $W$-algebra} associated with $\sfr\lfr_N$ is defined by
\begin{equation*}
    \Wbf(\sfr\lfr_N) = \left\{ f \in \Vbf(\pfr_-) \middle| \pi_{\pfr_-} \{e_{ij}(z), f(w)\} = 0 \text{ for all } i < j \right\}.
\end{equation*}
It is a fundamental result that the differential and Poisson structures of $\Vbf(\sfr\lfr_N)$ induce a well-defined differential and Poisson structure on $\Wbf(\sfr\lfr_N)$. A generating set $\{w_2(z), \cdots, w_N(z)\}$ of $\Wbf(\sfr\lfr_N)$ can be obtained from the following column determinant formula
\begin{equation}\label{eq:cdetnondeformed}
\begin{aligned}
    \text{cdet} \, &  \begin{pmatrix}
        \partial + e_{11}(z) & -1 & 0 &  \dots & 0 \\
        e_{21}(z) & \partial + e_{22}(z) & -1 &  \dots & 0\\
        e_{31}(z) & e_{32}(z) & \partial+e_{33}(z) &  \dots & 0\\
        \vdots & \vdots & \vdots  & \ddots & \vdots \\
        e_{N-1, 1}(z) & e_{N-1, 2}(z) & e_{N-1, 3}(z)  & \dots & -1\\
        e_{N1}(z) & e_{N2}(z) & e_{N3}(z) & \dots  & \partial+e_{NN}(z)
    \end{pmatrix}\\[1em]
    & = \partial^N + w_2(z)\partial^{N-2} + \dots + w_{N-1}(z)\partial + w_{N}(z),
\end{aligned}
\end{equation}
where the matrix has entries in $\Vbf(\pfr_-) \otimes \CC[\partial]$.
In order to get the RHS of \eqref{eq:cdetnondeformed}, one needs to use the  commutation relations
$   \partial e_{ij}(z) - e_{ij}(z) \partial = \partial_ze_{ij}(z) $
 in the ring $\Vbf(\pfr_-) \otimes \CC[\partial]$. Notice that the coefficient of $\partial^{N-1}$ is $e_{11}(z)+\dots + e_{NN}(z)=0$ and the following theorem holds.
 
\begin{thm}\cite{MR15}
    As a differential algebra, $\Wbf(\sfr\lfr_N)$ is freely generated by
    $w_2(z), \dots, w_{N}(z)$.
\end{thm}
\noindent
Moreover, explicit formulas for the Poisson brackets of classical $W$-algebras are also well understood (see \cites{DSKV16b,LSS23}).

Another well-known property of $\Wbf(\sfr\lfr_N)$ is the free field realization, which means it can be embedded in the following algebra:
\begin{equation*}
\Vbf(\hfr) := \CC\big[\partial_z^n e_{ii}(z) \mid 1\le i \le N, \  n \in \ZZ_{\ge 0}\big]\Big/\Big\langle \partial_z^n \sum_{i=1}^N e_{ii}(z) \Bigm| n \in \ZZ_{\ge 0}\Big\rangle
\end{equation*}
equipped with the Poisson bracket
$\{e_{ii}(z), e_{jj}(w)\} := \langle e_{ii}, e_{jj} \rangle \partial_w \delta(z,w),$
where $\langle \cdot, \cdot \rangle$ is the invariant symmetric bilinear form on $\sfr\lfr_N$. The following theorem implies the free field realization of $\Wbf(\sfr\lfr_N).$

\begin{thm}
    The \emph{Miura transformation}
    \begin{equation}
    m : \Wbf(\sfr\lfr_N) \to \Vbf(\hfr)
    \end{equation}
    which is the composition of the inclusion $\Wbf(\sfr\lfr_N) \to \Vbf(\pfr_-)$ and the natural projection map $\Vbf(\pfr_-) \to \Vbf(\hfr)$ is injective.
\end{thm}

\noindent
The image $m(w_r(z))$ of the generator $w_r(z)$ can be computed by the following formula 
    \begin{equation}\label{eq:Miura_nondeformed}
    \begin{aligned}
        &\big(\partial + e_{11}(z)\big)\big(\partial+e_{22}(z)\big)\dots\big(\partial+e_{NN}(z)\big) \\
        &= \partial^N+m(w_2(z))\partial^{N-2}+\dots+m(w_{N-1}(z))\partial+m(w_N(z)),
    \end{aligned}
    \end{equation}
and by the injectivity of the Miura map, $m(w_2(z)), \cdots, m(w_N(z))$ freely generates $m(\Wbf(\sfr\lfr_N))\simeq \Wbf(\sfr\lfr_N)$.
\vskip 2mm 
We summarize the following three essential properties of the classical $W$-algebra:
\begin{enumerate}[(1)]
    \item As a differential algebra, $\Wbf(\sfr\lfr_N)$ is freely generated by $N-1$ elements, which can be obtained from the column determinant of a matrix with entries in $\Vbf(\pfr_-)\otimes\CC[\partial]$.
    \item The Poisson algebra structure on $\Vbf(\sfr\lfr_N)$ naturally induces a Poisson algebra structure on $\Wbf(\sfr\lfr_N)$.
    \item The Miura transformation $\Wbf(\sfr\lfr_N) \to \Vbf(\hfr)$ is an injective Poisson algebra homomorphism.
\end{enumerate}
In the following subsections, we describe analogous properties of the deformed $W$-algebra.

\subsection{Construction of the deformed \texorpdfstring{$W$}{W}-algebra} \label{subsec:deformed W}
In parallel with the non-deformed case, we first define the $q$-difference algebra $\Vbf_q(B_+)$ corresponding to the differential algebra $\Vbf(\pfr_-)$. Let $\Vbf_q(B_+)$  be the quotient of $\Vbf_q(LSL_N)$ by the $q$-difference algebra ideal generated by the elements
    \begin{equation*}
        \mu_{ij}(z) + \delta_{i, j+1} \qquad \text{for } i>j.
    \end{equation*}
Then $\Vbf_q(B_+)$ inherits a natural $q$-difference algebra structure, and there is a natural surjective $q$-difference algebra homomorphism
\begin{equation*}
    \pi_{B_+} : \Vbf_q(LSL_N) \to \Vbf_q(B_+)
\end{equation*}
given by
\begin{equation} \label{eq:pi}
    \pi_{B_+}\big(\mu_{ij}(z)\big) = \begin{cases}
        \mu_{ij}(z) \quad&\text{if } i \le j,\\
        -1 \quad&\text{if } i=j+1,\\
        0 \quad&\text{if } i>j+1.
    \end{cases}
\end{equation}
Geometrically, one can understand $B_+$ as a subset of $LSL_N$ consisting of elements of the form
\begin{equation*}
    A(t) = \begin{pmatrix}
        A_{11}(t) & A_{12}(t) & \dots & A_{1,N-1}(t) & A_{1N}(t) \\
        -1 & A_{22}(t) & \dots & A_{2,N-1}(t) & A_{2N}(t) \\
        0 & -1 & \dots & A_{3,N-1}(t) & A_{3N}(t) \\
        \vdots & \vdots & \ddots & \vdots \\
        0 & 0 & \dots & -1 & A_{NN}(t)
    \end{pmatrix}, \qquad A_{ij}(t) \in \CC[t, t^{-1}].
\end{equation*}

Next, we define the $q$-gauge action on $\Vbf_q(LSL_N)$. We first describe the $q$-gauge action on the Fourier modes $\mu_{kl}[m]$. Let $N_+$ be the subgroup of $LSL_N$ consisting of upper triangular matrices with all diagonal entries equal to $1$. The $q$-gauge action of $N_+$ on $\mu_{kl}[m]$ is defined by
\begin{equation*}
    (g(t)\cdot \mu_{kl}[m])(A(t))
    := \mu_{kl}[m]\left(g(qt)A(t)g(t)^{-1}\right),
    \qquad g(t)\in N_+ .
\end{equation*}
The corresponding infinitesimal action of the Lie algebra
$\nfr_+ = \operatorname{Lie}(N_+)$ on the generators is given by
\begin{equation*}
    e_{ij}[p]\cdot \mu_{kl}[m]
    = q^p\delta_{ik}\mu_{jl}[m+p]
      -\delta_{jl}\mu_{ki}[m+p],
    \qquad i<j .
\end{equation*}
These actions extend to the formal generating series as follows:
\begin{align}
    & g(t)\cdot \mu_{kl}(z)
    := \sum_{m\in\ZZ} (g(t)\cdot \mu_{kl}[m])z^{-m-1}, \label{eq:Lie group action}\\
    & e_{ij}[p]\cdot \mu_{kl}(z)
    := \sum_{m\in\ZZ}(e_{ij}[p]\cdot \mu_{kl}[m])z^{-m-1}= \delta_{ik}(qz)^p\mu_{jl}(z)
       -\delta_{jl}z^p\mu_{ki}(z). \label{eq:Lie algebra action}
\end{align}
They further extend naturally to the algebras $\Vbf_q(LSL_N)$ and $\Vbf_q(B_+)$ by the Leibniz rule. As one can see in \eqref{eq:Lie algebra action}, these actions do not, in general, preserve $\Vbf_q(LSL_N)$ or $\Vbf_q(B_+)$. However, since we are only interested in elements of $\Vbf_q(B_+)$ that are annihilated by this action, it is sufficient to define the action formally. Moreover, whenever a generating series of the form $z^m f(z)$ appears, where $f(z)\in\Vbf_q(LSL_N)$, we apply $\pi_{B_+}$ only to the generating series $f(z)$ and leave the factor $z^m$ unchanged.

\begin{defn}\label{def:W}
    The \emph{deformed $W$-algebra} associated with $LSL_N$ is defined by
    \begin{equation}\label{eq:deformedWdef}
        \Wbf_q(LSL_N) = \left\{f(z) \in \Vbf_q(B_+) \;\middle|\; \pi_{B_+} \big(X \cdot f(z) \big) = 0 \text{ for all } X \in \nfr_+\right\}.
    \end{equation}
\end{defn}
The deformed $W$-algebra is a $q$-difference Poisson algebra endowed with the Poisson bracket introduced in the following subsection.
In the rest of this section, we show the deformed $W$-algebra \eqref{eq:deformedWdef} is identical to the $q$-difference algebra defined via the $q$-gauge action. The space of the $q$-gauge invariant functionals is defined as follows:
\begin{equation} \label{eq:tilde W}
    \widetilde{\Wbf}_q(LSL_N) = \left\{f(z) \in \Vbf_q(B_+) \;\middle|\; \pi_{B_+} \big(g(t) \cdot f(z) \big) = f(z) \text{ for all } g(t) \in N_+ \right\}.
\end{equation}
In the following lemma and proposition, we prove that $\widetilde{\Wbf}_q(LSL_N)$ and $\Wbf_q(LSL_N)$ are the same $q$-difference algebras.

\begin{lem}\label{lem:gen3}
    For $i<j$ and $f \in \Vbf_q(LSL_N)$, we have 
    \begin{equation}\label{eq:exponentialformula}
        \big((I+ e_{ij}[p])\cdot f(z)\big) = \sum_{m \ge 0} \frac{1}{m!} \big((e_{ij}[p])^m \cdot f(z)\big).
    \end{equation}
    Here, the left-hand side denotes the action of $N_+$, while the right-hand side denotes the action of $\nfr_+$, and the summation on the right-hand side is finite.
\end{lem}
\begin{proof}
    We first verify \eqref{eq:exponentialformula} for $f(z) = \mu_{kl}(z)$.
    \begin{equation}\label{eq:muaction}
    \begin{aligned}
        \big((I+e_{ij}[p])\cdot \mu_{kl}(z)\big) &= \mu_{kl}(z) + (qz)^p\delta_{ik} \mu_{jl}(z)-z^p\delta_{jl}\mu_{ki}(z) - q^pz^{2p}\delta_{ik}\delta_{jl} \mu_{lk}(z)\\
        &= \mu_{kl}(z)+ e_{ij}[p]\cdot \mu_{kl}(z) + \frac{1}{2} e_{ij}[p] \cdot (e_{ij}[p] \cdot \mu_{kl}(z))\\
        & = \sum_{m\geq 0}\frac{1}{m!}(e_{ij}[p])^m\cdot \mu_{kl}(z).
    \end{aligned}
    \end{equation}
    Notice that in the right-hand side of \eqref{eq:muaction}, we used  $(e_{ij}[p])^3\cdot\mu_{kl}(z)=0$ since $i<j$. To extend the result to the entire algebra, suppose that
    \begin{equation*}
        \begin{aligned}
            (I+e_{ij}[p]) \cdot f(z) &= \sum_{m_1 \ge 0} \frac{1}{m_1!}(e_{ij}[p])^{m_1} \cdot f(z), \\
            (I+e_{ij}[p]) \cdot g(z) &= \sum_{m_2 \ge 0} \frac{1}{m_2!}(e_{ij}[p])^{m_2} \cdot g(z)
        \end{aligned}
    \end{equation*}
    for $f(z), g(z) \in \Vbf_q(LSL_N)$, where right-hand sides of both equalities are finite sums. Taking the product of these series and collecting terms with the total degree $M=m_1+m_2$, we obtain:
    \begin{equation*}
        \begin{aligned}
            (I + e_{ij}[p]) \cdot (f(z)g(z)) &= \left( \sum_{m_1 \ge 0} \frac{1}{m_1!}(e_{ij}[p])^{m_1} \cdot f(z) \right)\left(\sum_{m_2 \ge 0} \frac{1}{m_2!}(e_{ij}[p])^{m_2} \cdot g(z)\right) \\
            &= \sum_{m_1 \ge 0} \sum_{m_2 \ge 0} \frac{1}{(m_1+m_2)!} \binom{m_1+m_2}{m_2} ((e_{ij}[p])^{m_1} \cdot f(z))((e_{ij}[p])^{m_2} \cdot g(z))\\
            &=\sum_{M \ge 0} \frac{1}{M!} \sum_{s=0}^M \binom{M}{s} (e_{ij}[p]^{M-s} \cdot f(z))(e_{ij}[p]^{s} \cdot g(z)) \\
            &= \sum_{M \ge 0} \frac{1}{M!} (e_{ij}[p])^M \cdot (f(z)g(z)),
        \end{aligned}
    \end{equation*}
    where the last equality follows from the Leibniz rule for the derivation. Moreover, since the sums for $(I+e_{ij}[p]) \cdot f(z)$ and $(I+e_{ij}[p])\cdot g(z)$ are finite, this resulting summation is also finite.
\end{proof}

Finally, we prove the following proposition.
\begin{prop}\label{prop:gen2}
We have
    \begin{equation*}
        \Wbf_q(LSL_N) = \widetilde{\Wbf}_q(LSL_N).
    \end{equation*}
\end{prop}

\begin{proof}
    Take any $f(z) \in \widetilde{\Wbf}_q(LSL_N)$ and $X \in \nfr_+$. By the definition of the infinitesimal action,
    \begin{equation*}
         X \cdot f(z) = \left.\frac{d}{ds}\right|_{s=0} \big( e^{sX} \cdot f(z)\big) = \left.\frac{d}{ds}\right|_{s=0} f(z) = 0,
    \end{equation*}
    because $f(z)$ is invariant under the action of $N_+$. Hence, $f(z) \in \Wbf_q(LSL_N)$. Conversely, let $f(z) \in \Wbf_q(LSL_N)$ and $g(t) \in N_+$. For $1 \le i < j \le N$ and $a(t) \in \CC[t, t^{-1}]$, let $E_{ij}(a(t)) \in N_+$ denote the elementary upper unipotent matrix having $a(t)$ as its only non-zero off-diagonal entry, located at $(i, j)$-position. Every element $g(t)\in N_+$ can be written as a finite ordered product
    \begin{equation*}
        g(t) = \overleftarrow{\prod_{j=2}^{N}}\left(\prod_{i=1}^{j-1} E_{ij}(g_{ij}(t))\right),
    \end{equation*}
    where $g_{ij}(t)$ denotes the $(i, j)$-entry of $g(t)$. Therefore, it suffices to show that
    \begin{equation*}
        E_{ij}(a(t)) \cdot f(z) = f(z)
    \end{equation*}
    for all $i < j$. 
   By Lemma~\ref{lem:gen3}, the action of $E_{ij}(a(t))$ on $f(z)$ can be written as a finite sum of infinitesimal Lie algebra actions on $f(z)$. Since $f(z)\in \Wbf_q(LSL_N)$, every nontrivial Lie algebra action term vanishes after applying the projection $\pi_{B_+}$, leaving only the contribution from the identity element. This proves the proposition.
\end{proof}

\subsection{Generators of deformed W-algebras} \label{subsec:gen of deformed W}
In this subsection, we find explicit generators of $\Wbf_q(LSL_N)$ as a $q$-difference algebra. As in the classical $W$-algebra case, these generators are obtained from the column determinant of a matrix whose entries lie in $\Vbf_q(B_+) \otimes \CC[D]$.

\begin{defn}
    Define the formal generating series
    \begin{equation}\label{eq:generatorFourier}
        T_r(z), \quad 1 \le r \le N,
    \end{equation}
    by the identity
    \begin{equation*}
    \begin{aligned}
        \operatorname{cdet}&\begin{pmatrix}
            D+\mu_{11}(z) & \mu_{12}(z) & \dots & \mu_{1, N-1}(z) & \mu_{1N}(z)\\
            -1 & D+\mu_{22}(z) & \dots & \mu_{2, N-1}(z) & \mu_{2N}(z) \\
            0 & -1 & \dots & \mu_{3, N-1}(z) & \mu_{3N}(z) \\
            \vdots & \vdots & \ddots & \vdots & \vdots \\
            0 & 0 & \dots & -1 & D+ \mu_{NN}(z)
        \end{pmatrix}\\
        &=D^{N} + T_1(z)D^{N-1} + T_2(z) D^{N-2}+ \dots + T_{N-1}(z)D + T_N(z),
        \end{aligned}
    \end{equation*}
    where the $D$ satisfies the commutation relation $D\mu_{ij}(z) = \mu_{ij}(qz)D$. Note that $T_N(z) =1$ since ${\det}(z) =1$ in $\Vbf_q(B_+)$. 
\end{defn}
For example, if $N=3$, the formal generating series are 
    \begin{equation*}
        \begin{aligned}
        T_1(z) &= \mu_{11}(z) + \mu_{22}(qz)+\mu_{33}(q^2z),\\
        T_2(z) &= \mu_{12}(z)+\mu_{23}(qz) + \mu_{11}(z)\mu_{22}(z) +\mu_{11}(z)\mu_{33}(qz) +\mu_{22}(qz)\mu_{33}(qz),\\
            T_3(z) & = \mu_{13}(z)+ \mu_{12}(z)\mu_{33}(z) + \mu_{11}(z) \mu_{23}(z)+\mu_{11}(z)\mu_{22}(z)\mu_{33}(z)=1.
        \end{aligned}
    \end{equation*}
In the following proposition, we obtain the explicit formula for $T_r(z)$ by direct computation of the column determinant. To describe the formula, let us introduce the index set
\begin{equation}\label{eq:index}
    \Jcal_{m, r}^N:= \left\{(\alpha, \beta) \in \ZZ^m \times \ZZ^m \;\middle|\;\begin{aligned} 1 \le \alpha_1\le\beta_1 &<\alpha_2 \le \beta_2 < \dots <\alpha_m \le \beta_m \le N\\
    &\sum_{j=1}^m (\beta_j -\alpha_j+1) = r\end{aligned}\right\},
\end{equation}
for  $1 \le r \le N, 1 \le m \le r$, where $\alpha=(\alpha_1, \dots, \alpha_m), \beta=(\beta_1, \dots, \beta_m)$.

\begin{prop}\label{prop:generators}
For $1 \le r \le N$,
\begin{equation}\label{eq:generators}
        T_r(z) = \sum_{m=1}^r \sum_{(\alpha, \beta) \in \Jcal_{m, r}^N}\prod_{x=1}^m \mu_{\alpha_x \beta_x}(q^{\gamma_x}z),
    \end{equation}
    with the exponent given by $\gamma_x := \alpha_x-\beta_{x-1}+\alpha_{x-1}-\beta_{x-2}+\dots-\beta_1+\alpha_1-x$.
\end{prop}
\begin{proof}
    We proceed by induction on $N$. The base case $N=2$ is straightforward.
    Assume that the formula holds for matrices of size up to $N$ and let us prove it for $N+1$. Let $T^{m}_r(z)$ denote the coefficient of $D^{m-r}$ of the column determinant of top-left $m \times m$ sub-matrix of the full $(N+1) \times (N+1)$ matrix. 
    By expanding the column determinant of the $(N+1) \times (N+1)$ matrix, we obtain the recursive relation:
    \begin{equation*}
        \begin{aligned}
        &D^{N+1}+T^{N+1}_1(z)D^{N}+\dots + T^{N+1}_N(z)D + T^{N+1}_{N+1}(z)\\
            &=\big(D^{N}+T_1^N(z)D^{N-1}+\dots+T_N^N(z)\big)(D+\mu_{N+1, N+1}(z))\\
            &\quad + \big(D^{N-1}+T_1^{N-1}(z)D^{N-2}+\dots+T_{N-1}^{N-1}(z)\big)\mu_{N, N+1}(z)\\
            &\quad + \dots\\
            &\quad + (D+T_1^1(z))\mu_{2, N+1}(z)\\
            &\quad + \mu_{1, N+1}(z),
        \end{aligned}
    \end{equation*}
    where $T_1^1(z) := \mu_{11}(z)$. Extracting the coefficients of $D^{N+1-r}$ on both sides of this expansion, we obtain
    \begin{equation}\label{eq:l}
        \begin{aligned}
            T^{N+1}_r(z) &= T_r^N(z) + T_{r-1}^N(z)\mu_{N+1, N+1}(q^{N+1-r}z)\\
            &\quad +T_{r-2}^{N-1}(z)\mu_{N, N+1}(q^{N+1-r}z)\\
            &\quad +\dots \\
            &\quad + T_1^{N-r+2}(z) \mu_{N-r+3, N+1}(q^{N+1-r}z)\\
            &\quad + \mu_{N-r+2, N+1}(q^{N+1-r}z).
        \end{aligned}
    \end{equation}
    We simplify \eqref{eq:l} by classifying the terms on the right-hand side according to their polynomial degree in the variables $\mu_{ij}(z)$. The degree $1$ terms arise exclusively from $T_r^N(z)$ and the final term, yielding:
    \begin{equation*}
    \begin{aligned}
        &\sum_{(\alpha, \beta) \in \Jcal_{1, r}^N}\mu_{\alpha_1\beta_1}(q^{\alpha_1-1}z) +\mu_{N-r+2, N+1}(q^{N+1-r}z)=\sum_{(\alpha, \beta) \in \Jcal_{1, r}^{N+1}}\mu_{\alpha_1\beta_1}(q^{\alpha_1-1}z).
        \end{aligned}
    \end{equation*}
    For a general degree $r-k$ where $0 \le k \le r-2$, the corresponding terms emerge from the first $k+1$ lines of \eqref{eq:l}:
    \begin{equation*}
    \begin{aligned}
        &\sum_{(\alpha, \beta) \in \Jcal_{r-k, r}^N} \prod_{x=1}^{r-k}\mu_{\alpha_x\beta_x}(q^{\gamma_x}z) +\sum_{l=0}^{k}\sum_{(\alpha, \beta) \in \Jcal_{r-k-1, r-1-l}^N} \prod_{x=1}^{r-k-1} \big(\mu_{\alpha_x\beta_x}(q^{\gamma_x}z)\big)\big(\mu_{N+1-l, N+1}(q^{N-r+1}z)\big).
    \end{aligned}
    \end{equation*}
    Observe that the index set $\Jcal_{r-k, r}^{N+1}$ admits the natural disjoint decomposition $\Jcal_{r-k, r}^{N+1} = \Jcal_{r-k, r}^{N} \cup \big( \bigcup_{l=0}^{k} \Jcal_{r-k, r}^{N+1}[l]\big),$
    where $\Jcal_{r-k, r}^{N+1}[l]$ denotes the subset of $\Jcal_{r-k, r}^{N+1}$ whose final indices $(\alpha_{r-k}, \beta_{r-k})$ are fixed as $(N+1-l, N+1)$. Therefore, the degree $r-k$ terms of $T_r^{N+1}(z)$ are
    \begin{equation*}
        \begin{aligned}
            &\sum_{(\alpha, \beta) \in \Jcal_{r-k, r}^N} \prod_{x=1}^{r-k}\mu_{\alpha_x\beta_x}(q^{\gamma_x}z) +  \sum_{l=0}^k \sum_{(\alpha, \beta) \in \Jcal_{r-k, r}^{N+1}[l]} \prod_{x=1}^{r-k} \mu_{\alpha_x\beta_x}(q^{\gamma_x}z)=\sum_{(\alpha, \beta) \in \Jcal_{r-k, r}^{N+1}} \prod_{x=1}^{r-k} \mu_{\alpha_x\beta_x}(q^{\gamma_x}z).
        \end{aligned}
    \end{equation*}
    Finally, we can conclude the proposition for the size $N+1$ case.
\end{proof}

Now, let us show that the formal generating series $T_r(z)$ are contained in $\Wbf_q(LSL_N)$ in the following lemma.

\begin{lem}\label{lem:gen1}
    For $r = 1, 2, \dots, N-1$, we have 
    $T_{r}(z) \in \Wbf_q(LSL_N).$
\end{lem}
\begin{proof}
    It is straightforward to verify that
    \begin{equation*}
        e_{i_1j_1}[p_1] \cdot \big( e_{i_2j_2}[p_2] \cdot \mu_{kl}(z) \big) - e_{i_2j_2}[p_2] \cdot \big( e_{i_1j_1}[p_1] \cdot \mu_{kl}(z) \big) = [e_{i_1j_1}[p_1], e_{i_2j_2}[p_2]] \cdot \mu_{kl}(z).
    \end{equation*}
    Furthermore, for $i<j$, we have
   $e_{ij}[p] = [e_{i, i+1}[0], [e_{i+1, i+2}[0], \dots, [e_{j-2, j-1}[0], e_{j-1, j}[p]]\dots ]],$
    and thus it suffices to show that
    \begin{equation*}
        \pi_{B_+}\big(e_{i, i+1}[p] \cdot T_r(z)\big) = 0, \quad i, r \in \{1, \dots, N-1\}.
    \end{equation*}
    For the remainder of this proof, we drop the notation $\pi_{B_+}$ by directly setting $\mu_{ij}(z) = 0$ for $i>j+1$ and $\mu_{i+1, i}(z) = -1$ for the sake of brevity. Now, fix $r, i \in \{1, 2, \dots, N-1\}$ and $p \in \ZZ$. We will prove that $e_{i, i+1}[p] \cdot T_r(z) =0$ by showing that its homogeneous components of degree $0$, degree $r$, and degree $k$ for $1 \le k \le r-1$ with respect to the variables $\mu_{ab}(z)$ for $a \le b$ all vanish. First, the degree $0$ component is trivial unless $r=1$. In the case $r=1$, we have
    \begin{equation*}
        e_{i, i+1}[p] \cdot \big( \mu_{ii}(q^{i-1}z) + \mu_{i+1, i+1}(q^{i}z) \big) = -(q^{i}z)^p + (q^{i}z)^p = 0.
    \end{equation*}
    Moreover, from \eqref{eq:generators}, the degree $r$ component of $T_r(z)$ is given by
    \begin{equation*}
        \sum_{1 \le \alpha_1 < \dots < \alpha_r \le N} \big(\mu_{\alpha_1\alpha_1}(q^{\alpha_1-1}z)\big) \dots \big(\mu_{\alpha_r\alpha_r}(q^{\alpha_r-r}z)\big).
    \end{equation*}
    Notice that whenever $e_{i, i+1}[p]$ acts on a factor $\mu_{\alpha_x \alpha_x}(q^{\alpha_x-x}z)$, the degree strictly decreases. Consequently, the action yields no degree $r$ component. Finally, we observe the degree $k$ component of $e_{i, i+1}[p] \cdot T_r(z)$ for $1 \le k \le r-1$. This component arises from the action of $e_{i, i+1}[p]$ on the degree $k$ and $k+1$ part of $T_r(z)$. To track the degree $k$ component arising from the degree $k+1$ part of $T_r(z)$, we need to find the monomials containing either $\mu_{ii}(z)$ or $\mu_{i+1, i+1}(z)$ factors. For $(\alpha, \beta) \in \Jcal_{k+1, r}^N$, define
    \begin{equation*}
        \mu_{(\alpha, \beta)} := \mu_{\alpha_1\beta_1}(q^{\gamma_1}z) \cdots \mu_{\alpha_k \beta_k}(q^{\gamma_k}z).
    \end{equation*}
    For a fixed $1 \le x \le k+1$, let $(\alpha, \beta), (\alpha', \beta')$ be elements in $\Jcal_{k+1, r}^N$ that differ only at the $x$-th entry. Specifically, we assume
    \begin{equation*}
        \begin{aligned}
            &\alpha_x = \beta_x = i \quad \text{with} \quad \alpha_{x+1} \neq i+1, \\
            &\alpha'_x = \beta'_x = i+1 \quad \text{with} \quad \beta'_{x-1} \neq i,
        \end{aligned}
    \end{equation*}
    under the convention that $\alpha_{k+2} = N+1$ and $\beta'_0 = 0$. For all $y \neq x$, we have $\alpha_y=\alpha_y', \beta_y=\beta_y'$. Then, when $e_{i, i+1}[p]$ acts on the sum $\mu_{(\alpha, \beta)} + \mu_{(\alpha', \beta')}$, the degree $k$ part of the action vanishes due to the cancellation at the $x$-th factor:
    \begin{equation*}
        e_{i, i+1}[p] \cdot \big( \mu_{ii}(q^{\gamma_x}z) + \mu_{i+1, i+1}(q^{\gamma_x+1}z)\big) = 0.
    \end{equation*}
    Because these paired terms cancel, it suffices to consider only the cases for $(\alpha, \beta) \in \Jcal_{k+1, r}^N$ where, for some $1 \le x \le k$, we have either $\alpha_x=\beta_x=i$ with $\alpha_{x+1} = i+1$, or $\alpha_{x+1}=\beta_{x+1}=i+1$ with $\beta_{x}=i$.

    For the case $(\alpha, \beta) \in \Jcal_{k+1, r}^N$ with $\alpha_x=\beta_x=i, \alpha_{x+1}=i+1$, the action of $e_{i, i+1}[p]$ yields
    \begin{equation}\label{eq:ii}
    \begin{aligned}
        &-\big(\mu_{\alpha_1\beta_1}(q^{\gamma_1}z)\big)\dots \big(\mu_{\alpha_{x-1}\beta_{x-1}}(q^{\gamma_{x-1}}z)\big)\\
        &\quad \times \big((q^{\gamma_x+1}z)^p\big)\big(\mu_{i+1, \beta_{x+1}}(q^{\gamma_{x+1}+1}z)\big)\dots \big(\mu_{\alpha_{k+1}\beta_{k+1}}(q^{\gamma_{k+1}}z)\big).
    \end{aligned}
    \end{equation}
    By applying the index shift $\beta_{y+1} \mapsto \beta_{y}$ for $y=x, x+1, \dots, k+1$, $\alpha_{y+1} \mapsto \alpha_{y}$ for $y=x+1, \dots, k+1$, and by setting $\alpha_{x} = i$, we obtain a new element $(\alpha, \beta)\in \Jcal_{k, r}^N$. We observe that the contribution from the action of $e_{i, i+1}[p]$ on the term corresponding to this $(\alpha, \beta) \in \Jcal_{k, r}^N$ exactly cancels the term \eqref{eq:ii}. Similarly, for the case $(\alpha, \beta) \in \Jcal_{k+1, r}^N$ with $\alpha_{x+1}=\beta_{x+1}=i+1, \beta_x=i$, the action of $e_{i, i+1}[p]$ exactly cancels with the corresponding term in the case where $(\alpha, \beta) \in \Jcal_{k, r}^N$, $\beta_x=i+1$. Since this exhausts all possible contributions, we conclude that the degree $k$ component completely vanishes.
\end{proof}

In order to show $T_r(z)$ generates the deformed $W$-algebra, we find generators of $\widetilde{\Wbf}_q(LSL_N)$, which is defined via Lie group action. Since Proposition \ref{prop:gen2} establishes that $\Wbf_q(LSL_N)$ and $\widetilde{\Wbf}_q(LSL_N)$ are identical, we will henceforth refer to this space simply as $\Wbf_q(LSL_N)$. The following lemma is a direct analogue of the result in \cite{FRS98}.

\begin{lem}\label{lem:gen4}
    For any $A(t) \in B_+$, the orbit of the $q$-gauge action has a unique element which takes the form
     \begin{equation} \label{eq:canonical}
         A^{\mathrm{can}}(t) := g(qt)A(t)g(t)^{-1} = \begin{pmatrix}
             \widetilde{A}_1(t) & \widetilde{A}_2(t) & \dots & \widetilde{A}_{N-1}(t) & 1\\
             -1 & 0 & \dots & 0 & 0\\
             0 & -1 & \dots & 0 & 0\\
             \vdots & \vdots & \ddots & \vdots & \vdots \\
             0 & 0 & \dots & -1 & 0
         \end{pmatrix}.
     \end{equation}
     Moreover, the linear part of $\widetilde{A}_r(t)$ with respect to the variables $A_{ij}(t)$ is given by
     \begin{equation*}
         A_{1r}(t) + A_{2, r+1}(qt)+\dots +A_{N-r+1, N}(q^{N-r}t).
     \end{equation*}
\end{lem}
\begin{proof}
        Following the approach of Lemma 1 of \cite{FRS98}, we recursively eliminate the entries of $A(t)$ strictly above the subdiagonal, with the exception of those in the first row. However, while the elimination of \cite{FRS98} is performed row by row, we proceed column by column from right to left. This modification allows us to transparently track the linear terms and degree of $\widetilde{A}_r(t)$.
        
        For $\alpha=1, \dots, N$, let $B_+^\alpha$ denote the subset of matrices in $B_+$ whose entries in columns $i=\alpha+1, \dots, N$ are all zero, except for the first row and the subdiagonal entries at $(i+1, i)$ which remain $-1$. We will prove that for any $A(t) \in B_+^\alpha$ with $\alpha >1$, there exists $g(t) \in N_+$, such that $g(qt)A(t)g(t)^{-1} \in B_+^{\alpha-1}$. Since $B_+^N = B_+$, this implies that each $q$-gauge orbit contains an element of the canonical form \eqref{eq:canonical}.
        
        Let $E_{ij}(f(t)) \in N_+$ denote the upper unipotent matrix whose only non-zero off-diagonal entry is $f(t)$ located at the $(i, j)$ position. To perform the elimination in the $\alpha$-th column, we start with $A(t)=A^{(0)}(t) \in B_+^\alpha$. We apply the $q$-gauge action $E_{\alpha-1, \alpha}\big(-A_{\alpha\alpha}^{(0)}(t)\big)$ to $A^{(0)}(t)$. This action yields a new matrix $A^{(1)}(t)$ which remains in $B^{\alpha}_+$, but with its $(\alpha, \alpha)$-entry eliminated. For each step $\beta=1, \dots, \alpha-2$, we sequentially apply the $q$-gauge action $E_{\alpha-\beta-1, \alpha}\big(-A_{\alpha-\beta, \alpha}^{(\beta)}(t)\big)$ to $A^{(\beta)}(t)$. This action yields a new matrix $A^{(\beta+1)}(t)$ which remains in $B_+^{\alpha}$, but additionally has its $(\alpha-\beta, \alpha)$-entry eliminated. Finally, we obtain the matrix $A^{(\alpha-1)}(t) \in B_+^{\alpha-1}$.

        To investigate the entries $\widetilde{A}_r(t)$ of the canonical matrix $A^{\mathrm{can}}(t)$, we explicitly compute the effect of the $q$-gauge action of $E_{\alpha-\beta-1, \alpha}\big(-A_{\alpha-\beta, \alpha}^{(\beta)}(qt)\big)$ on the intermediate matrix $A^{(\beta)}(t)$. This transformation updates the matrix entries according to the following relations:
        \begin{equation}\label{eq:recursive}
        \begin{aligned}
            A_{x, y}^{(\beta+1)}(t) &= A_{x, y}^{(\beta)}(t) &&\quad\text{for } y\ne \alpha-1, \alpha\\
            A_{x, \alpha-1}^{(\beta+1)}(t) &= A_{x, \alpha-1}^{(\beta)}(t) + \delta_{x, \alpha-\beta-1}A_{x+1, \alpha}^{(\beta)}(qt) &&\\
            A_{x, \alpha}^{(\beta+1)}(t) &= A_{x, \alpha}^{(\beta)}(t)+A_{x, \alpha-\beta-1}^{(\beta)}(t)A_{\alpha-\beta, \alpha}^{(\beta)}(t) &&\quad\text{for } x \le \alpha-\beta-1\\
            A_{x, \alpha}^{(\beta+1)}(t) &= 0 &&\quad\text{for } x \ge \alpha-\beta.
        \end{aligned}
        \end{equation}
        Observe that new linear terms emerge exclusively in the $(\alpha-\beta-1, \alpha-1)$-entry. Tracking only the linear parts during the transformation of $A(t)$ into $A^{\mathrm{can}}(t)$, we can see that each initial entry $A_{ij}(t)$ moves to the upper-left diagonal direction $i-1$ times. At each step of this translation, the variable $t$ is replaced by $qt$. This completes the proof.
\end{proof}

\begin{prop}\label{prop:gen}
    The $q$-difference algebra $\Wbf_q(LSL_N)$ has a free generating set $\{\widetilde{T}_1(z), \dots, \widetilde{T}_{N-1}(z)\}$, where the linear part of $\widetilde{T}_r(z)$ with respect to the variables $\mu_{ij}(z)$ coincides with that of $T_r(z)$.
\end{prop}

\begin{proof}
    Note that $\widetilde{A}_r(t)$ in \eqref{eq:canonical} is a polynomial in the variables $A_{ij}(q^nt)$, where $i \le j$ and $n \in \ZZ$. Define $\widetilde{T}_r(z)$ by replacing $A_{ij}(q^nt)$ with $\mu_{ij}(q^nz)$ in the expression for $\widetilde{A}_r(t)$. If we formally define
    \begin{equation*}
        \mu_{ij}(z)(A(t)) := \sum_{m \in \ZZ} \mu_{ij}[m](A(t)) z^{-m-1},
    \end{equation*}
    then $\widetilde{T}_r(z) \in \Wbf_q(LSL_N)$ if and only if $ \widetilde{T}_r(z)(A(t)) = \widetilde{T}_r(z)(A^{\mathrm{can}}(t))$
    for all $A(t) \in B_+$. This follows directly from the definition of $\widetilde{T}_r(z)$.
    
    Now, let $f(z) \in \Wbf_q(LSL_N)$. By the $q$-gauge invariance of $f$, we have
    \begin{equation*}
        f(z)(A(t)) = f(z)(A^{\mathrm{can}}(t))
    \end{equation*}
    for any $A(t) \in B_+$. Since the right-hand side is a polynomial in the entries $(q^nz)^{-1}\widetilde{A}_r(q^nz)$, it follows that $f(z)$ is a polynomial in $\widetilde{T}_r(q^nz)$. Here, the factor $(q^nz)^{-1}$ comes from our convention $\mu_{ij}(z)(A(t)) = z^{-1}A_{ij}(z)$. Therefore, $\widetilde{T}_1(z), \dots, \widetilde{T}_{N-1}(z)$ generate $\Wbf_q(LSL_N)$.
    
    Finally, since their linear parts are algebraically independent, $\widetilde{T}_1(z), \dots, \widetilde{T}_{N-1}(z)$ are algebraically independent. Hence $\Wbf_q(LSL_N)$ is freely generated by $\widetilde{T}_1(z), \dots, \widetilde{T}_{N-1}(z)$.
\end{proof}

From Proposition~\ref{prop:gen2}, Lemma~\ref{lem:gen1}, and Proposition~\ref{prop:gen}, we finally have the following Theorem.

\begin{thm}\label{cor:gen}
    The $q$-difference algebra $\Wbf_q(LSL_N)$ is freely generated by $T_1(z), \dots, T_{N-1}(z)$.
\end{thm}

\begin{proof}
    Define a grading on $\Vbf_q(B_+)$ by 
    \begin{equation*}
        \Delta(\mu_{ij}(z)) = j-i+1,
    \end{equation*}
    and extend it by
    \begin{equation*}
        \Delta(f(z)g(z)) = \Delta(f(z)) + \Delta(g(z)).
    \end{equation*}
    By \eqref{eq:recursive}, the element $\widetilde{T}_r(z)$ is homogeneous of degree $r$. Moreover, \eqref{eq:generators} implies that $T_r(z)$ is also homogeneous of degree $r$. Hence
    \begin{equation*}
        \Wbf_q(LSL_N) = \bigoplus_{\Delta \ge 0} \Wbf_q(LSL_N)_{\Delta},
    \end{equation*}
    where $\Wbf_q(LSL_N)_{\Delta}$ denotes the degree $\Delta$ subspace of $\Wbf_q(LSL_N)$.

    Let $\overline{\Wbf}_q(LSL_N)$ be the $q$-difference subalgebra of $\Wbf_q(LSL_N)$ generated by $T_1(z), \dots, T_{N-1}(z)$ and let $\overline{\Wbf}_q(LSL_N)_{\Delta}$ denote its degree $\Delta$ subspace. Since $\Wbf_q(LSL_N)$ is freely generated by $\widetilde{T}_1(z), \dots, \widetilde{T}_{N-1}(z)$, it suffices to show that
    \begin{equation*}
        \widetilde{T}_r \in \overline{\Wbf}_q(LSL_N)_{r}
    \end{equation*}
    for every $1 \le r \le N-1$. We prove this by induction on $r$. Since $T_r(z)$ and $\widetilde{T}_r(z)$ have the same linear part, their difference
    \begin{equation*}
        \widetilde{T}_r(z) - T_r(z)
    \end{equation*}
    has degree $r$ and contains no linear terms in the generators $\mu_{ij}(z)$. Any such non-linear homogeneous element of degree $r$ must be formed by multiplying generators of strictly lower degrees, and so it is a polynomial in $\widetilde{T}_1(z), \dots, \widetilde{T}_{r-1}(z)$. By the induction hypothesis, $\widetilde{T}_r(z) - T_r(z)$ is a polynomial in $T_1(z), \dots, T_{r-1}(z)$, which implies that $\widetilde{T}_r(z) \in \overline{\Wbf}_q(LSL_N)_{r}$. Therefore $\Wbf_q(LSL_N)$ is freely generated by $T_1(z), \allowbreak \dots, \allowbreak T_{N-1}(z)$.
\end{proof}

\subsection{Poisson structure of the deformed \texorpdfstring{$W$}{W}-algebra} \label{subsec:Poisson of deformed W}
In this subsection, we demonstrate that the Poisson structure of $\Vbf_q(LSL_N)$ induces a well-defined Poisson structure on $\Wbf_q(LSL_N)$.
Furthermore, we compute the explicit Poisson brackets between generating series $T_1(z), \dots, T_{N-1}(z)$ of $\Wbf_q(LSL_N)$.

\begin{prop}\label{prop:PA}
    The Poisson bracket of $\Vbf_q(LSL_N)$ induces a well-defined Poisson algebra structure on $\Wbf_q(LSL_N)$.
\end{prop}

\begin{proof}
    For any $T_i(z), T_j(z) \in \Wbf_q(LSL_N)$, let $\overline{T}_i(z), \overline{T}_j(z)$ be arbitrary lifts to $\Vbf_q(LSL_N)$ via surjective $q$-difference algebra homomorphism $\pi_{B_+}$. From Lemmas~\ref{lem:well1} and \ref{lem:well2},
    \begin{equation*}
        \pi_{B_+}\{\overline{T}_i(z), \overline{T}_j(w)\} = \pi_{B_+}\{T_i(z), T_j(w)\},
    \end{equation*}
    which shows that the induced bracket is independent of the choice of lifts.
\end{proof}

\begin{prop}\label{prop:independenth}
    In the algebra $\Wbf_q(LSL_N)$, the Poisson brackets for $i \le j$ are explicitly given by
    \begin{align}
    \{T_i(z), T_j(w)\}&= \sum_{r \in \ZZ} \left( \frac{q^{j-i}w}{z} \right)^r \frac{(1-q^{ir})(1-q^{(N-j)r})}{1-q^{Nr}}T_i(z)T_j(w) \\
    &\notag+ \sum_{x=1}^{\min\{i, N-j\}} \delta\left( \frac{q^xw}{z}\right) T_{j+x}(z)T_{i-x}(w) \\
    &\notag- \sum_{x=1}^{\min\{i, N-j\}} \delta\left( \frac{w}{q^{j-i+x}z}\right) T_{i-x}(z) T_{j+x}(w),
    \end{align}
    where we set $T_{0}(z) =T_{N}(z) = 1$.
\end{prop}

\begin{proof}
    Suppose $i\le j$. To extract the terms corresponding to $T_{\alpha}$ from the Poisson bracket, it suffices to track the terms containing $\mu_{1\alpha}(z)$. Thus, in computing $\{T_{i}(z), T_{j}(w)\}$, we restrict our attention to the monomial terms in $\mu_{1\alpha}(z)$. The computations for the cases $i+j \le N$ and $i+j>N$ are analogous. For the sake of brevity, we assume $i+j\le N$ and omit the details for the latter case. Let us decompose the underlying bracket into three parts:
    \begin{equation*}
    \begin{aligned}
    \{\mu_{ij}(z), \mu_{kl}(w)\}_1=& C_{ijkl}\left(\frac{w}{z}\right) \mu_{ij}(z) \mu_{kl}(w),\\
    \quad \{\mu_{ij}(z), \mu_{kl}(w)\}_2 =&  \frac{1}{2}(\varepsilon_{ik}+\varepsilon_{lj}) \delta\left(\frac{w}{z}\right) \mu_{kj}(z) \mu_{il}(w),\\
    \{\mu_{ij}(z), \mu_{kl}(w)\}_3 =&- \sum_{\alpha > j} \delta_{jk}\delta\left(\frac{qw}{z}\right) \mu_{i\alpha}(z) \mu_{\alpha l}(w) + \sum_{\alpha > i}\delta_{il} \delta\left(\frac{w}{qz}\right)\mu_{\alpha j}(z) \mu_{k\alpha}(w).
    \end{aligned}
    \end{equation*}
    The following term is the only term of the monomials in $\mu_{1 \alpha}(z)$ appearing as a result of bracket $\{T_i(z), T_j(w)\}_1 +\{T_i(z), T_j(w)\}_2$:
    \begin{equation*}
    \begin{aligned}
        \{\mu_{1i}(z), \mu_{1j}(w)\}_1 &+ \{\mu_{1i}(z), \mu_{1j}(w)\}_2 = \sum_{r \in \ZZ}\frac{(1-q^{ir})(1-q^{(N-j)r})}{1-q^{Nr}}\left(\frac{w}{z}\right)^r\mu_{1i}(z)\mu_{1j}(w).
        \end{aligned}
    \end{equation*}
    From the bracket $\{\cdot, \cdot\}_3$, we obtain the following contributions:
    \begin{equation*}
    \begin{aligned}
        \left\{\mu_{1i}(z), \sum_{y=1}^{i-1}\mu_{1y}(w)\mu_{i, i+j-y-1}(q^{i-y-1}w)\right\}_3 &= \sum_{y=1}^{i-1}\delta\left(\frac{q^{i-y}w}{z}\right)\mu_{1, i+j-y}(z)\mu_{1y}(w),\\
        \left\{\mu_{1i}(z), \mu_{i, i+j-1}(q^{i-1}w)\right\}_3 &= \delta\left(\frac{q^iw}{z}\right) \mu_{1, i+j}(z), \\
        \left\{\sum_{y=1}^{i-1}\mu_{1y}(z)\mu_{j, j+i-y-1}(q^{j-y-1}z), \mu_{1j}(w)\right\}_3&= -\sum_{y=1}^{i-1}\delta\left(\frac{w}{q^{j-y}z}\right) \mu_{1y}(z)\mu_{1, i+j-y}(w),\\
        \left\{\mu_{j, j+i-1}(q^{j-1}z), \mu_{1j}(w)\right\}_3 &= -\delta\left(\frac{w}{q^jz}\right) \mu_{1, i+j}(w).
        \end{aligned}
    \end{equation*}
    Making the change of variables $x=i-y$, we obtain
    \begin{equation*}
    \begin{aligned}
    \{T_{i}(z), T_{j}(w)\} &=\sum_{r \in \ZZ} \frac{(1-q^{ir})(1-q^{(N-j)r})}{1-q^{Nr}}\left( \frac{w}{z} \right)^r T_{i}(z)T_{j}(w) \\
    &\quad + \sum_{x=1}^{i} \delta\left( \frac{q^xw}{z}\right) T_{j+x}(z)T_{i-x}(w)  - \sum_{x=1}^{i} \delta\left( \frac{w}{q^{j-i+x}z}\right) T_{i-x}(z) T_{j+x}(w)
    \end{aligned}
    \end{equation*}
    as desired.
\end{proof}

\subsection{Miura transformations} \label{subsec:q-miura}

We now introduce the deformed analogue of the Miura transformation. Recall the Poisson bracket formula of $\Vbf_q(LSL_N)$:
\begin{equation*}
    \begin{aligned}
        \{\mu_{ij}(z), \mu_{kl}(w)\}&= C_{ijkl}\left(\frac{w}{z}\right) \mu_{ij}(z)\mu_{kl}(w) +\frac{1}{2}(\varepsilon_{ik}+\varepsilon_{lj})\delta\left(\frac{w}{z}\right)\mu_{kj}(z)\mu_{il}(w) \\
        &\quad -\delta_{jk}\sum_{\alpha > j} \delta\left(\frac{qw}{z}\right)\mu_{i\alpha}(z)\mu_{\alpha l}(w) + \delta_{il} \sum_{\alpha > i} \delta\left(\frac{w}{qz}\right) \mu_{\alpha j}(z) \mu_{k\alpha}(w) .
    \end{aligned}
\end{equation*}
Note that the first term $C_{ijkl}\left(\frac{w}{z}\right) \mu_{ij}(z)\mu_{kl}(w)$ originates from the diagonal part of the classical $r$-matrix.
\begin{defn}\label{def:deformed_Heisenberg_algebra}
    The \emph{deformed Heisenberg algebra} $\Vbf_q(H)$ is the $q$-difference algebra generated by $\mu_{ii}(z)$ for $1\le i \le N$ subject to the relation $\prod_{i=1}^N \mu_{ii}(z) =1$.
    We endow $\Vbf_q(H)$ with the following Poisson structure:
    \begin{equation}\label{eq:HeisenbergPB}
    \begin{aligned}
        \{\mu_{ii}(z), \mu_{jj}(w) \} &:= C_{iijj}\left(\frac{w}{z}\right) \mu_{ii}(z)\mu_{jj}(w) \\
        &= \begin{cases}
            -\sum_{r \in \ZZ} \frac{(1-q^r)^2}{1-q^{Nr}} \left(\frac{q^{N+i-j-1}w}{z}\right)^r \mu_{ii}(z)\mu_{jj}(w)\quad &\text{if } i < j,\\
            \sum_{r \in \ZZ} \frac{(1-q^r)(1-q^{(N-1)r})}{1-q^{Nr}} \left(\frac{w}{z}\right)^r \mu_{ii}(z)\mu_{ii}(w)\quad &\text{if } i=j.
        \end{cases}
    \end{aligned}
    \end{equation}
\end{defn}

\begin{lem} \label{lem:partial-D}
    The algebra $\Vbf_q(H)$ equipped with the bracket \eqref{eq:HeisenbergPB} is a Poisson algebra.
\end{lem}

\begin{proof}
    Skew-symmetry follows directly from $C_{iijj}[r] = -C_{jjii}[-r].$
    The Jacobi identity follows directly from straightforward computation:
\begin{equation*}
    \begin{aligned}
    &\{\mu_{ii}(z), \{\mu_{jj}(w), \mu_{kk}(x)\}\}+\{\mu_{jj}(w), \{\mu_{kk}(x), \mu_{ii}(z)\}\}+\{\mu_{kk}(x), \{\mu_{ii}(z), \mu_{jj}(w)\}\}\\
    &=\Bigg(C_{iijj}\left(\frac{w}{z}\right)C_{jjkk}\left(\frac{x}{w}\right)+C_{iikk}\left(\frac{x}{z}\right)C_{jjkk}\left(\frac{x}{w}\right)+C_{jjkk}\left(\frac{x}{w}\right)C_{kkii}\left(\frac{z}{x}\right)\\
    &\quad +C_{jjii}\left(\frac{z}{w}\right)C_{kkii}\left(\frac{z}{x}\right)+C_{kkii}\left(\frac{z}{x}\right)C_{iijj}\left(\frac{w}{z}\right)+C_{kkjj}\left(\frac{w}{x}\right)C_{iijj}\left(\frac{w}{z}\right)\Bigg) \mu_{ii}(z)\mu_{jj}(w)\mu_{kk}(x).
    \end{aligned}
\end{equation*}
  Since $C_{iijj}\left(\frac{w}{z}\right) = -C_{jjii}\left(\frac{z}{w}\right)$, we see that the second and third terms cancel, the fourth and fifth terms cancel, and the first and sixth terms cancel.
\end{proof}

\begin{defn}
    The \emph{deformed Miura transformation} is the composition of the inclusion $\Wbf_q(LSL_N) \to \Vbf_q(B_+)$ and the restriction $\Vbf_q(B_+) \to \Vbf_q(H)$. More explicitly,
    \begin{equation}\label{eq:deformedMiura}
    \begin{aligned}
        m_q : \Wbf_q(LSL_N) \to \Vbf_q(H), \quad 
        T_r(z) \mapsto  \sum_{1\le \alpha_1 < \alpha_2 < \dots < \alpha_r \le N} \,  \prod_{x=1}^r \mu_{\alpha_x\alpha_x}(q^{\alpha_x-x}z).
    \end{aligned}
    \end{equation}
\end{defn}

In other words, the deformed Miura transformation simply maps $\mu_{ij}(z)$ to $0$ for $i < j$.

\begin{prop}\label{prop:Miura}
    The deformed Miura transformation $m_q$ is an injective Poisson algebra homomorphism. Hence $\Wbf_q(LSL_N)\simeq m_q(\Wbf_q(LSL_N)).$
\end{prop}

\begin{proof}
    Let us first show $m_q$ is injective. Suppose not, that is, we assume the elements $m_q(T_1(z))$, $\dots$, $m_q(T_{N-1}(z))$ are algebraically dependent. After ignoring $q$ and variable $z$, this implies that the elements
    \begin{equation} \label{eq:Miura_image}
        \sum_{\alpha=1}^N \mu_{\alpha\alpha}, \quad \sum_{\alpha_1<\alpha_2} \mu_{\alpha_1\alpha_1}\mu_{\alpha_2\alpha_2}, \dots, \sum_{\alpha_1<\dots<\alpha_{N-1}}\mu_{\alpha_1\alpha_1}\dots\mu_{\alpha_{N-1}\alpha_{N-1}},\  \mu_{11}\mu_{22}\dots\mu_{NN}
    \end{equation}
    are algebraically dependent in $\CC[\mu_{11}, \dots, \mu_{NN}]$. However, these are exactly the elementary symmetric polynomials in the $N$ variables $\mu_{11}, \dots, \mu_{NN}$. One can readily see that the polynomials in \eqref{eq:Miura_image} arise as the images of the Miura transformation applied to the generators of the classical finite \(W\)-algebra associated with \(\mathfrak{sl}_N\). Since the injectivity of the Miura map for classical finite \(W\)-algebras is a well-known fact, the proposition follows. 
    
    Now, let us show $m_q$ is a Poisson algebra homomorphism. Since we already know $m_q$ is injective, it is enough to show if either $i < j$ or $k < l$, none of the terms of the form $\mu_{i_1i_1}(z)\mu_{i_2i_2}(z)$ appear from the Poisson bracket $\{\mu_{ij}(z), \mu_{kl}(w)\}$.
  Consequently, they all vanish under the deformed Miura transformation $m_q$. This implies that the Poisson bracket is preserved.
\end{proof}

By Proposition~\ref{prop:Miura}, we can get a free generating set $m_q(T_1(z)), \cdots, m_q(T_{N-1}(z))$ of $m_q(\Wbf_q(LSL_N))$, which is isomorphic to $\Wbf_q(LSL_N)$, by the following formula:
    \begin{equation}\label{eq:Miura_deformed}
    (D +\mu_{11}(z))(D+\mu_{22}(z))\dots(D+\mu_{NN}(z))= D^{N}+m_q(T_1(z))D^{N-1}+\dots+m_q(T_{N-1}(z))D+1.
    \end{equation}
This formula will be crucially used in Section \ref{sec:W-w}.

\section{Limits of Deformed Affine Poisson Algebras and Deformed Heisenberg Algebras}\label{sec:affinelimit}

In this section, we put $q=e^{-h}$ and introduce the $h \to 0$ limit process which yields the PVAs $\Vbf(\gfr)$ from the $q$-difference Poisson algebras $\Vbf_q(G)$. Throughout this section, we fix $(G, \gfr)$ to be either $(LSL_N, \sfr\lfr_N)$ or $(G, \gfr) = (H, \hfr)$.

\vskip 2mm

\subsection{Algebras of Fourier modes}
In this subsection, we construct the Poisson $\CC(q)$-algebra $\Vcal_q(G)$ and the Poisson $\CC$-algebra $\Vcal(\gfr)$, which consist of the Fourier modes of $\Vbf_q(G)$ and $\Vbf(\gfr)$, respectively. Finally, we define the Poisson $\CC(\!(h)\!)$-algebra $\widetilde{\Vcal}_h(\gfr)$, which serves as a bridge between these two algebras.

The following $\CC(q)$-algebra introduced in Definition \ref{def:V_q}  contains all Fourier coefficients of the generating series in $\Vbf_q(LSL_N)$. It was originally introduced in Section 5.2 and 2.3 of \cite{FRS98} to study the Poisson structures of the deformed affine Poisson algebra and the deformed $W$-algebra, respectively.
\begin{defn} \label{def:V_q}
\mbox{}
\begin{enumerate}[(1)]
    \item The $\CC(q)$-algebra $\Vcal_q(L\Mat_N)$ consists of finite linear combinations of expressions of the form
    \begin{equation*}
        \sum_{m_1+\cdots + m_k = M} c_{m_1\dots m_k}(q) \cdot \mu_{i_1j_1}[m_1] \cdots \mu_{i_kj_k}[m_k]
    \end{equation*}
    where $c_{m_1 \dots m_k}(q) \in \CC(q)$.
    \item The $\CC(q)$-algebra $\Vcal_q(G)$ is defined as the quotient of the $\CC(q)$-algebra consisting of finite linear combinations of expressions of the form
    \begin{equation*}
        \sum_{m_1+\cdots + m_k = M} c_{m_1 \dots m_k}(q) \cdot \mu_{i_1j_1}[m_1] \cdots \mu_{i_kj_k}[m_k]
    \end{equation*}
    where $c_{m_1 \dots m_k}(q) \in \CC(q)$, by the ideal generated by
    \begin{equation*}
        \det[m]-\delta_{m, 0} \quad \text{for all } m \in \ZZ.
    \end{equation*}
We endow $\Vcal_q(G)$ with the Poisson structure defined in \eqref{eq:Poisson_R} or \eqref{eq:HeisenbergPB}.
Here, if $(G, \gfr) = (H, \hfr)$, we formally set $\mu_{ij}[m] = 0$ for all $m \in \ZZ$ whenever $i \ne j$. Recall that
\begin{equation*}
    \det[m] = \sum_{m_1+\dots+m_N=m}\sum_{\sigma \in \Sfr_N} (-1)^{{\mathrm{sgn}}(\sigma)} \mu_{1 \sigma(1)}[m_1] \mu_{2 \sigma(2)}[m_2] \dots \mu_{N \sigma(N)}[m_N].
\end{equation*}
\end{enumerate}
\end{defn}

We analogously define the algebras of Fourier modes for the non-deformed Poisson algebras. Again, if $(G, \gfr) = (H, \hfr)$, we formally set $e_{ij}[m] = 0$ for all $m \in \ZZ$ whenever $i \ne j$.
\begin{defn}
    The $\CC$-algebra $\Vcal(\gfr)$ is defined as the quotient of the $\CC$-algebra consisting of finite linear combinations of the form
    \begin{equation*}
        \sum_{m_1+\dots+m_k=M}c_{m_1 \dots m_k} \cdot e_{i_1j_1}[m_1] \cdots e_{i_kj_k}[m_k]
    \end{equation*}
    where $c_{m_1 \dots m_k} \in \CC$, by the ideal generated by
    \begin{equation*}
        e_{11}[m] + \dots + e_{NN}[m] \quad \text{for all } m \in \ZZ.
    \end{equation*}
    We endow $\Vcal(\gfr)$ with the Poisson structure
    \begin{equation*}
        \{e_{ij}[m], e_{kl}[n]\} = \delta_{jk}e_{il}[m+n]-\delta_{il}e_{kj}[m+n] + m\delta_{m+n, 0} \left( \delta_{il}\delta_{jk} - \frac{1}{N} \delta_{ij}\delta_{kl}\right).
    \end{equation*}
\end{defn}

Finally, we define a $\CC(\!(h)\!)$-algebra $\widetilde{\Vcal}_h(\gfr)$, which captures $h$-expansion behavior of the algebra $\Vcal_q(G)$. If $(G, \gfr) = (H, \hfr)$, we again formally set $\mathsf{E}_{ij}[m] = 0$ for all $m \in \ZZ$ whenever $i \ne j$.
\begin{defn} \label{eq:affine_h-algebras}
\mbox{}
\begin{enumerate}[(1)]
    \item The $\CC(\!(h)\!)$-algebra $\widetilde{\Vcal}_h(\gfr\lfr_N)$ is the $\CC(\!(h)\!)$-algebra consisting of finite linear combinations of expressions of the form
    \begin{equation*}
        \sum_{m_1+\cdots + m_k = M} c_{m_1 \dots m_k}(h) \cdot \mathsf{E}_{i_1j_1}[m_1] \cdots \mathsf{E}_{i_kj_k}[m_k]
    \end{equation*}
    where $c_{m_1 \dots m_k}(h) \in \CC(\!(h)\!)$, and the $\CC$-algebra homomorphism $\iota : \Vcal_q(L\Mat_N) \to \widetilde{\Vcal}_h(\gfr\lfr_N)$ is defined by
    \begin{equation*}
        \mu_{ji}[m] \mapsto \delta_{ij}\delta_{m, 0} + h\mathsf{E}_{ij}[m], \quad q \mapsto e^{-h}.
    \end{equation*}
    \item The $\CC(\!(h)\!)$-algebra $\widetilde{\Vcal}_h(\gfr)$ is the quotient of the $\CC(\!(h)\!)$-algebra consisting of finite linear combinations of expressions of the form
    \begin{equation*}
        \sum_{m_1+\cdots + m_k = M} c_{m_1 \dots m_k}(h) \cdot \mathsf{E}_{i_1j_1}[m_1] \cdots \mathsf{E}_{i_kj_k}[m_k]
    \end{equation*}
    where $c_{m_1 \dots m_k}(h) \in \CC(\!(h)\!)$, by the ideal generated by
    \begin{equation*}
        \iota\big( \det[m] \big) - \delta_{m, 0} \quad \text{for all } m \in \ZZ.
    \end{equation*}
    We endow $\widetilde{\Vcal}_h(\gfr)$ with the Poisson structure
    \begin{equation*}
            \{\mathsf{E}_{ij}[m], \mathsf{E}_{kl}[n]\}_h := \frac{1}{h} \iota \big( \{ \mu_{ji}[m], \mu_{lk}[n]\}\big),
    \end{equation*}
    and by extending via the Leibniz rule.
\end{enumerate}
\end{defn}

\begin{ex}
    The algebra $\widetilde{\Vcal}_h(\sfr\lfr_2)$ is the quotient of the $\CC(\!(h)\!)$-algebra $\widetilde{\Vcal}_h(\gfr\lfr_2)$ by the $q$-difference ideal generated by
    \begin{equation*}
        h\big(\mathsf{E}_{11}[m] + \mathsf{E}_{22}[m]\big) + h^2\sum_{m_1+m_2=m}\big(\mathsf{E}_{11}[m_1]\mathsf{E}_{22}[m_2] - \mathsf{E}_{12}[m_1] \mathsf{E}_{21}[m_2] \big)
    \end{equation*}
    for $m \in \ZZ$. The following are Poisson brackets on $\widetilde{\Vcal}_h(\sfr\lfr_2)$:
    \begin{equation*}
    \begin{aligned}
        \{\mathsf{E}_{11}[m], \mathsf{E}_{22}[n] \}_h &= -\frac{1}{h} \sum_{r \in \ZZ} \frac{1-q^r}{1+q^r} \iota(\mu_{11}[m-r])\iota(\mu_{22}[n+r])\\
        &=-h\sum_{r \in \ZZ} \frac{1-q^r}{1+q^r} \mathsf{E}_{11}[m-r]\mathsf{E}_{22}[n+r] \\
        &\quad -\frac{1-q^m}{1+q^m} \mathsf{E}_{22}[m+n] - \frac{1-q^{-n}}{1+q^{-n}}\mathsf{E}_{11}[m+n] \\
        &\quad - \frac{1}{h}\frac{1-q^m}{1+q^m} \delta_{m+n, 0},\\
        \{\mathsf{E}_{11}[m], \mathsf{E}_{12}[n]\}_h &= q^n \mathsf{E}_{12}[m+n] + h\sum_{r \in \ZZ} q^{-r} \mathsf{E}_{12}[m-r]\mathsf{E}_{22}[n+r].
    \end{aligned}
    \end{equation*}
\end{ex}

\subsection{The limit and convergence}
In this subsection, we define the subalgebra $\Vcal_h(\gfr) \subset \widetilde{\Vcal}_h(\gfr)$ by extracting the elements that admit a well-defined limit as $h \to 0$.  To make this process precise, we consider the subring $A_1$ of $\CC(q)$ regular at $q=1$.

\vskip 2mm
Observe that not all elements in $\widetilde{\Vcal}_h(\gfr)$ possess a limit. For instance, the element $
    \frac{1}{h}\mathsf{E}_{ij}[m]$
does not admit a limit. In contrast, for any $l \in \ZZ_{>0}$, recalling that $q=e^{-h}$, the elements
\begin{equation}\label{eq:derivation}
    \frac{1-q^{(-m-1)l}}{lh} \mathsf{E}_{ij}[m]
\end{equation}
admit a well-defined limit and converge to $(-m-1)\mathsf{E}_{ij}[m]$. Therefore, we extract the elements that exhibit well-defined behavior under the $h \to 0$ limit. Via the natural inclusion $A_1 \hookrightarrow \CC(\!(h)\!)$ induced by $q=e^{-h}$, we regard $\widetilde{\Vcal}_h(\gfr)$ as an $A_1$-algebra.

\begin{defn}
The $A_1$-algebra $\Vcal_h(\gfr)$ is the subalgebra of $\widetilde{\Vcal}_h(\gfr)$ consisting of finite linear combinations of expressions of the form
\begin{equation*}
        \sum_{m_1+\cdots + m_k = M} c_{m_1 \dots m_k}(h) \cdot \mathsf{E}_{i_1j_1}[m_1] \cdots \mathsf{E}_{i_kj_k}[m_k],
\end{equation*}
    where the limit $
    \lim_{h \to 0}c_{m_1 \dots m_k}(h)$
    exists for all $m_1+ \dots + m_k = M$. The {\em $h\to 0$ limit} 
\begin{equation} \label{eq:limit}
    \lim_{h \to 0}: \Vcal_h(\gfr) \to \Vcal(\gfr)
\end{equation}
is the linear map given by $ c_{m_1 \dots m_k}(h) \cdot  \prod_{s=1}^k\mathsf{E}_{i_sj_s}[m_s] \mapsto \lim_{h \to 0}(c_{m_1 \dots m_k}(h)) \cdot \prod_{s=1}^k (e_{i_s j_s}-\delta_{i_s,j_s}\frac{I_N}{N})[m_s]$ for $I_N=e_{11}+\cdots+e_{NN}.$
\end{defn}

It is clear that $\Vcal_h(\gfr)$ is closed under the multiplication, since the product of two functions admitting well-defined limits also admits a limit. Furthermore, note that the coefficient $q \in A_1$ specializes to $1$ in this limit. Therefore, we obtain a $\CC$-algebra homomorphism
\begin{equation*}
    \lim_{h \to 0} : \Vcal_h(\gfr) \to \Vcal(\gfr)
\end{equation*}
by taking limit $h \to 0$.

We introduce a family of elements $\mathsf{E}_{ij}^{k, l}[m]$ in $\Vcal_h(\gfr)$ that play the role of derivations of $\mathsf{E}_{ij}[m]$ in the limit. To motivate this definition, we multiply the expression in \eqref{eq:derivation} by the formal variable $z^{-m-1}$ and sum over $m \in \ZZ$, which yields
    \begin{equation} \label{eq:motiv_generator}
        \lim_{h \to 0} \left. \left( \frac{1-D^l}{lh} \right)\mathsf{E}_{ij}(z)\right|_{q=e^{-h}}= z\partial_z\mathsf{E}_{ij}(z).
    \end{equation}
    Recall the standard identity for the operators: $\prod_{s=0}^{k-1}(z\partial_z -s) = z^k\partial_z^k$. In light of this relation, the product of difference operators
    \begin{equation}
        \left(\frac{1-D^l}{lh}\right)\left(\frac{1-D^l}{lh}-1\right)\dots\left(\frac{1-D^l}{lh}-k+1\right)
    \end{equation}
    can be naturally viewed as the operator $z^k\partial_z^k$ in the limit $h \to 0$ with $q=e^{-h}$. Accordingly, we define the elements $\mathsf{E}_{ij}^{k, l}[m]$ for $1 \le i, j \le N, l \in \ZZ_{>0}, k \in \ZZ_{\ge 0}$, and $m \in \ZZ$ as follows:
    \begin{equation} \label{eq:E_{ij}^{kl}}
    \mathsf{E}_{ij}^{k, l}[m] := \begin{cases}
    \mathsf{E}_{ij}[m] \quad &\text{ if } k = 0,\\
        \prod_{s=0}^{k-1} \left(\frac{1-q^{(-m-1)l}}{lh}-s\right)\mathsf{E}_{ij}[m] \quad &\text{ if } k>0.
    \end{cases}
\end{equation}
If we set the formal series 
$ \mathsf{E}_{ij}^{k, l}(z) = \sum_{m \in \ZZ} \mathsf{E}_{ij}^{k, l}[m]z^{-m-k-1}$ and $e_{ij}(z) = \sum_{m \in \ZZ} e_{ij}[m]z^{-m-1}$ then we have 
\begin{equation} \label{eq:derivation of E_{ij}}
    \lim_{h \to 0}\big(\mathsf{E}_{ij}^{k, l}(z)\big) = \partial_z^k \Big(e_{ij}-\delta_{ij}\frac{I_N}{N}\Big)(z).
\end{equation}
Later, \eqref{eq:derivation of E_{ij}} will be used to describe the relationship between difference algebras and differential algebras.

Next, we investigate the Poisson structure on $\Vcal_h(\gfr)$. Recall from \eqref{eq:C} that the numerator of $C_{ijkl}[m]$ vanishes at $q=1$. Thus, although the denominator of $C_{ijkl}[m]$ contains the factor $1-q^{Nm}$, the limit
$  \lim_{q \to 1} C_{ijkl}[m] $
exists for all $1 \le i, j, k, l \le N$, and $m \in \ZZ$. Consequently, even when we view $\widetilde{\Vcal}_h(\gfr)$ as an $A_1$-algebra, the Poisson bracket remains well-defined.

\begin{prop}\label{prop:affinelimit}
Let $1 \le i, j, k, l \le N$ and $m, n \in \ZZ$.
\begin{enumerate}[(1)]
    \item For $\mathsf{E}_{ij}[m],\mathsf{E}_{kl}[n]\in \Vcal_h(\sfr\lfr_N)$, the Poisson bracket $\{\mathsf{E}_{ij}[m], \mathsf{E}_{kl}[n]\}_h \in \Vcal_h(\sfr\lfr_N)$ and the limit is 
    \begin{equation*}
        \begin{aligned}
            \lim_{h\to 0}\big(\{\mathsf{E}_{ij}[m], \mathsf{E}_{kl}[n]\}_h\big) &= \delta_{jk}\lim_{h\to 0}(\mathsf{E}_{il}[m+n]) - \delta_{il}\lim_{h\to 0}(\mathsf{E}_{kj}[m+n]) + \left(\delta_{jk}\delta_{il}-\frac{1}{N}\delta_{ij}\delta_{kl} \right) m\delta_{m+n, 0}.
        \end{aligned}
    \end{equation*}
    Equivalently, in terms of the formal series, we have 
    \begin{equation*}
        \lim_{h\to 0}\big(\{\mathsf{E}_{ij}(z), \mathsf{E}_{kl}(w)\}_h\big) = \left(\delta_{jk}\lim_{h\to 0}(\mathsf{E}_{il}(w)) - \delta_{il}\lim_{h\to 0}(\mathsf{E}_{kj}(w))\right) \delta(z, w) + \left(\delta_{jk}\delta_{il} - \frac{1}{N} \delta_{ij}\delta_{kl}\right) \partial_w\delta(z, w),
    \end{equation*}
    where $\delta(z, w) := \sum_{r \in \ZZ} z^{-r-1}w^r$.
    \item For $\mathsf{E}_{ij}[m],\mathsf{E}_{kl}[n]\in  \Vcal_h(\hfr)$, the Poisson bracket $\{\mathsf{E}_{ii}[m], \mathsf{E}_{jj}[n]\}_h \in \Vcal_h(\hfr)$ and the limit satisfies 
    \begin{equation*}
        \lim_{h\to 0}\big(\{\mathsf{E}_{ii}(z), \mathsf{E}_{jj}(w)\}_h\big) = \left(\delta_{ij}- \frac{1}{N} \right) \partial_w\delta(z, w).
    \end{equation*}
\end{enumerate}
\end{prop}

\begin{proof}
We present the proof in terms of generating series. The proof of the second statement is omitted, as it follows from the exact same computations. We first examine $\{\mathsf{E}_{ij}(z), \mathsf{E}_{kl}(w)\}_h$:
    \begin{equation}\label{eq:expansion}
\begin{aligned}
    &\{\mathsf{E}_{ij}(z), \mathsf{E}_{kl}(w)\}_h \\
    &=  h\Bigg(C_{jilk}\left(\frac{w}{z}\right)\mathsf{E}_{ji}(z)\mathsf{E}_{lk}(w) +\frac{1}{2}(\varepsilon_{jl}+\varepsilon_{ki})\delta\left(\frac{w}{z}\right)\mathsf{E}_{li}(z)\mathsf{E}_{jk}(w)\\
    &\qquad -\delta_{il}\sum_{\alpha > i}\delta \left(\frac{w}{qz}\right)\mathsf{E}_{j\alpha}(z)\mathsf{E}_{\alpha k}(w) + \delta_{jk}\sum_{\alpha > j} \delta\left(\frac{qw}{z}\right) \mathsf{E}_{\alpha i}(z)\mathsf{E}_{l \alpha}(w)\Bigg)\\
    &\quad +\frac{1}{w}\Bigg(C_{jilk}\left(\frac{w}{z}\right)\delta_{kl}\mathsf{E}_{ij}(z) \\
    &\qquad + \frac{1}{2}\delta_{jk}(\varepsilon_{jl}+\varepsilon_{ji})\mathsf{E}_{il}(z)\delta\left(\frac{w}{z}\right) -\delta_{il}\delta_{k>i}\mathsf{E}_{kj}(z)\delta\left(\frac{w}{qz}\right)+\delta_{jk}\delta_{l>j}\mathsf{E}_{il}(z)\delta\left(\frac{qw}{z}\right)\Bigg)\\
    &\quad +\frac{1}{z}\Bigg(C_{jilk}\left(\frac{w}{z}\right)\delta_{ij}\mathsf{E}_{kl}(w) \\
    &\qquad +\frac{1}{2}\delta_{il}(\varepsilon_{ji}+\varepsilon_{ki})\mathsf{E}_{kj}(w)\delta\left(\frac{w}{z}\right) -\delta_{il}\delta_{j>i}\mathsf{E}_{kj}(w)\delta\left(\frac{w}{qz}\right)+\delta_{jk}\delta_{i>j}\mathsf{E}_{il}(w)\delta\left(\frac{qw}{z}\right)\Bigg)\\
    &+ \frac{1}{zwh}\Bigg(C_{jilk}\left(\frac{w}{z}\right)\delta_{ij}\delta_{kl} + \delta_{j>i}\delta_{il}\delta_{jk}\left(\delta\left(\frac{w}{z}\right) -\delta\left(\frac{w}{qz}\right)\right) + \delta_{i>j} \delta_{il}\delta_{jk} \left(\delta\left(\frac{qw}{z}\right)-\delta\left(\frac{w}{z}\right)\right)\Bigg).
\end{aligned}
\end{equation}
    Evaluating the limit of the terms in the last line of \eqref{eq:expansion} yields:
    \begin{equation}\label{eq:limit1}
    \begin{aligned}
        \lim_{h \to 0} \frac{1}{zwh}C_{iikk}\left(\frac{w}{z}\right)\delta_{ij}\delta_{kl} &= \begin{cases}
        \delta_{ij}\delta_{kl} \displaystyle\lim_{h\to 0} \sum_{r \in \ZZ}q^{(N-i+k-1)r} \frac{(1-q^r)^2}{h(1-q^{Nr})}z^{-r-1}w^{r-1} \quad &(i>k)\\
        -\delta_{ij}\delta_{kl} \displaystyle\lim_{h\to 0} \sum_{r \in \ZZ}\frac{(1-q^{r})(1-q^{(N-1)r})}{h(1-q^{Nr})}z^{-r-1}w^{r-1} \quad &(i=k)\\
        \delta_{ij}\delta_{kl} \displaystyle\lim_{h\to 0} \sum_{r \in \ZZ}q^{(k-i-1)r}\frac{(1-q^r)^2}{h(1-q^{Nr})} z^{-r-1}w^{r-1}\quad &(i<k)
        \end{cases} \\
        &= \delta_{ij}\delta_{kl}\left(\delta_{ik}-\frac{1}{N}\right)\partial_w \delta(z, w) = \left(\delta_{il}\delta_{jk}\delta_{ij} - \frac{1}{N}\delta_{ij}\delta_{kl}\right) \partial_w\delta(z, w),
        \end{aligned}
    \end{equation}
    \begin{equation}\label{eq:limit2}
        \lim_{h \to 0} \frac{1}{zwh}\delta_{j>i}\delta_{il}\delta_{jk}\left(\delta\left(\frac{w}{z}\right) -\delta\left(\frac{w}{qz}\right)\right)=\delta_{il}\delta_{jk}\delta_{j>i}\partial_w\delta(z, w),
    \end{equation}
    and
    \begin{equation}\label{eq:limit3}
            \lim_{h \to 0} \frac{1}{zwh}\delta_{i>j} \delta_{il}\delta_{jk} \left(\delta\left(\frac{qw}{z}\right)-\delta\left(\frac{w}{z}\right)\right)= \delta_{il}\delta_{jk}\delta_{i>j}\partial_w\delta(z, w).
    \end{equation}
    Combining \eqref{eq:limit1}+\eqref{eq:limit2}+\eqref{eq:limit3}, we obtain
    \begin{equation*}
        \left(\delta_{il}\delta_{jk}-\frac{1}{N}\delta_{ij}\delta_{kl}\right)\partial_w\delta(z, w).
    \end{equation*}
    Now we observe the middle terms of \eqref{eq:expansion}:
    \begin{equation}\label{eq:limit4}
    \begin{aligned}
        \lim_{h\to 0}\left(\frac{1}{w}\delta_{kl}C_{jilk}\left(\frac{w}{z}\right)\mathsf{E}_{ij}(z) \right)&= \frac{1}{2}\delta_{kl}(\delta_{jk}-\delta_{ik})\lim_{h\to 0}\big(\mathsf{E}_{ij}(w)\big)\delta(z, w)\\
        &= \frac{1}{2}\delta_{jk}\delta_{jl}\lim_{h\to 0}\big(\mathsf{E}_{il}(w)\big)\delta(z, w)-\frac{1}{2}\delta_{il}\delta_{ik}\lim_{h\to 0}\big(\mathsf{E}_{kj}(w)\big)\delta(z, w)\\
        \lim_{h\to 0}\left(\frac{1}{z}\delta_{ij}C_{jilk}\left(\frac{w}{z}\right)\mathsf{E}_{kl}(w) \right)&= \frac{1}{2}\delta_{ij}(\delta_{ik}-\delta_{il})\lim_{h\to 0}\big(\mathsf{E}_{kl}(w)\big)\delta(z, w)\\
        &= \frac{1}{2}\delta_{jk}\delta_{ij}\lim_{h\to 0}\big(\mathsf{E}_{il}(w)\big)\delta(z, w)-\frac{1}{2}\delta_{il}\delta_{ij}\lim_{h\to 0}\big(\mathsf{E}_{kj}(w)\big)\delta(z, w),
    \end{aligned}
    \end{equation}
    and the rest of the terms in \eqref{eq:expansion} converge to
    \begin{equation}\label{eq:limit5}
        \frac{1}{2}\delta_{jk}(\delta_{j\ne l}+\delta_{i\ne j})\lim_{h\to 0}\big(\mathsf{E}_{il}(w)\big)\delta(z, w) - \frac{1}{2}\delta_{il}(\delta_{i\ne j} + \delta_{i\ne k})\lim_{h\to 0}\big(\mathsf{E}_{kj}(w)\big)\delta(z, w).
    \end{equation}
    Combining \eqref{eq:limit4} + \eqref{eq:limit5}, we obtain
    \begin{equation*}
        \delta_{jk}\lim_{h\to 0}\big(\mathsf{E}_{il}(w)\big)\delta(z, w) - \delta_{il}\lim_{h\to 0}\big(\mathsf{E}_{kj}(w)\big)\delta(z, w).
    \end{equation*}
    Since $\lim_{h \to 0} hC_{jilk}\left(\frac{w}{z}\right) = 0$, no further nontrivial terms can arise in the limit. This completes the proof.
\end{proof}

\begin{cor}
    The $A_1$-algebra $\Vcal_h(\gfr)$ is a Poisson subalgebra of $\widetilde{\Vcal}_h(\gfr)$.
\end{cor}

\begin{proof}
    Let $X$, $Y$ be two arbitrary elements in $\Vcal_h(\gfr)$, written as finite linear combinations of the form
    \begin{equation*}
        \begin{aligned}
            X&=\sum_{m_1+\cdots + m_k = M} c_{m_1 \dots m_k}(h) \cdot \mathsf{E}_{i_1j_1}[m_1] \cdots \mathsf{E}_{i_kj_k}[m_k],\\
            Y&=\sum_{m_1'+\cdots + m_{k'}' = M'} c'_{m_1' \dots m_{k'}'}(h) \cdot \mathsf{E}_{i_1'j_1'}[m_1'] \cdots \mathsf{E}_{i_{k'}'j_{k'}'}[m_{k'}'].
        \end{aligned}
    \end{equation*}
    By applying the Leibniz rule for the Poisson bracket, it suffices to show that
    \begin{equation*}
        \lim_{h \to 0}\Big( c_{m_1 \dots m_k}(h)c'_{m_1' \dots m_k'}(h) \{\mathsf{E}_{i_sj_s}[m_s'], \mathsf{E}_{i_t'j_t'}[m_t']\}_h\Big)
    \end{equation*}
    exists for all choices of $s$ and $t$. This follows from Proposition~\ref{prop:affinelimit}.
\end{proof}

The following diagram summarizes the limit process from $\Vbf_q(G)$ to $\Vcal(\gfr)$ for $(G,\gfr)=(LSL_N,\sfr\lfr_N)$ or $(H,\hfr)$:
\begin{equation}
    \begin{tikzcd}[column sep=large]
    \Vbf_q(G) \arrow[d, dashed, "\text{\scriptsize Fourier modes}"'] & & \Vbf(\gfr) \arrow[d, dashed, "\text{\scriptsize Fourier modes}"]\\
    \Vcal_q(G) \arrow[r, hook, "\iota"] & \widetilde{\Vcal}_h(\gfr) \supset \Vcal_h(\gfr) \arrow[r, two heads, "\lim_{h\to 0}"] & \Vcal(\gfr)
    \end{tikzcd}
\end{equation}

\noindent More precisely, the following theorem describes the limit process developed in this section.

\begin{thm} \label{thm:main_affine}
Recall the formal generating series $\mathsf{E}_{ij}^{k,l}(z)=\sum_{m\in \mathbb{Z}} \mathsf{E}_{ij}^{k,l}[m] z^{-m-k-1}$ for the Fourier mode $\mathsf{E}_{ij}^{kl}[m]$ defined in  \eqref{eq:E_{ij}^{kl}}  and its derivatives \eqref{eq:derivation of E_{ij}}.
\begin{enumerate}[(1)]
    \item The $A_1$-algebra $\Vcal_h(\sfr\lfr_N)$ is a Poisson subalgebra of $\widetilde{\Vcal}_h(\sfr\lfr_N)$, and
    \begin{equation*}
        \lim_{h \to 0} : \Vcal_h(\sfr\lfr_N) \to \Vcal(\sfr\lfr_N)
    \end{equation*}
    given in \eqref{eq:limit} is a Poisson algebra homomorphism. Furthermore, in terms of the formal generating series, the limit and derivation commutes:
    \begin{equation*}
        \lim_{h\to 0}\big(\{\mathsf{E}_{ij}^{m_1, k_1}(z), \mathsf{E}_{kl}^{m_2, k_2}(w)\}_h\big) = \partial_z^{k_1}\partial_w^{k_2}\lim_{h\to 0}\big(\{\mathsf{E}_{ij}(z), \mathsf{E}_{kl}(w)\}_h\big).
    \end{equation*}
    Consequently, the limit recovers the Poisson bracket of $\Vbf(\sfr\lfr_N)$.
    \item The $A_1$-algebra $\Vcal_h(\hfr)$ is a Poisson subalgebra of $\widetilde{\Vcal}_h(\hfr)$, and
    \begin{equation*}
        \lim_{h \to 0} : \Vcal_h(\hfr) \to \Vcal(\hfr)
    \end{equation*}
    is a Poisson algebra homomorphism. Furthermore, in terms of the formal generating series, the limit and derivation commutes. Consequently, the limit recovers the Poisson bracket of $\Vbf(\hfr)$.
\end{enumerate}
\end{thm}

Since the Poisson structures on $\Vcal_h(\sfr\lfr_N)$ and $\Vcal_h(\hfr)$ are induced from $\Vbf_q(LSL_N)$ and $\Vbf_q(H)$, the theorem above rigorously justifies regarding the Poisson algebras $\Vbf_q(LSL_N)$ and $\Vbf_q(H)$ as the $q$-deformations of the Poisson algebras $\Vbf(\sfr\lfr_N)$ and $\Vbf(\hfr)$.

\section{Connection between classical \texorpdfstring{W}{W}-algebras and deformed \texorpdfstring{W}{W}-algebras} \label{sec:W-w}

In this section, we prove the main theorem of the paper, which establishes classical $W$-algebras as the limits of deformed $W$-algebras. This result concretizes what had been anticipated in the last part of \cite{FR96}*{Section~10.2}.

\vskip 2mm
Recall that $\Wbf_q(LSL_N)$ and $\Wbf(\sfr\lfr_N)$ admit embeddings into the Heisenberg algebras $\Vbf_q(H)$ and $\Vbf(\hfr)$, respectively. Throughout this section, we identify the $W$-algebras with their images under these embeddings and omit the symbols $m$ and $m_q$. In other words, according to Proposition \ref{prop:generators}, the generator $T_r(z)$ of $\Wbf_q(LSL_N)$ in this section is 
\begin{equation} \label{eq:T in Sec 5}
    T_r(z)  = \sum_{1 \le \alpha_1<\dots < \alpha_r \le N} \prod_{x=1}^r \mu_{\alpha_x\alpha_x}(q^{\alpha_x-x}z) \in \Wbf_q(LSL_N).
\end{equation}
The following definition gives analogous notions for $W$-algebras to the one defined in the previous section.
\begin{defn} Let $T(z)\in \Wbf_q(LSL_N)$ be in \eqref{eq:T in Sec 5}.
    \begin{enumerate}[(1)]
        \item For $1 \le r \le N$, the element in $\Vcal_q(H)$ denoted by  $T_r[m] $ is the Fourier mode in  
        \begin{equation*}
            T_r(z) = \sum_{m \in \ZZ} T_r[m] z^{-m-r}.
        \end{equation*} 
        \item Let $\Wcal_q(LSL_N)$ be the $\CC(q)$-subalgebra of $\Vcal_q(H)$ generated by the elements $T_r[m]$ for $1 \le r \le N-1$ and $m \in \ZZ$.
        \item For the map $\iota$ in Definition \ref{eq:affine_h-algebras} (1), let  $\widetilde{\Wcal}_h(\sfr\lfr_N)$ be the $\CC(\!(h)\!)$-subalgebra of $\widetilde{\Vcal}_h(\hfr)$ generated by the images $\iota(\Wcal_q(LSL_N))$. We denote
        \begin{equation*}
            \mathsf{T}_r(z):= \sum_{m\in \mathbb{Z}} \mathsf{T}_r[m] z^{-m-r} \quad  \text{ for } \quad \mathsf{T}_r[m] := \iota(T_r[m]).
        \end{equation*}
        \item We denote 
            \begin{equation*}
                \Wcal_h(\sfr\lfr_N) := \widetilde{\Wcal}_h(\sfr\lfr_N) \cap \Vcal_h(\hfr).
            \end{equation*} In other words, it is the $A_1$-subalgebra of  $\widetilde{\Wcal}_h(\sfr\lfr_N)$ consisting of elements with well-defined $h\to 0 $ limit.
        \end{enumerate}
\end{defn}

In the following example, we explain how to find $\mathsf{T}_r[m]$ in $\widetilde{\Wcal}_h(\sfr\lfr_N).$ 

\begin{ex}
If $N=3$, the Fourier modes of $T_r(z)$ can be expressed as follows.
\begin{equation*}
    \begin{aligned}
        T_1[m] &= \mu_{11}[m] + q^{-m-1}\mu_{22}[m] + q^{-2m-2}\mu_{33}[m], \\
        T_2[m] &=\sum_{m_1+m_2=m} \big(\mu_{11}[m_1]\mu_{22}[m_2] + q^{-m_2-1}\mu_{11}[m_1]\mu_{33}[m_2] + q^{-m-2}\mu_{22}[m_1]\mu_{33}[m_2]\big).
    \end{aligned}
\end{equation*}
Note that $   T_3[m] = \sum_{m_1+m_2+m_3=m} \mu_{11}[m_1]\mu_{22}[m_2]\mu_{33}[m_3] = \delta_{m, 0}$
in the algebra $\Vcal_q(H)$. On the other hand, the Fourier modes of $\mathsf{T}_r(z)$ are as follows:
\begin{equation}\label{eq:N=3examples1}
\begin{aligned}
    \mathsf{T}_1[m] &= \delta_{m,0}(1+q^{-1}+q^{-2}) + h(\mathsf{E}_{11}[m]+ q^{-m-1}\mathsf{E}_{22}[m] + q^{-2m-2}\mathsf{E}_{33}[m]), \\
    \mathsf{T}_2[m] &= \delta_{m,0}(1+q^{-1}+q^{-2}) \\
    &\quad + h\big((1+q^{-1})\mathsf{E}_{11}[m] + \mathsf{E}_{22}[m]+q^{-m-2}\mathsf{E}_{22}[m] + (q^{-m-1}+q^{-m-2})\mathsf{E}_{33}[m]\big)\\
    &\quad + h^2\sum_{m_1+m_2=m} \big( \mathsf{E}_{11}[m_1]\mathsf{E}_{22}[m_2] + q^{-m_2-1}\mathsf{E}_{11}[m_1]\mathsf{E}_{33}[m_2] +q^{-m-2}\mathsf{E}_{22}[m_1]\mathsf{E}_{33}[m_2]\big),
\end{aligned}
\end{equation}
and $\mathsf{T}_3[m] = \delta_{m, 0}.$ In other words, we have the relation 
\begin{equation}\label{eq:N=3examples2}
\begin{aligned}
     & h\big(\mathsf{E}_{11}[m]+\mathsf{E}_{22}[m] + \mathsf{E}_{33}[m]\big) \\
    &+ h^2\sum_{m_1+m_2=m} \big(\mathsf{E}_{11}[m_1]\mathsf{E}_{22}[m_2] + \mathsf{E}_{11}[m_1]\mathsf{E}_{33}[m_2] + \mathsf{E}_{22}[m_1]\mathsf{E}_{33}[m_2]\big) \\
    &+ h^3\sum_{m_1+m_2+m_3=m}\big(\mathsf{E}_{11}[m_1]\mathsf{E}_{22}[m_2]\mathsf{E}_{33}[m_3]\big) = 0.
    \end{aligned}
\end{equation}
\end{ex}

To investigate the algebra $\widetilde{\Wcal}_h(\sfr\lfr_N)$, it is useful to recall some identities for ordinary and Gaussian binomial coefficients. Let $\binom{N}{r}_q$ denote the Gaussian binomial coefficient, defined by
    \begin{equation} \label{eq:Gaussian binomial}
        \binom{N}{r}_q := \frac{(1-q^N)(1-q^{N-1})\dots (1-q^{N-r+1})}{(1-q)(1-q^2)\dots(1-q^r)}.
    \end{equation}
We can see \eqref{eq:Gaussian binomial} is in $\CC[q]$ from the identity 
\begin{equation*}
    \sum_{0 \le \beta_1 \le \dots \le \beta_m \le n} q^{\beta_1+\dots+\beta_m} = \binom{m+n}{m}_q
\end{equation*}
which follows from the combinatorial interpretation of Gaussian binomial coefficients. We shall also use the $q$-Vandermonde identity
\begin{equation*}
    \binom{m+n}{k}_q = \sum_{j=0}^k \binom{m}{k-j}_q\binom{n}{j}_q q^{j(m-k+j)}.
\end{equation*}
The following lemma describes the form of $\mathsf{T}_r[m]$.
\begin{lem}\label{lem:expansion}
    $\mathsf{T}_r[m]$ has the form
    \begin{equation*}
            \mathsf{T}_r[m]-\delta_{m, 0}\binom{N}{r}_{q^{-1}} = h\mathsf{W}_{r, 1}[m] + h^2\mathsf{W}_{r, 2}[m] +\dots + h^r\mathsf{W}_{r, r}[m],
    \end{equation*}
    for some  $\mathsf{W}_{r, s}[m] \in \Vcal_h(\hfr)$ such that 
    \begin{equation*}
    \lim_{h \to 0} \mathsf{W}_{r, s}[m] = \binom{N-s}{r-s} \sum_{1 \le \alpha_1 < \dots < \alpha_s \le N} \sum_{m_1+\dots + m_s = m} \lim_{h\to 0}(\mathsf{E}_{\alpha_1 \alpha_1}[m_1]) \dots \lim_{h \to 0}(\mathsf{E}_{\alpha_s \alpha_s}[m_s]).
    \end{equation*}
    In particular,
    \begin{equation*}
        \mathsf{W}_{r, 1}[m] = \sum_{\alpha=1}^N \sum_{x=0}^{r-1} q^{(-m-1-x)(\alpha-r+x)}\binom{\alpha-1}{r-x-1}_{q^{-1}}\binom{N-\alpha}{x}_{q^{-1}} \mathsf{E}_{\alpha\alpha}[m].
    \end{equation*}
\end{lem}

\begin{proof}
    From the equality \eqref{eq:T in Sec 5}, we obtain
    \begin{equation*}
        T_r[m] = \sum_{1\le \alpha_1 <\dots < \alpha_r \le N} \sum_{m_1+\dots+m_r=m} \prod_{x=1}^r \left( q^{(-m_x-1)(\alpha_x-x)}\mu_{\alpha_x\alpha_x}[m_x]\right).
    \end{equation*}
    Applying $\iota$, it follows that
    \begin{equation*}
            \mathsf{T}_r[m] = \sum_{1\le \alpha_1 < \dots < \alpha_r \le N}\sum_{m_1+\dots+m_r=m}\prod_{x=1}^r \left( q^{-(\alpha_x-x)}\delta_{m_x, 0} +hq^{(-m_x-1)(\alpha_x-x)}\mathsf{E}_{\alpha_x\alpha_x}[m_x]\right).
    \end{equation*}
    Expanding the product with respect to powers of $h$, we obtain
    \begin{equation*}
        \mathsf{T}_r[m] - \delta_{m, 0}\binom{N}{r}_{q^{-1}}= h\mathsf{W}_{r, 1}[m] + \dots + h^r \mathsf{W}_{r, r}[m].
    \end{equation*}
    for some $\mathsf{W}_{r, s}[m]\in \widetilde{\Vcal}_h(\hfr).$
    Recall that the $q^{-1}$-binomial coefficient $\binom{n}{k}_{q^{-1}}$ is precisely the sum of $q^{-\sum (\gamma_i-i)}$ over all $1 \le \gamma_1 < \dots < \gamma_k \le n$. Applying this, for a fixed $s$, the coefficient $\mathsf{W}_{r, s}[m]$ is obtained by choosing $s$ factors of $q^{(-m_x-1)(\alpha_x-x)}\mathsf{E}_{\alpha_x, \alpha_x}[m_x]$ and $r-s$ factors of $q^{-(\alpha_x-x)}\delta_{m_x, 0}$. Hence, as $h \to 0$,
    \begin{equation*}
        q^{(-m_x-1)(\alpha_x-x)}\mathsf{E}_{\alpha_x\alpha_x}[m_x] \to \lim_{h \to 0} \left( \mathsf{E}_{\alpha_x\alpha_x}[m_x]\right),
    \end{equation*}
    while the remaining coefficient converges to
    \begin{equation*}
        q^{-\sum_x (\alpha_x-x)} =\binom{N-s}{r-s}_{q^{-1}} \to \binom{N-s}{r-s}.
    \end{equation*}
    Hence $\mathsf{W}_{r, s}[m]\in \Vcal_h(\hfr).$
    To compute $\mathsf{W}_{r, 1}[m]$ explicitly, fix $1 \le \alpha \le N$. In the expansion of the product
    \begin{equation*}
        \Big(q^{-(\alpha_1-1)}\delta_{m_1, 0} + hq^{(-m_1-1)(\alpha_1-1)}\mathsf{E}_{\alpha_1 \alpha_1}[m_1]\Big)\cdots\Big(q^{-(\alpha_r-r)}\delta_{m_r, 0} + hq^{(-m_r-1)(\alpha_r-r)}\mathsf{E}_{\alpha_r \alpha_r}[m_r]\Big),
    \end{equation*}
    the term  $\mathsf{E}_{\alpha\alpha}[m]$
    appears by selecting exactly one $h$-factor corresponding to $\alpha$, and $r-1$ constant factors (with respect to $h$) from the rest. Suppose the selected $h$-factor occurs at the $(r-x)$-th position in the product. This implies there are exactly $r-x-1$ chosen indices strictly less than $\alpha$, and $x$ chosen indices strictly greater than $\alpha$. Thus, choosing the sequence $1 \le \alpha_1 < \dots < \alpha_r \le N$ with $\alpha_{r-x} = \alpha$ corresponds to independently making two choices. First, choosing $r-x-1$ indices among $1, \dots, \alpha-1$ yields
    \begin{equation*}
        \binom{\alpha-1}{r-x-1}_{q^{-1}}.
    \end{equation*}
    For the indices strictly greater than $\alpha$, we obtain an additional factor of $q^{-x(\alpha-r+x)}$ from the index shift, yielding
    \begin{equation*}
        q^{-x(\alpha-r+x)}\binom{N-\alpha}{x}_{q^{-1}}.
    \end{equation*}
    Finally, the selected $h$-factor itself contributes $q^{(-m-1)(\alpha-r+x)}$, so combining all these contributions and summing over the possible values of $x$ ($0 \le x \le r-1$) and $\alpha$ ($1 \le \alpha \le N$) yields the desired formula for $\mathsf{W}_{r, 1}[m]$.
\end{proof}

As an application of the above lemma, we prove Proposition~\ref{prop:main}, which was also anticipated in the paper \cite{FR96}. Before stating Proposition~\ref{prop:main}, we introduce an algebra that contains all Fourier modes of the generating series of the classical $W$-algebra.
Recall from \eqref{eq:Miura_nondeformed} that the generators of $\Wbf(\sfr\lfr_N)\hookrightarrow \Vbf(\hfr)$ can be obtained via the Miura transformation. Since we omit the symbol $m$ for the Miura map, throughout this section, denote by $w_2(z),\ldots,w_N(z)$ the generators of $\Wbf(\sfr\lfr_N)$ determined by the identity
        \begin{equation} \label{eq:w(z) Section5}
        \begin{aligned}
            & (\partial+e_{11}(z))(\partial+e_{22}(z))\cdots (\partial+e_{NN}(z))\\
            & =\partial^N+ w_2(z)\partial^{N-2}+ \cdots + w_{N-1}(z) \partial+ w_N(z).
        \end{aligned}
        \end{equation}
     Now we consider the algebra of Fourier modes for $\Wbf(\sfr\lfr_N).$

        \begin{defn} \label{def:W(sl_N)}
            Let $\Wcal(\sfr\lfr_N)$ be the $\CC$-subalgebra of $\Vcal(\hfr)$ generated by the elements $w_r[m]$ for $2 \le r \le N$ and $m \in \ZZ$, where
        \begin{equation*}
            w_r(z) = \sum_{m \in \ZZ} w_r[m]z^{-m-r} \in \Wbf(\sfr\lfr_N).
        \end{equation*}
        \end{defn}

Now we state the proposition. The case $r=1$ of the following proposition provides the base step in the inductive proof of Theorem~\ref{thm:main}.

\begin{prop}\label{prop:main}
    For $1 \le r \le N-1$ and $m \in \ZZ$, we have 
    \begin{equation*}
       \frac{1}{h^2}\ \bigg(\, \mathsf{T}_r[m]-\delta_{m, 0}\binom{N}{r}_{q^{-1}}\bigg) \in \Wcal_h(\sfr\lfr_N)
    \end{equation*}
    and its limit is
    \begin{equation*}
    \sum_{m \in \ZZ} \lim_{h \to 0} \frac{\mathsf{T}_r[m]-\delta_{m, 0}\binom{N}{r}_{q^{-1}}}{h^2}z^{-m-2} = -\binom{N-2}{r-1} w_2(z).
    \end{equation*}
\end{prop}

\begin{proof}
    From \eqref{eq:Miura_nondeformed}, the coefficient of $z^{-m-2}$ in $w_2(z)$ is 
    \begin{equation*}
    \begin{aligned}
 \sum_{1 \le \alpha_1 < \alpha_2 \le N} \sum_{m_1+m_2=m} \lim_{h \to 0} \left(\mathsf{E}_{\alpha_1\alpha_1}[m_1]\mathsf{E}_{\alpha_2\alpha_2}[m_2]\right) + \sum_{\alpha=1}^N(\alpha-1)(-m-1)\lim_{h \to 0}\mathsf{E}_{\alpha\alpha}[m].
    \end{aligned}
    \end{equation*}
    By Lemma~\ref{lem:expansion} and the relation $\mathsf{T}_N[m] = \delta_{m, 0}$,
    \begin{equation*}
        \begin{aligned}
            & \mathsf{T}_r[m]-\delta_{m, 0}\binom{N}{r}_{q^{-1}} \\
            & = \bigg(\mathsf{T}_r[m]-\delta_{m, 0}\binom{N}{r}_{q^{-1}}\bigg)-\binom{N-1}{r-1}_{q^{-1}}\bigg(\mathsf{T}_N[m]-\delta_{m, 0}\bigg)\\
            &=h\bigg( \mathsf{W}_{r, 1}[m] - \binom{N-1}{r-1}_{q^{-1}}\mathsf{W}_{N, 1}[m] \bigg)  + h^2\bigg( \mathsf{W}_{r, 2}[m] - \binom{N-1}{r-1}_{q^{-1}}\mathsf{W}_{N, 2}[m]\bigg) + h^3(\dots).
        \end{aligned}
    \end{equation*}
    Since the coefficient of $h^2$ is 
    \begin{equation*}
    \begin{aligned}
        &\lim_{h \to 0} \bigg(\mathsf{W}_{r, 2}[m] - \binom{N-1}{r-1}_{q^{-1}} \mathsf{W}_{N, 2}[m]\bigg) \\
        &= \left(\binom{N-2}{r-2}-\binom{N-1}{r-1}\right) \sum_{1 \le \alpha_1 < \alpha_2 \le N}\sum_{m_1+m_2=m} \lim_{h \to 0}\big(\mathsf{E}_{\alpha_1 \alpha_1}[m_1]\mathsf{E}_{\alpha_2\alpha_2}[m_2]\big)\\
        &=-\binom{N-2}{r-1} \sum_{1 \le \alpha_1<\alpha_2\le N}\sum_{m_1+m_2=m}\lim_{h \to 0}\big( \mathsf{E}_{\alpha_1\alpha_1}[m_1]\mathsf{E}_{\alpha_2\alpha_2}[m_2]\big),
    \end{aligned}
    \end{equation*}
    it suffices to show that
    \begin{equation*}
        \lim_{h \to 0} \frac{\mathsf{W}_{r, 1}[m] - \binom{N-1}{r-1}_{q^{-1}}\mathsf{W}_{N, 1}[m]}{h} = -\binom{N-2}{r-1} \sum_{\alpha=1}^N (\alpha-1)(-m-1)\lim_{h \to 0} \mathsf{E}_{\alpha\alpha}[m].
    \end{equation*}
    By Lemma~\ref{lem:expansion} and the $q$-Vandermonde identity, we have
    \begin{equation*}
    \begin{aligned}
        &\mathsf{W}_{r, 1}[m] - \binom{N-1}{r-1}_{q^{-1}}\mathsf{W}_{N, 1}[m] \\
        &= \sum_{\alpha=1}^N\sum_{x=0}^{r-1} \big(q^{(-m-x-1)(\alpha-r+x)} - q^{-x(\alpha-r+x)} \big)\binom{\alpha-1}{r-x-1}_{q^{-1}}\binom{N-\alpha}{x}_{q^{-1}} \mathsf{E}_{\alpha\alpha}[m].
    \end{aligned}
    \end{equation*}
    Therefore, we have
    \begin{equation*}
    \begin{aligned}
        &\lim_{h \to 0} \frac{\mathsf{W}_{r, 1}[m] - \binom{N-1}{r-1}_{q^{-1}}\mathsf{W}_{N, 1}[m]}{h} \\
        &=\lim_{h\to 0}\sum_{\alpha=1}^N\sum_{x=0}^{r-1} q^{(-m-x-1)(\alpha-r+x)}(-\alpha+r-x)\binom{\alpha-1}{r-x-1}\binom{N-\alpha}{x} \mathsf{E}_{\alpha\alpha}^{1, -\alpha+r-x}[m] \\
        &=\sum_{\alpha=1}^N\sum_{x=0}^{r-1} (m+1)(\alpha-r+x)\binom{\alpha-1}{r-x-1}\binom{N-\alpha}{x}  \lim_{h \to 0} \mathsf{E}_{\alpha\alpha}[m]\\
        &= (m+1) \sum_{\alpha=1}^N \sum_{x=0}^{r-1} (\alpha-1) \binom{\alpha-2}{r-x-1} \binom{N-\alpha}{x} \lim_{h \to 0} \mathsf{E}_{\alpha\alpha}[m]\\
        &= (m+1) \binom{N-2}{r-1}\sum_{\alpha=1}^N(\alpha-1) \lim_{h \to 0} \mathsf{E}_{\alpha\alpha}[m],
    \end{aligned}
    \end{equation*}
    which implies the proposition.
\end{proof}

In the rest of this section, we explain how to get generators of $\Wcal(\sfr\lfr_N)$ as limits of elements in $\Wcal_h(\sfr\lfr_N).$
Let $N \ge 2$ and consider the new elements $\mathfrak{T}_{r}[m] \in \widetilde{\Wcal}_h(\sfr\lfr_N)$ for $1\leq r \leq N-1, m \in \ZZ$ as follows:
\begin{equation}\label{eq:main}
\begin{aligned}
   \mathfrak{T}_{r}[m]\textstyle &:=\frac{1}{h^{r+1}}\Big(\mathsf{T}_r[m] - q^{-1}\binom{N-r}{1}_{q^{-1}} \mathsf{T}_{r-1}[m]+ \\
   &\quad \dots + (-1)^{r-1}q^{-\frac{(r-1)r}{2}}\binom{N-2}{r-1}_{q^{-1}}\mathsf{T}_1[m] + (-1)^r q^{-\frac{r(r+1)}{2}}\binom{N-1}{r}_{q^{-1}} \delta_{m, 0} - \delta_{m, 0} \Big).
    \end{aligned}
\end{equation}
Using the formal generating series
$ \mathfrak{T}_r(z) := \sum_{m \in \ZZ} \mathfrak{T}_r[m]z^{-m-r-1},$
we can state the main theorem of this paper.
\begin{thm}\label{thm:main}
For $1 \le r \le N-1$ and $m \in \ZZ$, the limit of $\mathfrak{T}_r[m]$ as $h \to 0$ is well defined. Equivalently,
\begin{equation*}
\mathfrak{T}_r[m] \in \Wcal_h(\sfr\lfr_N).
\end{equation*}
Moreover, the limit of the generating series corresponds to a generating series of $\Wbf(\sfr\lfr_N)$:
\begin{equation*}
\lim_{h \to 0} \mathfrak{T}_r(z)
= -w_{r+1}(z)
\in \Wbf(\sfr\lfr_N)
\end{equation*}
and hence
$\lim_{h\to 0}: \Wcal_h(\sfr\lfr_N)\to \Wcal(\sfr\lfr_N)$ is a surjective Poisson algebra homomorphism.
\end{thm}

In the previous section, we justified that $\Vbf_q(H)$ is a $q$-deformation of $\Vbf(\hfr)$. Furthermore, since the above theorem demonstrates that the limit homomorphism recovers $\Wcal(\sfr\lfr_N)$ from $\Wcal_h(\sfr\lfr_N)$, it rigorously establishes that the Poisson algebra $\Wbf_q(LSL_N)$ is indeed a $q$-deformation of $\Wbf(\sfr\lfr_N)$. The following diagram summarizes the limit process from $\Wbf_q(LSL_N)$ to $\Wcal(\sfr\lfr_N)$.
\begin{equation}
\begin{tikzcd}[column sep=large, row sep=large]
\Wbf_q(LSL_N) \arrow[d, dashed, "\text{\scriptsize Fourier modes}"'] & & \Wbf(\sfr\lfr_N) \arrow[d, dashed, "\text{\scriptsize Fourier modes}"]\\
\Wcal_q(LSL_N)
  \arrow[r, hook, "\iota"]
  \arrow[d, hook]
&
\widetilde{\Wcal}_h(\sfr\lfr_N)
  \supset
\Wcal_h(\sfr\lfr_N)
  \arrow[r, two heads, "\lim_{h\to 0}"]
  \arrow[d, hook]
&
\Wcal(\sfr\lfr_N)
  \arrow[d, hook]
\\
\Vcal_q(H)
  \arrow[r, hook, "\iota"]
&
\widetilde{\Vcal}_h(\hfr)
  \supset
\Vcal_h(\hfr)
  \arrow[r, two heads, "\lim_{h\to 0}"]
&
\Vcal(\hfr).
\end{tikzcd}
\end{equation}

Before proceeding with the proof, we outline the overall strategy. Figure~\ref{fig:induction_strategy} below summarizes the induction steps. Note that the base case $r=1$ for all $N \ge 2$ has already been established in Proposition~\ref{prop:main} (represented by the blue dots). Next, we prove the case $r=N-1$ for all $N \ge 2$ (represented by the red dots). Finally, for $r \ge 2$, assuming the statement holds for a fixed $N$ at both $r-1$ and $r$, we prove the statement for $(N+1, r)$. This inductive step is indicated by the black arrows. This covers all $1 \le r \le N-1$.

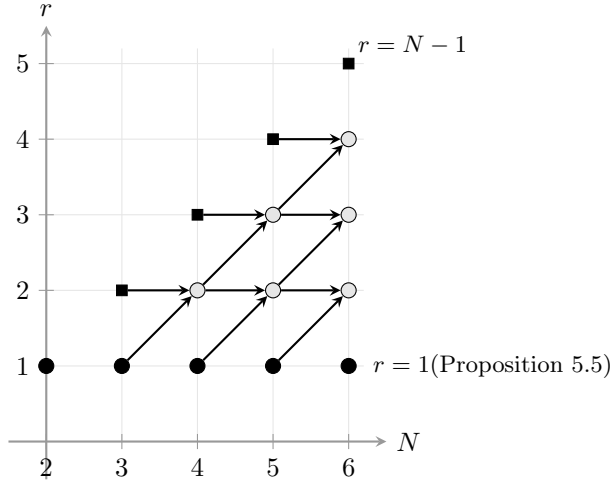
\begin{figure}[htbp]
    \centering
    \hskip 3cm
    \begin{tikzpicture}[scale=1, >=stealth]
        
        \foreach \n in {2, 3, 4, 5, 6} {
            \draw[gray!20, thin] (\n, 0) -- (\n, 5.2);
        }
        \foreach \r in {1, 2, 3, 4, 5} {
            \draw[gray!20, thin] (2, \r) -- (6.2, \r);
        }

        \draw[->, thick, gray!80] (1.5, 0) -- (6.5, 0) node[right, black] {$N$};
        \draw[->, thick, gray!80] (2, -0.5) -- (2, 5.5) node[above, black] {$r$};
        
        \foreach \n in {2, 3, 4, 5, 6} {
            \draw[gray!80] (\n, 0.1) -- (\n, -0.1) node[below, black] {$\n$};
        }
        \foreach \r in {1, 2, 3, 4, 5} {
            \draw[gray!80] (2.1, \r) -- (1.9, \r) node[left, black] {$\r$};
        }

        \foreach \n in {2, 3, 4, 5, 6} {
            \node[circle, fill=black!100, draw=black, inner sep=2pt] (N\n_r1) at (\n, 1) {};
        }
        \node[right, black, font=\small] at (6.2, 1) {$r=1$(Proposition~\ref{prop:main})};

        \foreach \n in {3, 4, 5, 6} {
           \node[rectangle, fill=black!100, draw=black, inner sep=2pt] (N\n_r\the\numexpr\n-1\relax) at (\n, \n-1) {};
        }
        
        \node[above right, black, font=\small] at (6, 5) {$r=N-1$};

        \node[circle, fill=black!10, draw=black, inner sep=2pt] (N4_r2) at (4, 2) {};

        \node[circle, fill=black!10, draw=black, inner sep=2pt] (N5_r2) at (5, 2) {};
        \node[circle, fill=black!10, draw=black, inner sep=2pt] (N5_r3) at (5, 3) {};

        \node[circle, fill=black!10, draw=black, inner sep=2pt] (N6_r2) at (6, 2) {};
        \node[circle, fill=black!10, draw=black, inner sep=2pt] (N6_r3) at (6, 3) {};
        \node[circle, fill=black!10, draw=black, inner sep=2pt] (N6_r4) at (6, 4) {};

        \draw[->, thick, black] (N3_r1) -- (N4_r2);
        \draw[->, thick, black] (N3_r2) -- (N4_r2);
        
        \draw[->, thick, black] (N4_r2) -- (N5_r3);
        \draw[->, thick, black] (N4_r3) -- (N5_r3);

        \draw[->, thick, black] (N4_r1) -- (N5_r2);
        \draw[->, thick, black] (N4_r2) -- (N5_r2);
        
        \draw[->, thick, black] (N5_r4) -- (N6_r4);
        \draw[->, thick, black] (N5_r3) -- (N6_r4);
        
        \draw[->, thick, black] (N5_r3) -- (N6_r3);
        \draw[->, thick, black] (N5_r2) -- (N6_r3);
        
        \draw[->, thick, black] (N5_r2) -- (N6_r2);
        \draw[->, thick, black] (N5_r1) -- (N6_r2);

    \end{tikzpicture}
    \caption{Strategy of induction steps.}
    \label{fig:induction_strategy}
\end{figure}

\noindent
Since the induction step depends on $N$, we make the dependence on $N$ explicit in the notation. We denote $\Vbf(\hfr)$ and $\Vbf_q(H)$ by $\Vbf(\hfr^N)$ and $\Vbf_q(H^N)$, respectively. Recall the definition of $w_r(z) \in \Wbf(\sfr\lfr_N)$ in \eqref{eq:w(z) Section5} and $T_r(z) \in \Wbf_q(LSL_N)$ in \eqref{eq:T in Sec 5}. We denote the generators $w_r(z)$ and $T_r(z)$ by $w_r^N(z)$ and $T_r^N(z)$, respectively. In particular, when comparing $w_r^N(z)$ and $T_r^N(z)$ with $w_r^{N+1}(z)$ and $T_r^{N+1}(z)$, we shift the indices
\begin{equation*}
    e_{11}(z), \dots, e_{NN}(z), \quad \mu_{11}(z), \dots, \mu_{NN}(z)
\end{equation*}
to
\begin{equation*}
    e_{22}(z), \dots, e_{N+1, N+1}(z), \quad \mu_{22}(z), \dots, \mu_{N+1, N+1}(z).
\end{equation*}
The following lemma describes the relations between $w_r^{N+1}(z)$ and $w_r^N(z)$, and between $T_r^{N+1}(z)$ and $T_r^N(z)$.
\begin{lem}\label{lem:indicesshift}
For $N \ge 2$ and $2 \le r \le N$,
\begin{equation}\label{eq:WNN+1relation}
    \begin{aligned}
        w_{1}^{N+1}(z) &= w_1^N(z) + e_{11}(z),\\
        w_{r}^{N+1}(z) &= w_r^N(z) + e_{11}(z)w_{r-1}^N(z) + \partial_z w_{r-1}^N(z) \quad (r=2, 3, \dots, N),\\
        w_{N+1}^{N+1}(z) &= e_{11}(z)w_{N}^N(z) + \partial_zw_N^N(z)
        \end{aligned}
        \end{equation}
        in the algebra $\Vbf(\hfr^{N+1})$, and
        \begin{equation}\label{eq:TNN+1relation}
        \begin{aligned}
        T_{1}^{N+1}(z) &= T_1^N(qz) + \mu_{11}(z),\\
        T_{r}^{N+1}(z) &= T_r^N(qz) + \mu_{11}(z)T_{r-1}^N(z) \quad (r=2, 3, \dots, N),\\
        T_{N+1}^{N+1}(z) &= \mu_{11}(z)T_{N}^{N}(z)
    \end{aligned}
\end{equation}
in the algebra $\Vbf_q(H^{N+1})$. In particular, comparing the Fourier coefficients,
\begin{equation}\label{eq:TNN+1relation2}
\begin{aligned}
    \mathsf{T}_1^{N+1}[m] &= q^{-m-1}\mathsf{T}_1^{N}[m] + h\mathsf{E}_{11}[m] + \delta_{m, 0},\\
    \mathsf{T}_{r}^{N+1}[m] &= q^{-m-r}\mathsf{T}_{r}^{N}[m] + h\sum_{m_1+m_2=m}\mathsf{E}_{11}[m_1]\mathsf{T}_{r-1}^N[m_2] + \mathsf{T}_{r-1}^N[m], \\
    \mathsf{T}_{N+1}^{N+1}[m] &= h\sum_{m_1+m_2 = m} \mathsf{E}_{11}[m_1]\mathsf{T}_{N}^N[m_2] + \mathsf{T}_{N}^N[m]
    \end{aligned}
\end{equation}
in the algebra $\widetilde{\Vcal}_h(\hfr^{N+1})$.
\end{lem}
\begin{proof}
We use the column determinant formulas
\begin{equation*}
\begin{aligned}
    &(\partial + e_{11}(z))\dots (\partial+e_{NN}(z))(\partial + e_{N+1, N+1}(z)) \\
    &\quad = (\partial + e_{11}(z))(\partial^N + w_1^N(z)\partial^{N-1} + \dots + w_{N-1}^N(z)\partial + w_{N}^N(z))\\
    &\quad = \partial^{N+1} + w_1^{N+1}(z)\partial^{N} + \dots + w_{N}^{N+1}(z)\partial + w_{N+1}^{N+1}(z),
\end{aligned}
\end{equation*}
and
\begin{equation*}
\begin{aligned}
    &(D + \mu_{11}(z))\dots (D + \mu_{NN}(z))(D+ \mu_{N+1, N+1}(z)) \\
    &\quad = (D + \mu_{11}(z))(D^{N} + T_1^{N}(z)D^{N-1} + \dots + T_{N-1}^{N}(z)D + T_N^N(z))\\
    &\quad = D^{N+1} + T_1^{N+1}(z)D^{N} + \dots + T_{N}^{N+1}(z)D +T_{N+1}^{N+1}(z).
    \end{aligned}
\end{equation*}
Comparing the coefficients of the powers of $\partial$ and $D$ proves the lemma.
\end{proof}

\begin{proof}[Proof of Theorem~\ref{thm:main}]
   Recall the definition of $\mathfrak{T}_r[m]$ in \eqref{eq:main}. Since
    \begin{equation}
        -q^{-1}\binom{N-1}{1}_{q^{-1}}-1 = -\binom{N}{1}_{q^{-1}},
    \end{equation}
    the case $r=1$ has already been proved in Proposition~\ref{prop:main} for all $N \ge 2$.
    
    We now assume the theorem is true for $N$ and $r=N-1$, and prove the statement for $N+1$ and $r=N$.
    After shifting the indices as in Lemma~\ref{lem:indicesshift}, we have
\begin{equation*}
    \lim_{h \to 0} \mathfrak{T}_{N-1}^{N}(z)=-w_{N}^N(z). 
\end{equation*}
Applying \eqref{eq:TNN+1relation2},
\begin{equation}\label{eq:mainproof1}
    h^{N+1}\mathfrak{T}_{N}^{N+1}[m] =- h^N q^{-m-N}\mathfrak{T}_{N-1}^{N}[m] +\sum_{m_1+m_2=m} h^{N+1}\mathsf{E}_{11}[m_1]\mathfrak{T}_{N-1}^N[m_2] + h^N\mathfrak{T}_{N-1}^N[m] .
\end{equation}
The detailed verification of this equality is analogous to the computation in \eqref{eq:mainproof2}, and therefore omitted. Divide \eqref{eq:mainproof1} by $h^{N+1}$ and take the limit $h \to 0$, we obtain
\begin{equation*}
\begin{aligned}
    \lim_{h \to 0} \mathfrak{T}_{N}^{N+1}(z) &=\lim_{h \to 0} \left(\mathsf{E}_{11}(z)\mathfrak{T}_{N-1}^N(z)\right) +\sum_{m \in \ZZ}\lim_{h \to 0} \left(\frac{\mathfrak{T}_{N-1}^N[m] - q^{-m-N}\mathfrak{T}_{N-1}^N[m]}{h}z^{-m-N-1} \right) \\
    &= -\lim_{h \to 0}\big(\mathsf{E}_{11}(z)\big)w_{N}^{N}(z) - \partial_z w_{N}^{N}(z) = -w_{N+1}^{N+1}(z).
    \end{aligned}
\end{equation*}
Now, fix $N \ge 3$ and $1 \le r \le N-1$. After shifting the indices as in Lemma~\ref{lem:indicesshift}, we assume
\begin{equation*}
    \lim_{h \to 0} \mathfrak{T}_{r-1}^N(z) = -w_{r}^N(z), \quad \lim_{h \to 0} \mathfrak{T}_{r}^N(z) = -w_{r+1}^N(z).
\end{equation*}
After applying \eqref{eq:TNN+1relation},
\begin{equation}\label{eq:mainproof2}
\begin{aligned}
    & h^{r+1}\mathfrak{T}_{r}^{N+1}[m]= q^{-m-r}\mathsf{T}_r^N[m] + h\sum_{m_1+m_2=m}\mathsf{E}_{11}[m_1]\mathsf{T}_{r-1}^N[m_2] + \mathsf{T}_{r-1}^N[m] \\
    &\quad -q^{-1}\binom{N-r+1}{1}_{q^{-1}} \Big( q^{-m-r+1}\mathsf{T}_{r-1}^N[m] + h\sum_{m_1+m_2=m}\mathsf{E}_{11}[m_1]\mathsf{T}_{r-2}^N[m_2] + \mathsf{T}_{r-2}^N[m] \Big) \\
    &\quad + \dots \\
    &\quad +(-1)^{r-2}q^{\frac{-(r-2)(r-1)}{2}} \binom{N-2}{r-2}_{q^{-1}}\Big( q^{-m-2}\mathsf{T}_2^N[m] +h\sum_{m_1+m_2=m}\mathsf{E}_{11}[m_1]\mathsf{T}_1^N[m_2] + \mathsf{T}_1^N[m] \Big) \\
    &\quad + (-1)^{r-1}q^{-\frac{(r-1)r}{2}} \binom{N-1}{r-1}_{q^{-1}} \Big( q^{-m-1}\mathsf{T}_1^N[m] + h\mathsf{E}_{11}[m] + \delta_{m, 0} \Big) \\
    &\quad + (-1)^r q^{-\frac{r(r+1)}{2}} \binom{N}{r}_{q^{-1}} - \mathsf{T}_N^N[m] - h\sum_{m_1+m_2=m}\mathsf{E}_{11}[m_1] \mathsf{T}_N^N[m_2]\\
    &=q^{-m-r}\Big(h^{r+1}\mathfrak{T}_{r}^N[m] -h^r\mathfrak{T}_{r-1}^N[m]\Big) +h^{r+1} \sum_{m_1+m_2=m} \mathsf{E}_{11}[m_1]\mathfrak{T}_{r-1}^N[m_2] + h^r\mathfrak{T}_{r-1}^N[m].
    \end{aligned}
\end{equation}
Divide \eqref{eq:mainproof2} by $h^{r+1}$ and take the limit $h \to 0$, we obtain
\begin{equation*}
\begin{aligned}
    \lim_{h \to 0}\mathfrak{T}_r^{N+1}(z) &= \lim_{h \to 0} \left(q^{-m-r}\mathfrak{T}_r^N(z) + \mathsf{E}_{11}(z) \mathfrak{T}_{r-1}^N(z)\right) \\
    &\quad + \lim_{h\to 0}\left(\frac{\mathfrak{T}_{r-1}^N[m]-q^{-m-r}\mathfrak{T}_{r-1}^N[m]}{h}z^{-m-r-1} \right) \\
    &= -w_{r+1}^N(z) - \lim_{h \to 0}\big(\mathsf{E}_{11}(z)\big)w_{r}^N(z) -\partial_zw_{r}^N(z)= -w_{r+1}^{N+1}(z).
\end{aligned}
\end{equation*}
Since each $\mathsf{T}_{s}[m]$ is a $\CC(\!(h)\!)$-linear combination of the elements $\mathfrak{T}_r[m]$'s, it follows that $\widetilde{\Wcal}_h(\sfr\lfr_N)$ is generated by the $\mathfrak{T}_r[m]$. Therefore, the map
$
\lim_{h\to 0} : \Wcal_h(\sfr\lfr_N)\to \Wcal(\sfr\lfr_N)$
is well defined.
Moreover, since both $\widetilde{\Wcal}_h(\sfr\lfr_N)$ and $\Vcal_h(\hfr)$ are Poisson algebras, their intersection $\Wcal_h(\sfr\lfr_N)$ is closed under the Poisson bracket. Finally, because
$
\lim_{h\to 0}:\Vcal_h(\hfr)\to \Vcal(\hfr)
$
is a Poisson algebra homomorphism, its restriction
$
\lim_{h\to 0}:\Wcal_h(\sfr\lfr_N)\to \Wcal(\sfr\lfr_N)
$
is also a Poisson algebra homomorphism.
\end{proof}

\begin{ex}
    Recall the equations \eqref{eq:N=3examples1} and \eqref{eq:N=3examples2}. We have 
    \begin{equation*}
    \begin{aligned}
        \frac{\mathsf{T}_{1}[m]-\delta_{m, 0}\binom{3}{1}_{q^{-1}}}{h^2} &=\frac{\left(\mathsf{T}_1[m] - \delta_{m, 0}\binom{3}{1}_{q^{-1}}\right) -\left(\mathsf{T}_3[m]-\delta_{m, 0}\right)}{h^2}\\
        &=\frac{1}{h}\Big( (q^{-m-1}-1)\mathsf{E}_{22}[m] +(q^{-2m-2}-1) \mathsf{E}_{33}[m]\Big)\\
        &\quad -\sum_{m_1+m_2=m}\Big(\mathsf{E}_{11}[m_1]\mathsf{E}_{22}[m_2] + \mathsf{E}_{11}[m_1]\mathsf{E}_{33}[m_2] + \mathsf{E}_{22}[m_1]\mathsf{E}_{33}[m_2]\Big)\\
        &\quad -h\sum_{m_1+m_2+m_3=m}\Big(\mathsf{E}_{11}[m_1]\mathsf{E}_{22}[m_2]\mathsf{E}_{33}[m_3]\Big)\\
        &=-\Bigg(\mathsf{E}_{22}^{1, 1}[m] + 2\mathsf{E}_{33}^{1, 2}[m] \\
        &\quad +\sum_{m_1+m_2=m}\Big(\mathsf{E}_{11}[m_1]\mathsf{E}_{22}[m_2] + \mathsf{E}_{11}[m_1]\mathsf{E}_{33}[m_2] + \mathsf{E}_{22}[m_1]\mathsf{E}_{33}[m_2]\Big)\Bigg) \\
        &\quad -h\sum_{m_1+m_2+m_3=m}\Big(\mathsf{E}_{11}[m_1]\mathsf{E}_{22}[m_2]\mathsf{E}_{33}[m_3]\Big),
    \end{aligned}
    \end{equation*}
    where $\mathsf{E}_{ij}^{k,l}$ is defined in \eqref{eq:E_{ij}^{kl}}.
    Hence we conclude $\frac{\mathsf{T}_{1}[m]-\delta_{m, 0}\binom{3}{1}_{q^{-1}}}{h^2} \in  \Vcal_h(\sfr\lfr_3)$. Similarly, since 
    \begin{equation*}
    \begin{aligned}
        \frac{\mathsf{T}_{2}[m]-\delta_{m, 0}\binom{3}{2}_{q^{-1}}}{h^2} &= \frac{\left(\mathsf{T}_2[m] - \delta_{m, 0}\binom{3}{2}_{q^{-1}}\right) - \binom{2}{1}_{q^{-1}}\left( \mathsf{T}_3[m]- \delta_{m, 0}\right)}{h^2}\\
        &=\frac{1}{h}\Big( (q^{-m-2}-q^{-1})\mathsf{E}_{22}[m] +(q^{-m-2}+q^{-m-1}-q^{-1}-1) \mathsf{E}_{33}[m]\Big)\\
        &\quad +\sum_{m_1+m_2=m}\Big(\mathsf{E}_{11}[m_1]\mathsf{E}_{22}[m_2] + q^{-m_2-1}\mathsf{E}_{11}[m_1]\mathsf{E}_{33}[m_2] +q^{-m-2} \mathsf{E}_{22}[m_1]\mathsf{E}_{33}[m_2]\Big)\\
        &\quad -(1+q^{-1}) \sum_{m_1+m_2=m}\Big( \mathsf{E}_{11}[m_1]\mathsf{E}_{22}[m_2] + \mathsf{E}_{11}[m_1]\mathsf{E}_{33}[m_2] + \mathsf{E}_{22}[m_1]\mathsf{E}_{33}[m_2]\Big)\\
        &\quad -h(1+q^{-1})\sum_{m_1+m_2+m_3=m}\Big(\mathsf{E}_{11}[m_1]\mathsf{E}_{22}[m_2]\mathsf{E}_{33}[m_3]\Big)\\
        &= \Bigg(q^{-1}\mathsf{E}_{22}^{1, 1}[m] + (q^{-1}+1)\mathsf{E}_{33}^{1, 1}[m]\\
        &\quad +\sum_{m_1+m_2=m}\Big(\mathsf{E}_{11}[m_1]\mathsf{E}_{22}[m_2] + q^{-m_2-1}\mathsf{E}_{11}[m_1]\mathsf{E}_{33}[m_2] +q^{-m-2} \mathsf{E}_{22}[m_1]\mathsf{E}_{33}[m_2]\Big)\\
        &\quad -(1+q^{-1}) \sum_{m_1+m_2=m}\Big( \mathsf{E}_{11}[m_1]\mathsf{E}_{22}[m_2] + \mathsf{E}_{11}[m_1]\mathsf{E}_{33}[m_2] + \mathsf{E}_{22}[m_1]\mathsf{E}_{33}[m_2]\Big)\Bigg)\\
        &\quad -h(1+q^{-1})\sum_{m_1+m_2+m_3=m}\Big(\mathsf{E}_{11}[m_1]\mathsf{E}_{22}[m_2]\mathsf{E}_{33}[m_3]\Big),
    \end{aligned}
    \end{equation*}
    we also conclude $\frac{\mathsf{T}_{2}[m]-\delta_{m, 0}\binom{3}{2}_{q^{-1}}}{h^2} \in \Vcal_h(\sfr\lfr_3)$. Therefore, we obtain
    \begin{equation*}
    \begin{aligned}
        \sum_{m \in \ZZ} \lim_{h \to 0} \frac{\mathsf{T}_1[m]- \delta_{m, 0}\binom{3}{1}_{q^{-1}}}{h^2} z^{-m-2} &= -w_2(z),\\
        \sum_{m \in \ZZ} \lim_{h \to 0} \frac{\mathsf{T}_2[m]- \delta_{m, 0}\binom{3}{2}_{q^{-1}}}{h^2} z^{-m-2} &= -w_2(z),
    \end{aligned}
    \end{equation*}
    which is consistent with Proposition~\ref{prop:main}. Moreover,
    \begin{equation*}
    \begin{aligned}
        \mathfrak{T}_2[m] &= \frac{1}{h^3} \left( \mathsf{T}_2[m] - q^{-1} \mathsf{T}_1[m] + q^{-3}\delta_{m, 0} - \mathsf{T}_3[m] \right)\\
        &= \frac{1}{h^2}  \left( q^{-m-1}+q^{-m-2}-q^{-2m-3}-1\right) \mathsf{E}_{33}[m]\\
        &\quad + \frac{1}{h} \sum_{m_1+m_2=m}\left( \left(q^{-m_2-1}-1\right)\mathsf{E}_{11}[m_1]\mathsf{E}_{22}[m_2] + \left(q^{-m-2}-1\right)\mathsf{E}_{22}[m_1]\mathsf{E}_{33}[m_2]\right)\\
        &\quad -\sum_{m_1+m_2+m_3=m} \mathsf{E}_{11}[m_1]\mathsf{E}_{22}[m_2]\mathsf{E}_{33}[m_3]\\
        &= -\Bigg(\mathsf{E}_{33}^{2, 1}[m] +\sum_{m_1+m_2=m}\left( \mathsf{E}_{11}[m_1] \mathsf{E}_{22}^{1, 1}[m_2] + \mathsf{E}_{22}^{1, 1}[m_1]\mathsf{E}_{33}^{1, 1}[m_2]\right)\\
        &\quad + \sum_{m_1+m_2+m_3 = m} \mathsf{E}_{11}[m_1]\mathsf{E}_{22}[m_2]\mathsf{E}_{33}[m_3])\Bigg)\\
        &\quad + \sum_{n \ge 1} h^n\left( \frac{1}{n!}\sum_{k=0}^{n-1} \binom{n}{k} (m+1)^k \mathsf{E}_{33}[m] \right) \in \Vcal_h(\sfr\lfr_3).
    \end{aligned}
    \end{equation*}
    Therefore,
    \begin{equation*}
            \lim_{h \to 0}\mathfrak{T}_2(z) = -w_3(z)
    \end{equation*}
    which is consistent with Theorem~\ref{thm:main}.
\end{ex}

\appendix
\section{Classical Yang-Baxter equation}\label{AppendixA}

In this Appendix, we summarize the necessary background on the (modified) classical Yang-Baxter equation. Further details can be found in \cites{BD82, BD84, CP94}.

Let $\gfr$ be a complex Lie algebra endowed with a nondegenerate symmetric invariant bilinear form $\langle \cdot, \cdot \rangle$. Given an element $r \in \gfr \otimes \gfr$, we define a $3$-tensor
    \begin{equation*}
        [[r, r]] := [r^{12}, r^{13}] + [r^{12}, r^{23}] + [r^{13}, r^{23}],
    \end{equation*}
where $[(X \otimes Y)^{12}, (Z \otimes W)^{13}] = [X, Z] \otimes Y \otimes W$, $[(X \otimes Y)^{12}, (Z \otimes W)^{23}] = X \otimes [Y, Z] \otimes W$, and $[(X \otimes Y)^{13}, (Z \otimes W)^{23}] = X \otimes Z \otimes [Y, W]$ for $X, Y, Z, W \in \gfr$. We also define a linear map $R_+ : \gfr \to \gfr$ determined by the relation
\begin{equation*}
    \langle R_+(X), Y \rangle = \langle r, X \otimes Y \rangle, \quad X, Y \in \gfr,
\end{equation*}
and define $R_- := -R_+^\ast$.
\begin{defn}
    The Classical Yang-Baxter Equation (CYBE) is the equation
    \begin{equation}\label{eq:CYBE}
        [[r, r]] = 0.
    \end{equation}
    An element $r \in \gfr \otimes \gfr$ is called a classical $r$-matrix if $r$ satisfies the CYBE \eqref{eq:CYBE}. In terms of $R_\pm$, the CYBE has the form
    \begin{equation}\label{eq:CYBE2}
        [R_+(X), R_+(Y)] = R_+\big([R_+(X), Y] + [X, R_-(Y)]\big), \quad X, Y \in \gfr.
    \end{equation}
    Moreover, a skew-symmetric linear operator $R: \gfr \to \gfr$ is said to be a solution of the modified Classical Yang-Baxter Equation (mCYBE) if $R$ satisfies
    \begin{equation}\label{eq:mCYBE}
        [R(X), R(Y)] = R\big([R(X), Y] + [X, R(Y)]\big)-[X, Y], \quad X, Y \in \gfr.
    \end{equation}
\end{defn}

For the remainder of this appendix, we assume that $\gfr$ admits a decomposition into Lie subalgebras
\begin{equation*}
\gfr = \nfr_+ \oplus \hfr \oplus \nfr_-
\end{equation*}
such that $[\hfr, \nfr_\pm] \subseteq \nfr_\pm$. Moreover, we impose two additional conditions on $\langle \cdot, \cdot \rangle$:
\begin{enumerate}[(1)]
    \item $P_+^\ast = P_-$, where $P_\pm$ are the projections onto $\nfr_\pm$, respectively.
    \item The restriction of $\langle \cdot, \cdot\rangle$ to $\hfr$ is nondegenerate.
\end{enumerate}

\begin{lem}\label{lemma:AppendixA1}
    Suppose a linear operator $R_0 : \hfr \to \hfr$ satisfies
    \begin{equation}\label{eq:cartanconditions}
    \begin{aligned}
        R_0 + R_0^\ast &= I,\\
        [R_0(X), R_0(Y)] &= R_0\big([R_0(X), Y]-[X, R_0^\ast(Y)]\big), \quad X, Y \in \hfr,
    \end{aligned}
    \end{equation}
    where $I$ is the identity map on $\hfr$, and $R_0^\ast$ is the adjoint operator of $R_0$ with respect to the restriction of $\langle \cdot, \cdot \rangle$ to $\hfr$. Then the operators $R_+$ and $R_-$, defined by
    \begin{equation*}
        R_+ = P_+ + R_0P_0, \quad R_- = -P_- - R_0^\ast P_0,
    \end{equation*}
    satisfy \eqref{eq:CYBE2}, where $P_0$ is the projection onto $\hfr$.
\end{lem}
\begin{proof}
    Let $X = X_+ + X_0 + X_-, Y = Y_+ + Y_0 + Y_- \in \gfr = \nfr_+ \oplus \hfr \oplus \nfr_-$. Then we have
    \begin{equation*}
    \begin{aligned}
        R_+\big([R_+(X), Y] + [X, R_-(Y)]\big) &= R_+\big([X_++R_0(X_0), Y] + [X, - Y_- - R_0^\ast(Y_0)]\big) \\
        &= R_+\big([X_+, Y_+] + [X_+, Y_0] +[X_+, Y_-]\big)\\
        & \quad + R_+\big([R_0(X_0), Y_+] + [R_0(X_0), Y_0] + [R_0(X_0), Y_-]\big) \\
        & \quad - R_+\big([X_+, Y_-] + [X_0, Y_-] + [X_-, Y_-]\big) \\
        & \quad - R_+\big([X_+, R_0^\ast(Y_0)]+[X_0, R_0^\ast(Y_0)]+[X_-, R_0^\ast(Y_0)]\big)\\
        &= [X_+, Y_+] + [X_+, Y_0] + [R_0(X_0), Y_+]\\
        & \quad + R_0\big([R_0(X_0), Y_0]\big) -[X_+, R_0^\ast(Y_0)]-R_0\big([X_0, R_0^\ast(Y_0)]\big)\\
        &= [X_+, Y_+] + [X_+, R_0(Y_0)] + [R_0(X_0), Y_+] \\
        & \quad + R_0\big([R_0(X_0), Y_0] - [X_0, R_0^\ast(Y_0)]\big).
    \end{aligned}
    \end{equation*}
    Using the condition \eqref{eq:cartanconditions}, this simplifies to
    \begin{equation*}
        \begin{aligned}
        & [X_+, Y_+] + [X_+, R_0(Y_0)] + [R_0(X_0), Y_+]+ [R_0(X_0), R_0(Y_0)]\\
        &= [X_+ + R_0(X_0), Y_+ + R_0(Y_0)] = [R_+(X), R_+(Y)].
        \end{aligned}
    \end{equation*}
\end{proof}

A linear operator $R_0$ satisfying \eqref{eq:cartanconditions} can be constructed as follows:
\begin{lem}\label{lemma:AppendixA2}
    If a Lie algebra homomorphism $\theta : \hfr \to \hfr$ satisfies
    \begin{align}
        &\det(I-\theta) \ne 0;\label{eq:thetaconditionsAppendix1}\\
        &\langle\theta(X), \theta(Y) \rangle = \langle X, Y\rangle, \quad X, Y \in \hfr, \label{eq:thetaconditionsAppendix2}
    \end{align}
    then the operator
    \begin{equation*}
        R_0 = \frac{I}{I-\theta}
    \end{equation*}
    satisfies \eqref{eq:cartanconditions}.
\end{lem}
\begin{proof}
For any $X, Y \in \hfr$, we have
    \begin{equation*}
        \begin{aligned}
            \left\langle \frac{I}{I-\theta}X, Y\right\rangle &= \left\langle \frac{I}{I-\theta}X, \frac{I-\theta}{I-\theta}Y\right\rangle = \left\langle \frac{I}{I-\theta}X, \frac{I}{I-\theta}Y \right\rangle -\left\langle \frac{I}{I-\theta}X, \frac{\theta}{I-\theta}Y \right\rangle \\
            &= \left\langle \frac{\theta}{I-\theta}X, \frac{\theta}{I-\theta}Y \right\rangle -\left\langle \frac{I}{I-\theta}X, \frac{\theta}{I-\theta}Y \right\rangle = -\left\langle X, \frac{\theta}{I-\theta}Y \right\rangle,
        \end{aligned}
    \end{equation*}
    which implies
    \begin{equation*}
        R_0 + R_0^\ast = \frac{I}{I-\theta} + \left(\frac{I}{I-\theta}\right)^\ast = \frac{I}{I-\theta} - \frac{\theta}{I-\theta} = I.
    \end{equation*}
    Furthermore,
    \begin{equation*}
        \begin{aligned}
            R_0\big([R_0(X), Y] - [X, R_0^\ast(Y)]\big) &= \frac{I}{I-\theta}\left(\left[\frac{I}{I-\theta}X, Y\right] - \left[X, \left(-\frac{\theta}{I-\theta}\right)Y\right]\right) \\
            &= \frac{I}{I-\theta}\left[ \frac{I}{I-\theta}X, \frac{I-\theta}{I-\theta}Y\right] + \frac{I}{I-\theta}\left[\frac{I-\theta}{I-\theta}X, \frac{\theta}{I-\theta}Y\right] \\
            &= \frac{I}{I-\theta}\left[\frac{I}{I-\theta}X, \frac{I}{I-\theta}Y\right]-\frac{I}{I-\theta}\left[\frac{I}{I-\theta}X, \frac{\theta}{I-\theta}Y\right] \\
            & \quad + \frac{I}{I-\theta}\left[\frac{I}{I-\theta}X, \frac{\theta}{I-\theta}Y\right] -\frac{I}{I-\theta}\left[\frac{\theta}{I-\theta}X, \frac{\theta}{I-\theta}Y\right].
        \end{aligned}
    \end{equation*}
    Since $\theta$ is a Lie algebra homomorphism, this simplifies to
    \begin{equation*}
        \frac{I}{I-\theta}\left[\frac{I}{I-\theta}X, \frac{I}{I-\theta}Y\right]-\frac{\theta}{I-\theta}\left[\frac{I}{I-\theta}X, \frac{I}{I-\theta}Y\right]= \left[ \frac{I}{I-\theta}X, \frac{I}{I-\theta}Y\right] = [R_0(X), R_0(Y)].
    \end{equation*}
\end{proof}

Given a Lie algebra homomorphism $\theta : \hfr \to \hfr$ with  \eqref{eq:thetaconditionsAppendix1} and \eqref{eq:thetaconditionsAppendix2}, let
\begin{equation*}
R_+ = P_+ + \frac{I}{I-\theta}P_0, \quad R_-=-P_- + \frac{\theta}{I-\theta}P_0, \quad \text{and} \quad R= R_+ + R_-.
\end{equation*}
Then it is straightforward to verify that
\begin{equation*}
    [R_-(X), R_-(Y)] = R_-\big([R_-(X), Y]+[X, R_+(Y)]\big)
\end{equation*}
and
\begin{equation*}
    [R(X), R(Y)] = R\big([R(X), Y]+[X, R(Y)]\big)-[X, Y].
\end{equation*}
Thus, $R$ satisfies the mCYBE \eqref{eq:mCYBE}.

\section{Poisson-Lie groups} \label{AppendixB}
This appendix provides some relevant concepts from the geometry of Poisson-Lie groups. In the first subsection, we briefly summarize the geometric motivation underlying the Poisson bracket \eqref{eq:PoissonBracketDef}. In the second subsection, we provide a purely algebraic proof for the well-definedness of the Poisson bracket \eqref{eq:PoissonBracketDef}, which was originally established within a geometric framework in \cite{S85}.

\subsection{Geometric Origin} \label{subsec:B1}
Let $G$ be a Poisson-Lie group with tangent Lie bialgebra $\gfr$, endowed with a nondegenerate bilinear form $\langle \cdot, \cdot\rangle$. The left and right gradients are defined as follows:
     \begin{equation*}
     \langle \nabla\varphi(g), X \rangle = \left.\frac{d}{ds}\right|_{s=0} \varphi(e^{sX}g), \quad \langle \nabla'\varphi(g), X \rangle = \left.\frac{d}{ds}\right|_{s=0}\varphi(ge^{sX}) , \qquad \varphi \in C^\infty(G), g\in G, X \in \gfr .
     \end{equation*}
For example, let $G = LSL_N$ be the group of $N \times N$ matrices with entries in $\CC[t, t^{-1}]$ and determinant $1$. For an element
\begin{equation*}
    g(t) = \begin{pmatrix}
    g_{11}(t) & g_{12}(t) & \dots & g_{1N}(t) \\
    g_{21}(t) & g_{22}(t) & \dots & g_{2N}(t) \\
    \vdots & \vdots & \ddots & \vdots \\
    g_{N1}(t) & g_{N2}(t) & \dots & g_{NN}(t)
\end{pmatrix} \in G, \quad g_{ij}(t) = \sum_{n \in \ZZ} g_{ij}[n] t^{-n},
\end{equation*}
we can compute the gradients of the functionals $\mu_{ij}[n] \in C^\infty(G)$, given by $\mu_{ij}[n](g(t)) = g_{ij}[n]$, as follows:
\begin{equation}\label{eq:nabla}
\begin{aligned}
\nabla \mu_{ij}[n] (g(t)) &= t^{-n}\begin{pmatrix}
        -\frac{1}{N}g_{ij}(t) &  & g_{1j}(t) &  & \\
        & \ddots & \vdots & & \\
        & & \frac{N-1}{N} g_{ij}(t) & & \\
        & & \vdots & \ddots & \\
        & & g_{Nj}(t) & & -\frac{1}{N} g_{ij}(t)
\end{pmatrix} \in \gfr,\\
\nabla' \mu_{ij}[n] (g(t)) &= t^{-n}\begin{pmatrix}
        -\frac{1}{N}g_{ij}(t) &  & &  & \\
        & \ddots &  & & \\
        g_{i1}(t) & \dots & \frac{N-1}{N} g_{ij}(t) & \dots & g_{iN}(t) \\
        & & & \ddots & \\
        & & & & -\frac{1}{N} g_{ij}(t)
\end{pmatrix} \in \gfr.
\end{aligned}
\end{equation}

Returning to the general setting, suppose that $G$ is a Poisson-Lie group whose associated Lie algebra $\gfr$ is equipped with a nondegenerate symmetric invariant bilinear form $\langle \cdot, \cdot \rangle$ and a classical $r$-matrix $R_\pm : \gfr \to \gfr$. Let $\dfr := \gfr \oplus \gfr$ be the direct sum of Lie algebras, equipped with the nondegenerate symmetric invariant bilinear form
    \begin{equation*}
        \langle(X_1, Y_1), (X_2, Y_2)\rangle := \langle X_1, X_2\rangle - \langle Y_1, Y_2 \rangle.
    \end{equation*}
    Then $\dfr$ splits into the direct sum of the following isotropic Lie subalgebras:
    \begin{equation*}
        \gfr^\delta := \{(X, X) \in \dfr \ | \ X \in \gfr\}, \qquad \gfr^\ast := \{(R_+(X), R_-(X)) \in \dfr \ | \ X \in \gfr\} .
    \end{equation*}
    The projection $R^\dfr_+ := P_{\gfr^\ast}$ of $\dfr$ onto $\gfr^\ast$ is a classical $r$-matrix of $\dfr$ which induces the Poisson-Lie group structure on the Drinfeld double $D(G) := G \times G$. The group $G$ can be identified with the diagonal subgroup of $D(G)$, and the Poisson structure of $D(G)$ can be restricted to $G$. Its Poisson bracket is given by:
    \begin{equation}\label{eq:AppendixBPB1}
        \{\varphi, \psi\}_{G} = \langle R(\nabla \varphi), \nabla \psi \rangle - \langle R(\nabla' \varphi), \nabla' \psi \rangle .
    \end{equation}
    The map
    \begin{equation*}
        D(G) \to G, \qquad (g, h) \mapsto gh^{-1}
    \end{equation*}
    is a Poisson map. We consider the restriction of this map to the Poisson-Lie subgroup $G^\ast \subset D(G)$ corresponding to the dual Lie algebra $\gfr^\ast$. The restricted map is a Poisson map, and its image, denoted by $G_\ast$, is a dense open subset of $G$; see \cites{S92, S85, S20}. The pushforward of the Poisson structure on $G^\ast$ via this map is given by
    \begin{equation*}
        \{\varphi, \psi\}_{G_\ast} = \frac{1}{2}\langle R(\nabla\varphi), \nabla\psi\rangle + \frac{1}{2}\langle R(\nabla'\varphi), \nabla'\psi\rangle - \langle R_+(\nabla'\varphi), \nabla \psi\rangle -\langle R_-(\nabla\varphi), \nabla'\psi\rangle.
    \end{equation*}
    
    Let $\tau$ be an automorphism of the Poisson-Lie group $G$ such that its corresponding automorphism on $\gfr$ satisfies $(\tau \otimes \tau)(r) = r$. It was proved in \cite{S92}, using a geometric framework, that the map
    \begin{equation*}
        D(G) \to G, \qquad (g, h) \mapsto \tau(g)h^{-1}
    \end{equation*}
    is a Poisson map. In this case, the resulting Poisson structure on $G_\ast$ is given by
    \begin{equation}\label{eq:AppendixBPB2}
            \{\varphi, \psi\}_{G_\ast} = \frac{1}{2}\langle R(\nabla\varphi), \nabla\psi\rangle + \frac{1}{2}\langle R(\nabla'\varphi), \nabla'\psi\rangle - \langle (\tau \circ R_+)(\nabla'\varphi), \nabla \psi\rangle -\langle (R_-\circ \tau^{-1})(\nabla\varphi), \nabla'\psi\rangle.
    \end{equation}

\subsection{Algebraic Proof of Well-definedness} \label{subsec:Jacobi, well-definedness}
In this subsection, we prove the skew-symmetry and Jacobi identity of the Poisson brackets \eqref{eq:AppendixBPB1} and \eqref{eq:AppendixBPB2}.

We summarize the notation used throughout the appendices. Let $G$ be a Lie group with Lie algebra $\gfr$. Suppose there is a decomposition of Lie subalgebras
\begin{equation*}
    \gfr = \nfr_+ \oplus \hfr \oplus \nfr_-
\end{equation*}
such that $[\nfr_\pm, \hfr] \subseteq \nfr_\pm$. Let $P_+, P_0, P_-$ be the projections onto $\nfr_+, \hfr, \nfr_-$, respectively. We assume that $\gfr$ is equipped with a nondegenerate symmetric invariant bilinear form $\langle \cdot, \cdot \rangle$ such that $P_+^\ast = P_-$ and whose restriction to $\hfr$ remains nondegenerate.

Choose a basis $\{X_\alpha\}_{\alpha \in \Ical}$ of $\gfr$, and let $\{X^\alpha\}_{\alpha \in \Ical}$ denote the corresponding dual basis with respect to $\langle \cdot, \cdot \rangle$. Let
\begin{equation}\label{eq:rmatricesbyprojections}
    R_+ := P_+ + \frac{I}{I-\theta}P_0, \quad R_- := -P_- + \frac{\theta}{I-\theta}P_0, \quad R:= R_++R_-
\end{equation}
where $\theta : \hfr \to \hfr$ is a Lie algebra homomorphism satisfying
\begin{equation*}
 \text{(i) $I-\theta$ is invertible, \quad 
    (ii) $\left\langle \theta(X), \theta(Y) \right\rangle = \left\langle X, Y \right\rangle$ for all $X, Y \in \hfr$.}
\end{equation*}
Let
\begin{equation*}
R_\pm(X^\alpha) = \sum_{\beta\in \Ical}R_\pm^{\alpha\beta}X_\beta, \quad R(X^\alpha) = \sum_{\beta \in \Ical} R^{\alpha\beta}X_\beta.
\end{equation*}
Since $R_+^\ast = -R_-$, it immediately follows that $R_+^{\alpha\beta} + R_-^{\beta\alpha}=0$ and $R^{\alpha\beta} + R^{\beta\alpha} = 0$. Suppose $\tau : \gfr \to \gfr$ is a Lie algebra automorphism satisfying the following conditions:
\begin{equation*}
    \text{ (i) $\langle\tau(X), \tau(Y)\rangle =\langle X, Y\rangle$ for all $X, Y \in \gfr$, \quad (ii) $\tau \circ R = R \circ \tau$.} 
\end{equation*}

For $X \in \gfr$ and $\varphi \in C^\infty(G)$, we define the right- and left-invariant vector fields generated by $X$, denoted by $X^R$ and $X^L$ respectively, as follows:
\begin{equation*}
\begin{aligned}
    (X^R \varphi)(g) &= \left.\frac{d}{ds}\right|_{s=0} \varphi(e^{sX}g) = \langle \nabla\varphi(g), X\rangle,\\
    (X^L \varphi)(g) &= \left.\frac{d}{ds}\right|_{s=0} \varphi(ge^{sX}) =\langle \nabla' \varphi(g), X \rangle.
\end{aligned}
\end{equation*}
Let us briefly summarize basic properties of the right- and left-invariant vector fields. For $X, Y \in \gfr$ and $\varphi \in C^\infty(G)$, we have the following relations by definition:
    \begin{equation*}
        (X^RY^R - Y^RX^R)\varphi = -[X, Y]^R \varphi, \qquad (X^LY^L-Y^LX^L)\varphi = [X, Y]^L \varphi,
    \end{equation*}
and
\begin{equation*}
    (X^LY^R-Y^RX^L)\varphi = 0.
\end{equation*}
Moreover, if we write
    \begin{equation*}
    \nabla \varphi = \sum_{\alpha \in \Ical} (\nabla \varphi)_\alpha X^\alpha, \quad \nabla' \varphi = \sum_{ \alpha \in \Ical} (\nabla' \varphi)_\alpha X^\alpha,
    \end{equation*}
    where $(\nabla \varphi)_\alpha, (\nabla' \varphi)_\alpha \in C^\infty(G)$ for each $\alpha \in \Ical$, then
    \begin{equation*}
    (\nabla \varphi)_\alpha = X_\alpha^R \varphi, \qquad (\nabla' \varphi)_\alpha = X_\alpha^L \varphi.
    \end{equation*}
    Furthermore, the following relations hold for the gradient components:
    \begin{equation*}
    \begin{aligned}
    &(\nabla(\nabla \varphi)_\alpha)_\beta -(\nabla (\nabla \varphi)_\beta)_\alpha = -[X_\beta, X_\alpha]^R \varphi,\\
    &(\nabla'(\nabla \varphi)_\alpha)_\beta -(\nabla (\nabla' \varphi)_\beta)_\alpha = 0,\\
    &(\nabla(\nabla' \varphi)_\alpha)_\beta -(\nabla' (\nabla \varphi)_\beta)_\alpha =0,\\
    &(\nabla'(\nabla' \varphi)_\alpha)_\beta -(\nabla' (\nabla' \varphi)_\beta)_\alpha = [X_\beta, X_\alpha]^L \cdot \varphi.
    \end{aligned}
    \end{equation*}

\begin{prop}\label{prop:AppendixB}
    For $\varphi, \psi \in C^\infty(G)$, the brackets defined by
    \begin{equation*}
    \{\varphi, \psi \} = \langle R(\nabla \varphi), \nabla \psi \rangle \pm \langle R(\nabla' \varphi), \nabla' \psi \rangle
    \end{equation*}
    are both well-defined Poisson brackets.
\end{prop}

\begin{proof}
    First, we verify skew-symmetry. Since $R^\ast = -R$, for any $\varphi, \psi \in C^\infty(G)$, we have
    \begin{equation*}
    \{\varphi, \psi \} = \langle R(\nabla \varphi), \nabla \psi \rangle \pm \langle R(\nabla' \varphi), \nabla' \psi \rangle = \langle \nabla \varphi , R^\ast(\nabla \psi) \rangle \pm \langle \nabla' \varphi, R^\ast(\nabla' \psi) \rangle = -\{\psi, \varphi\}.
    \end{equation*}
    Next, we verify the Jacobi identity. Let us denote the individual components of the bracket as
    \begin{equation*}
        \{\varphi, \psi\}_\nabla = \langle R(\nabla \varphi), \nabla \psi \rangle, \quad \{\varphi, \psi\}_{\nabla'} = \langle R(\nabla' \varphi), \nabla' \psi \rangle.
    \end{equation*}
    Since the left- and right- invariant vector fields commute, it can be verified that the following holds:
    \begin{equation*}
            \Big(\{\varphi_1, \{\varphi_2, \varphi_3\}_\nabla \}_{\nabla'}+ \text{cyclic} \Big)+\Big(\{\varphi_1, \{\varphi_2, \varphi_3\}_{\nabla'} \}_{\nabla} + \text{cyclic}\Big) = 0
    \end{equation*}
    for $\varphi_1, \varphi_2, \varphi_3 \in C^\infty(G)$, where $\dots + \text{cyclic}$ stands for cyclic permutations of the indices $(1, 2, 3)$. Now, expanding $\{\varphi_1, \{\varphi_2, \varphi_3\}_\nabla \}_\nabla$ and $\{\varphi_1, \{\varphi_2, \varphi_3\}_{\nabla'}\}_{\nabla'}$ using the basis $\{X_\alpha\}_{\alpha \in \Ical}$, the cyclic sum results in twelve terms. Collecting terms involving the second derivatives of $\varphi_1$, we obtain:
    \begin{equation*}
    \begin{aligned}
    &\sum_{\alpha,\beta,\gamma,\delta}R^{\alpha\beta}R^{\gamma\delta} ([X_\alpha, X_\delta]^R \varphi_1) (\nabla \varphi_2)_\beta (\nabla \varphi_3)_\gamma - \sum_{\alpha,\beta,\gamma,\delta}R^{\alpha\beta}R^{\gamma\delta} ([X_\alpha, X_\delta]^L \varphi_1) (\nabla' \varphi_2)_\beta (\nabla' \varphi_3)_\gamma\\
    &= -\sum_{\beta,\gamma} \Big( [R(X^\beta), R(X^\gamma)]^R \varphi_1 \Big) (\nabla \varphi_2)_\beta (\nabla \varphi_3)_\gamma +\sum_{\beta,\gamma} \Big( [R(X^\beta), R(X^\gamma)]^L \varphi_1 \Big) (\nabla' \varphi_2)_\beta (\nabla' \varphi_3)_\gamma\\
    &= -\langle \nabla \varphi_1, [R(\nabla \varphi_2), R(\nabla \varphi_3)] \rangle +\langle \nabla' \varphi_1, [R(\nabla' \varphi_2), R(\nabla' \varphi_3)] \rangle.
        \end{aligned}
    \end{equation*}
    Applying the mCYBE and invariance of the bilinear form, the expression simplifies to:
    \begin{equation*}
        \begin{aligned}
            &=-\langle\nabla \varphi_1, R([R(\nabla \varphi_2), \nabla \varphi_3]+[\nabla \varphi_2, R(\nabla \varphi_3)])\rangle+\langle\nabla \varphi_1, [\nabla \varphi_2, \nabla \varphi_3]\rangle\\
    &\quad +\langle\nabla'\varphi_1, R([R(\nabla' \varphi_2), \nabla' \varphi_3] +[\nabla' \varphi_2, R(\nabla' \varphi_3)])\rangle - \langle \nabla' \varphi_1, [\nabla' \varphi_2, \nabla' \varphi_3]\rangle\\
    &= \langle\nabla \varphi_3, [R(\nabla \varphi_1), R(\nabla \varphi_2)]\rangle +\langle\nabla \varphi_2, [R(\nabla \varphi_3), R(\nabla \varphi_1)] \rangle \\
    &\quad -\langle\nabla' \varphi_3, [R(\nabla' \varphi_1), R(\nabla' \varphi_2)]\rangle -\langle\nabla' \varphi_2, [R(\nabla' \varphi_3), R(\nabla' \varphi_1)] \rangle.
        \end{aligned}
    \end{equation*}
    These terms cancel with the terms arising from the second derivatives of $\varphi_2$ and $\varphi_3$ in the cyclic sum. Thus, the Jacobi identity holds.
\end{proof}

\begin{prop}\label{prop:Jacobi}
    For $\varphi, \psi \in C^\infty(G)$, the bracket defined by
    \begin{equation}\label{eq:PoissonBracketDef}
    \{\varphi, \psi \} = \frac{1}{2}\langle R(\nabla \varphi), \nabla \psi \rangle +\frac{1}{2}\langle R(\nabla' \varphi), \nabla' \psi\rangle-\langle (\tau \circ R_+)(\nabla' \varphi), \nabla \psi\rangle-\langle (R_- \circ \tau^{-1})(\nabla \varphi), \nabla' \psi\rangle
    \end{equation}
    is a well-defined Poisson bracket.
\end{prop}
\begin{proof}
    We decompose the bracket into four distinct components:
    \begin{equation*}
    \begin{aligned}
    \{\varphi, \psi\}_1 = \frac{1}{2}\langle R(\nabla \varphi), \nabla \psi \rangle,& \quad \{\varphi, \psi\}_2 = \frac{1}{2}\langle R(\nabla' \varphi), \nabla' \psi\rangle\\
    \{\varphi, \psi\}_3 = -\langle (\tau \circ R_+)(\nabla' \varphi), \nabla \psi\rangle,& \quad \{\varphi, \psi\}_4 = -\langle (R_- \circ \tau^{-1})(\nabla \varphi), \nabla' \psi\rangle. 
    \end{aligned}
    \end{equation*}
    Skew-symmetry for the first two brackets follows from $R^*=-R$. Furthermore,
    \begin{equation*}
    \begin{aligned}
        \{\varphi, \psi\}_3 &= -\langle R_+(\nabla' \varphi), \tau^{-1}(\nabla \psi)\rangle = \langle \nabla' \varphi, (R_- \circ \tau^{-1})\nabla \psi\rangle = -\{\psi, \varphi\}_4\\
        \{\varphi, \psi\}_4 &= \langle \tau^{-1}(\nabla \varphi), R_+(\nabla' \psi) \rangle = \langle \nabla \varphi, (\tau \circ R_+)(\nabla' \psi)\rangle = -\{\psi, \varphi\}_3.
        \end{aligned}
    \end{equation*}
    Thus, the bracket \eqref{eq:PoissonBracketDef} is skew-symmetric. To verify the Jacobi identity, we define
    \begin{equation*}
        \text{Jac}(i, j) := \{\varphi_1, \{\varphi_2, \varphi_3\}_i\}_j + \{\varphi_2, \{\varphi_3, \varphi_1\}_i\}_j + \{\varphi_3, \{\varphi_1, \varphi_2\}_i\}_j, \quad i, j \in \{1, 2, 3, 4\}.
    \end{equation*}
    From Proposition~\ref{prop:AppendixB}, we have
    \begin{equation*}
    \text{Jac}(1, 1) + \text{Jac}(1, 2) + \text{Jac}(2, 1) + \text{Jac}(2, 2)=0.
    \end{equation*}
    It remains to show that the remaining terms vanish.
    We first compute $\text{Jac}(3, 3) + \text{Jac}(3, 4)+ \text{Jac}(4, 3) + \text{Jac}(4, 4)$. Note that the brackets $\{\varphi, \psi\}_3$ and $\{\varphi, \psi\}_4$ can be expanded as follows:
    \begin{equation*}
    \begin{aligned}
        \{\varphi, \psi\}_3 &=-\sum_{\alpha, \beta, \gamma} R_+^{\alpha\beta}(\nabla' \varphi)_\alpha (\nabla \psi)_\gamma \langle\tau(X_\beta), X^\gamma\rangle\\
        \{\varphi, \psi\}_4 &= \sum_{\alpha, \beta, \gamma} R_+^{\alpha\beta}(\nabla \varphi)_\gamma (\nabla' \psi)_\alpha \langle \tau(X_\beta), X^{\gamma}\rangle.
    \end{aligned}
    \end{equation*}
    Using these expressions, the corresponding cyclic sums evaluate to:
    \begin{align}
        \text{Jac}(3, 3) &= \sum_{\substack{\alpha\beta\gamma\\\delta\epsilon\zeta}} R_+^{\alpha\beta}R_+^{\delta\epsilon}\langle \tau(X_\beta), X^\gamma \rangle \langle \tau(X_\epsilon), X^\zeta \rangle (\nabla' \varphi_1)_\delta (\nabla (\nabla' \varphi_2)_\alpha)_\zeta (\nabla \varphi_3)_\gamma + \text{cyclic} \label{eq:tauNNP1}\\
        &+\sum_{\substack{\alpha\beta\gamma\\\delta\epsilon\zeta}} R_+^{\alpha\beta}R_+^{\delta\epsilon} \langle \tau(X_\beta), X^\gamma \rangle \langle \tau(X_\epsilon), X^\zeta \rangle (\nabla' \varphi_1)_\delta (\nabla' \varphi_2)_\alpha(\nabla(\nabla \varphi_3)_\gamma)_\zeta + \text{cyclic},\label{eq:tauNN1}\\
        \text{Jac}(3, 4) &= -\sum_{\substack{\alpha\beta\gamma\\\delta\epsilon\zeta}}  R_+^{\alpha\beta}R_+^{\delta\epsilon} \langle \tau(X_\beta), X^\gamma \rangle \langle \tau(X_\epsilon), X^\zeta \rangle(\nabla \varphi_1)_\zeta (\nabla'(\nabla' \varphi_2)_\alpha)_\delta(\nabla \varphi_3)_\gamma +\text{cyclic}\label{eq:tauNPNP1}\\
        &-\sum_{\substack{\alpha\beta\gamma\\\delta\epsilon\zeta}}  R_+^{\alpha\beta}R_+^{\delta\epsilon} \langle \tau(X_\beta), X^\gamma \rangle \langle \tau(X_\epsilon), X^\zeta \rangle(\nabla \varphi_1)_\zeta (\nabla' \varphi_2)_\alpha(\nabla'(\nabla \varphi_3)_\gamma)_{\delta} +\text{cyclic}, \label{eq:tauNPN1}\\
        \text{Jac}(4, 3) &= -\sum_{\substack{\alpha\beta\gamma\\\delta\epsilon\zeta}}  R_+^{\alpha\beta}R_+^{\delta\epsilon} \langle \tau(X_\beta), X^\gamma \rangle \langle \tau(X_\epsilon), X^\zeta \rangle (\nabla' \varphi_1)_\delta (\nabla(\nabla \varphi_2)_\gamma)_\zeta(\nabla' \varphi_3)_\alpha +\text{cyclic}\label{eq:tauNN2}\\
        &-\sum_{\substack{\alpha\beta\gamma\\\delta\epsilon\zeta}}  R_+^{\alpha\beta}R_+^{\delta\epsilon} \langle \tau(X_\beta), X^\gamma \rangle \langle \tau(X_\epsilon), X^\zeta \rangle (\nabla' \varphi_1)_\delta (\nabla \varphi_2)_\gamma(\nabla(\nabla' \varphi_3)_\alpha)_\zeta +\text{cyclic},\label{eq:tauNNP2}\\
        \text{Jac}(4, 4) &= \sum_{\substack{\alpha\beta\gamma\\\delta\epsilon\zeta}}  R_+^{\alpha\beta}R_+^{\delta\epsilon} \langle \tau(X_\beta), X^\gamma \rangle \langle \tau(X_\epsilon), X^\zeta \rangle (\nabla \varphi_1)_\zeta (\nabla'(\nabla \varphi_2)_\gamma)_\delta(\nabla' \varphi_3)_\alpha +\text{cyclic}\label{eq:tauNPN2}\\
        &+\sum_{\substack{\alpha\beta\gamma\\\delta\epsilon\zeta}}  R_+^{\alpha\beta}R_+^{\delta\epsilon} \langle \tau(X_\beta), X^\gamma \rangle \langle \tau(X_\epsilon), X^\zeta \rangle (\nabla \varphi_1)_\zeta (\nabla \varphi_2)_\gamma(\nabla'(\nabla' \varphi_3)_\alpha)_{\delta} +\text{cyclic}.\label{eq:tauNPNP2}
    \end{align}
    To observe the cancellations systematically, we group specific pairings of these equations:
    \begin{align}
       \eqref{eq:tauNN1}+\eqref{eq:tauNN2} &= \langle\nabla \varphi_1, [(\tau\circ R_+)(\nabla' \varphi_2), (\tau \circ R_+)(\nabla' \varphi_3)]\rangle + \text{cyclic}, \label{eq:tauNNPNPremain}\\
       \eqref{eq:tauNNP1}+\eqref{eq:tauNPN1} &=\eqref{eq:tauNPN2}+\eqref{eq:tauNNP2}=0,\\
       \eqref{eq:tauNPNP2}+\eqref{eq:tauNPNP1} &=-\langle \nabla' \varphi_1, [(R_- \circ \tau^{-1})(\nabla \varphi_2), (R_- \circ \tau^{-1}) (\nabla \varphi_3)] \rangle + \text{cyclic}.\label{eq:tauNPNNremain}
    \end{align}
    Now we turn our attention to the remaining cross terms: $\text{Jac}(1, 3)+\text{Jac}(3, 1)+\text{Jac}(1, 4)+\text{Jac}(4, 1)$. Expanding these yields:
    \begin{align}
        \text{Jac}(1, 3) &= -\frac{1}{2}\sum_{\substack{\alpha\beta\\\gamma\delta\epsilon}} R^{\alpha\beta}R_+^{\gamma\delta}\langle \tau(X_\delta), X^\epsilon\rangle (\nabla' \varphi_1)_\gamma (\nabla(\nabla \varphi_2)_\alpha)_\epsilon(\nabla \varphi_3)_\beta +\text{cyclic}\label{eq:tauNNNNP1}\\
        &-\frac{1}{2}\sum_{\substack{\alpha\beta\\\gamma\delta\epsilon}} R^{\alpha\beta}R_+^{\gamma\delta}\langle \tau(X_\delta), X^\epsilon\rangle (\nabla' \varphi_1)_\gamma (\nabla \varphi_2)_\alpha(\nabla(\nabla \varphi_3)_\beta)_{\epsilon} +\text{cyclic}\label{eq:tauNNNPN1}\\
        \text{Jac}(3, 1) &= -\frac{1}{2}\sum_{\substack{\alpha\beta\\\gamma\delta\epsilon}} R_+^{\alpha\beta}R^{\delta\epsilon}\langle\tau(X_\beta), X^\gamma\rangle (\nabla \varphi_1)_\delta (\nabla(\nabla' \varphi_2)_\alpha)_\epsilon(\nabla \varphi_3)_\gamma +\text{cyclic}\label{eq:tauNNPNN1}\\
        &-\frac{1}{2}\sum_{\substack{\alpha\beta\\\gamma\delta\epsilon}} R_+^{\alpha\beta}R^{\delta\epsilon}\langle\tau(X_\beta), X^\gamma\rangle(\nabla \varphi_1)_\delta (\nabla' \varphi_2)_\alpha(\nabla(\nabla \varphi_3)_\gamma)_{\epsilon} +\text{cyclic} \label{eq:tauNNNNP2}\\
        \text{Jac}(1, 4) &= \frac{1}{2}\sum_{\substack{\alpha\beta\\\gamma\delta\epsilon}} R^{\alpha\beta}R_+^{\gamma\delta} \langle\tau(X_\delta), X^\epsilon\rangle(\nabla \varphi_1)_\epsilon (\nabla'(\nabla \varphi_2)_\alpha)_\gamma(\nabla \varphi_3)_\beta +\text{cyclic}\label{eq:tauNPNNN1}\\
        &+\frac{1}{2}\sum_{\substack{\alpha\beta\\\gamma\delta\epsilon}} R^{\alpha\beta}R_+^{\gamma\delta} \langle\tau(X_\delta), X^\epsilon\rangle (\nabla \varphi_1)_\epsilon (\nabla \varphi_2)_\alpha(\nabla'(\nabla \varphi_3)_\beta)_\gamma +\text{cyclic}\label{eq:tauNPNNN2}\\
        \text{Jac}(4, 1) &= \frac{1}{2}\sum_{\substack{\alpha\beta\\\gamma\delta\epsilon}} R_+^{\alpha\beta}R^{\delta\epsilon}\langle \tau(X_\beta), X^\gamma\rangle (\nabla \varphi_1)_\delta (\nabla(\nabla \varphi_2)_\gamma)_\epsilon(\nabla' \varphi_3)_\alpha +\text{cyclic}\label{eq:tauNNNPN2}\\
        &+\frac{1}{2}\sum_{\substack{\alpha\beta\\\gamma\delta\epsilon}} R_+^{\alpha\beta}R^{\delta\epsilon}\langle \tau(X_\beta), X^\gamma\rangle (\nabla \varphi_1)_\delta (\nabla \varphi_2)_\gamma(\nabla(\nabla' \varphi_3)_\alpha)_\epsilon +\text{cyclic} \label{eq:tauNNPNN2}
    \end{align}
    Grouping these terms appropriately, we have
    \begin{align}
       \eqref{eq:tauNNNNP2}+\eqref{eq:tauNNNNP1} &=-\frac{1}{2}\langle \nabla \varphi_1, [R(\nabla \varphi_2), (\tau \circ R_+)(\nabla' \varphi_3)]\rangle + \text{cyclic} \label{eq:tauNPNNremain2}\\
       \eqref{eq:tauNNNPN2}+\eqref{eq:tauNNNPN1} &= -\frac{1}{2}\langle \nabla \varphi_1, [(\tau \circ R_+)(\nabla' \varphi_2), R(\nabla \varphi_3)]\rangle + \text{cyclic} \label{eq:tauNPNNremain3}\\
       \eqref{eq:tauNPNNN2}+\eqref{eq:tauNNPNN1} &= \eqref{eq:tauNPNNN1}+\eqref{eq:tauNNPNN2} =0.
    \end{align}
    Furthermore, the term \eqref{eq:tauNPNNremain} cancels with \eqref{eq:tauNPNNremain2} and \eqref{eq:tauNPNNremain3} as follows:
    \begin{equation*}
    \begin{aligned}
        \eqref{eq:tauNPNNremain} &= -\langle \nabla' \varphi_1, R_-[R_- (\tau^{-1}(\nabla \varphi_2)), \tau^{-1}(\nabla \varphi_3)]\rangle -\langle \nabla' \varphi_1, R_-[\tau^{-1}(\nabla \varphi_2), R_+(\tau^{-1}(\nabla \varphi_3))]\rangle + \text{cyclic}\\
        &= \langle R_+ (\nabla' \varphi_1), [R_- (\tau^{-1}(\nabla \varphi_2)), \tau^{-1}(\nabla \varphi_3)] \rangle + \langle R_+ (\nabla' \varphi_1), [\tau^{-1}(\nabla \varphi_2), R_+(\tau^{-1}(\nabla \varphi_3))]\rangle +\text{cyclic}\\
        &= \langle \tau^{-1}(\nabla \varphi_3), [R_+(\nabla' \varphi_1), R_-(\tau^{-1}(\nabla \varphi_2))]\rangle + \langle \tau^{-1}(\nabla \varphi_2), [R_+(\tau^{-1}(\nabla \varphi_3)), R_+(\nabla' \varphi_1)] \rangle +\text{cyclic}\\
        &= \frac{1}{2} \langle \tau^{-1}(\nabla \varphi_3), [R_+(\nabla' \varphi_1), R(\tau^{-1}(\nabla \varphi_2))]\rangle + \frac{1}{2} \langle \tau^{-1}(\nabla \varphi_2), [R(\tau^{-1}(\nabla \varphi_3)), R_+(\nabla' \varphi_1)]\rangle +\text{cyclic}\\
        &= \frac{1}{2} \langle \nabla \varphi_3, [(\tau \circ R_+)(\nabla' \varphi_1), R(\nabla \varphi_2)]\rangle + \frac{1}{2} \langle \nabla \varphi_2, [R(\nabla \varphi_3), (\tau \circ R_+)(\nabla' \varphi_1)]\rangle +\text{cyclic} \\
        &= -\eqref{eq:tauNPNNremain2} - \eqref{eq:tauNPNNremain3}.
    \end{aligned}
    \end{equation*}
Finally, we evaluate the last set of mixed terms: $\text{Jac}(2, 3)+\text{Jac}(3, 2)+\text{Jac}(2, 4)+\text{Jac}(4, 2)$. By a similar calculation, we have:
    \begin{align}
        \text{Jac}(2, 3) &= -\frac{1}{2}\sum_{\substack{\alpha\beta\\\gamma\delta\epsilon}} R^{\alpha\beta}R_+^{\gamma\delta}\langle \tau(X_\delta), X^\epsilon\rangle (\nabla' \varphi_1)_\gamma (\nabla(\nabla' \varphi_2)_\alpha)_\epsilon(\nabla' \varphi_3)_\beta +\text{cyclic}\label{eq:tauNNNNP11}\\
        &-\frac{1}{2}\sum_{\substack{\alpha\beta\\\gamma\delta\epsilon}} R^{\alpha\beta}R_+^{\gamma\delta}\langle \tau(X_\delta), X^\epsilon\rangle (\nabla' \varphi_1)_\gamma (\nabla' \varphi_2)_\alpha(\nabla(\nabla' \varphi_3)_\beta)_{\epsilon} +\text{cyclic}\label{eq:tauNNNPN11}\\
        \text{Jac}(3, 2) &= -\frac{1}{2}\sum_{\substack{\alpha\beta\\\gamma\delta\epsilon}} R_+^{\alpha\beta}R^{\delta\epsilon}\langle\tau(X_\beta), X^\gamma\rangle (\nabla' \varphi_1)_\delta (\nabla'(\nabla' \varphi_2)_\alpha)_\epsilon(\nabla \varphi_3)_\gamma +\text{cyclic}\label{eq:tauNNPNN11}\\
        &-\frac{1}{2}\sum_{\substack{\alpha\beta\\\gamma\delta\epsilon}} R_+^{\alpha\beta}R^{\delta\epsilon}\langle\tau(X_\beta), X^\gamma\rangle(\nabla' \varphi_1)_\delta (\nabla' \varphi_2)_\alpha(\nabla'(\nabla \varphi_3)_\gamma)_{\epsilon} +\text{cyclic} \label{eq:tauNNNNP22}\\
        \text{Jac}(2, 4) &= \frac{1}{2}\sum_{\substack{\alpha\beta\\\gamma\delta\epsilon}} R^{\alpha\beta}R_+^{\gamma\delta} \langle\tau(X_\delta), X^\epsilon\rangle(\nabla \varphi_1)_\epsilon (\nabla'(\nabla' \varphi_2)_\alpha)_\gamma(\nabla' \varphi_3)_\beta +\text{cyclic}\label{eq:tauNPNNN11}\\
        &+\frac{1}{2}\sum_{\substack{\alpha\beta\\\gamma\delta\epsilon}} R^{\alpha\beta}R_+^{\gamma\delta} \langle\tau(X_\delta), X^\epsilon\rangle (\nabla \varphi_1)_\epsilon (\nabla' \varphi_2)_\alpha(\nabla'(\nabla' \varphi_3)_\beta)_\gamma +\text{cyclic}\label{eq:tauNPNNN22}\\
        \text{Jac}(4, 2) &= \frac{1}{2}\sum_{\substack{\alpha\beta\\\gamma\delta\epsilon}} R_+^{\alpha\beta}R^{\delta\epsilon}\langle \tau(X_\beta), X^\gamma\rangle (\nabla' \varphi_1)_\delta (\nabla'(\nabla \varphi_2)_\gamma)_\epsilon(\nabla' \varphi_3)_\alpha+\text{cyclic} \label{eq:tauNNNPN22}\\
        &+\frac{1}{2}\sum_{\substack{\alpha\beta\\\gamma\delta\epsilon}} R_+^{\alpha\beta}R^{\delta\epsilon}\langle \tau(X_\beta), X^\gamma\rangle (\nabla' \varphi_1)_\delta (\nabla \varphi_2)_\gamma(\nabla'(\nabla' \varphi_3)_\alpha)_\epsilon +\text{cyclic} \label{eq:tauNNPNN22}
    \end{align}
    Grouping these terms appropriately, we have
    \begin{align}
    \eqref{eq:tauNNNNP22}+\eqref{eq:tauNNNNP11} &= \eqref{eq:tauNNNPN22}+\eqref{eq:tauNNNPN11} = 0 \\
       \eqref{eq:tauNPNNN22}+\eqref{eq:tauNNPNN11} &= \frac{1}{2} \langle\nabla' \varphi_1, [(R_- \circ \tau^{-1})(\nabla \varphi_2), R(\nabla' \varphi_3)]\rangle +\text{cyclic}\label{eq:tauNNPNPremain2}\\
       \eqref{eq:tauNPNNN11}+\eqref{eq:tauNNPNN22} &= \frac{1}{2} \langle \nabla' \varphi_1, [R(\nabla' \varphi_2), (R_- \circ \tau^{-1})(\nabla \varphi_3)] \rangle +\text{cyclic}\label{eq:tauNNPNPremain3}
    \end{align}
    The term \eqref{eq:tauNNPNPremain} cancels with \eqref{eq:tauNNPNPremain2} and \eqref{eq:tauNNPNPremain3} as follows:
    \begin{equation*}
    \begin{aligned}
        \eqref{eq:tauNNPNPremain} &= \langle \tau^{-1}(\nabla \varphi_1), [R_+(\nabla' \varphi_2), R_+(\nabla' \varphi_3)] \rangle+\text{cyclic}\\
        &=\langle \tau^{-1}(\nabla \varphi_1), R_+[R_+(\nabla' \varphi_2), \nabla' \varphi_3]\rangle +\langle \tau^{-1}(\nabla \varphi_1), R_+[\nabla' \varphi_2, R_-(\nabla' \varphi_3)]\rangle + \text{cyclic}\\
        &= -\langle (R_- \circ \tau^{-1}) (\nabla \varphi_1), [R_+(\nabla' \varphi_2), \nabla' \varphi_3] \rangle - \langle (R_- \circ \tau^{-1}) (\nabla \varphi_1), [\nabla' \varphi_2, R_-(\nabla' \varphi_3)]\rangle +\text{cyclic}\\
        &= -\langle\nabla' \varphi_3, [(R_- \circ \tau^{-1})(\nabla \varphi_1), R_+(\nabla' \varphi_2)] \rangle - \langle \nabla' \varphi_2, [R_-(\nabla' \varphi_3), (R_- \circ \tau^{-1})(\nabla \varphi_1)]\rangle +\text{cyclic} \\
        &= -\frac{1}{2} \langle \nabla' \varphi_3, [(R_- \circ \tau^{-1})(\nabla \varphi_1), R(\nabla' \varphi_2)] \rangle -\frac{1}{2} \langle \nabla' \varphi_2, [R(\nabla' \varphi_3), (R_- \circ \tau^{-1})(\nabla \varphi_1)] \rangle + \text{cyclic}\\
        &=-\eqref{eq:tauNNPNPremain2} - \eqref{eq:tauNNPNPremain3}.
    \end{aligned}
    \end{equation*}
    Combining all the cancellation above, all terms in the cyclic sum vanish. Hence the Jacobi identity holds.
\end{proof}

\section{Well-definedness of the Poisson structure on \texorpdfstring{$\Wbf_q(LSL_N)$}{WqLSLN}}\label{AppendixC}
In this appendix, we prove that the Poisson structure on $\Vbf_q(LSL_N)$ induces a well-defined Poisson structure on $\Wbf_q(LSL_N)$. Throughout this appendix, we use the convention that $\mu_{ij}(z) = 0$ whenever either index $i$ or $j$ lies outside the range $\{1, \dots, N\}$.

\begin{lem}\label{lem:AppendixC1}
    Let $1 \le i \le N-1$ and $1 \le k \le l \le N$.
    \begin{enumerate}[(1)]
    \item If $k < i$ and $l \ge i+1$,
    \begin{equation*}
        \pi_{B_+} \{\mu_{i+1, i}(z), \mu_{kl}(w)\} =\delta\left(\frac{w}{z}\right) \mu_{ki}(z) \mu_{i+1, l}(w).
    \end{equation*}
    \item If $k = l = i$,
    \begin{equation*}
        \pi_{B_+} \{\mu_{i+1, i}(z), \mu_{kl}(w)\} =\delta\left(\frac{qw}{z}\right) \mu_{i+1, i+1}(z).
    \end{equation*}
    \item If $k = i$ and $l \ge i+1$,
    \begin{equation*}
        \begin{aligned} \pi_{B_+} \{\mu_{i+1, i}(z), \mu_{kl}(w)\} &=\delta\left(\frac{qw}{z}\right) \mu_{i+1, l+1}(z)+\delta\left(\frac{w}{z}\right)\mu_{ii}(z) \mu_{i+1, l}(w)  \\
        &\quad - \sum_{\alpha=i+1}^l \delta\left(\frac{qw}{z}\right) \mu_{i+1, \alpha}(z)\mu_{\alpha l}(w). \end{aligned}
    \end{equation*}
    \item If $k =  i+1$,
    \begin{equation*}
        \pi_{B_+} \{\mu_{i+1, i}(z), \mu_{kl}(w)\}= -\delta\left(\frac{w}{z}\right) \mu_{i+1, l}(z).
    \end{equation*}
    \item Otherwise,
    \begin{equation*}
        \pi_{B_+} \{\mu_{i+1, i}(z), \mu_{kl}(w)\}= 0.
    \end{equation*}
    \end{enumerate}
\end{lem}

For $f(z) \in \Vbf_q(B_+)$, we define $\deg(f(z))$ to be the polynomial degree of $f(z)$ with respect to the variables $\mu_{ij}(q^nz)$ for $i \le j$ and $n \in \ZZ$. Observe that $\pi_{B_+}\{\mu_{i+1, i}(z), \mu_{kl}(w)\}$ is a linear combination of degree $1$ and $2$ terms. Furthermore, degree $1$ terms arise if and only if $k=i$ or $k=i+1$.
\begin{lem}\label{lem:well1}
    For $i, r \in \{1, 2, \dots, N-1\}$,
    \begin{equation*}
    \pi_{B_+}\{\mu_{i+1, i}(z), T_r(w)\}=0.
    \end{equation*}
\end{lem}
\begin{proof}
    We observe that $T_r(w)$ consists of monomials in $\mu_{\alpha \beta}(w)$, $\alpha \le \beta$, with degrees $1, 2, \dots, r$. Furthermore, applying $\pi_{B_+}\{\mu_{i+1, i}(z), \cdot\}$ to a term of degree $m$ yields terms of degrees $m$ and $m+1$. Consequently, to prove the assertion, we verify that the expansion of $\pi_{B_+}\{\mu_{i+1, i}(z), T_r(w)\}$ vanishes degree by degree. Specifically, we examine the components of degree $1$, the maximal degree $r+1$, and the intermediate degrees $2 \le m \le r$.

    Throughout this proof, we omit the notation $\pi_{B_+}$ by directly setting $\mu_{\alpha\beta}(z) = 0$ for $\alpha > \beta+1$ and $\mu_{\alpha+1, \alpha}(z) = -1$ for the sake of brevity.

    We first track the linear terms from the following Poisson bracket:
    \begin{equation*}
        \left\{\mu_{i+1, i}(z), \mu_{1r}(w) + \mu_{2, r+1}(qw) + \dots + \mu_{N-r+1, N}(q^{N-r}w) \right\}.
    \end{equation*}
    The degree $1$ terms are explicitly given by
    \begin{equation*}
        \delta_{i+r \le N}\left(\delta\left(\frac{q^iw}{z}\right) \mu_{i+1, i+r}(z) - \delta\left(\frac{q^{i}w}{z}\right) \mu_{i+1, i+r}(z)\right) = 0.
    \end{equation*}
    
    Next, we track the degree $r+1$ terms. Since the bracket
    \begin{equation*}
        \left\{\mu_{i+1, i}(z), \mu_{\alpha\alpha}(w)\right\} = -\delta_{\alpha, i+1} \delta\left(\frac{w}{z}\right) \mu_{i+1, i+1}(z) + \delta_{\alpha i} \delta\left(\frac{qw}{z}\right) \mu_{i+1, i+1}(z)
    \end{equation*}
    is of degree $1$, it follows that no term of degree $r+1$ appears in $\left\{\mu_{i+1, i}(z), T_r(w)\right\}.$
    
    Finally, we prove that the terms of degree $m$ vanish for $2 \le m \le r$. Recall the definition of the index set $\Jcal_{m, r}^{N}$ given in \eqref{eq:index}.

    For a term of degree $m$ to produce a term of the same degree under the Poisson bracket with $\mu_{i+1, i}(z)$, it must contain a factor of the form $\mu_{i\beta}(w)$ or $\mu_{i+1, \beta}(w)$. Consider any $(\alpha, \beta) \in \Jcal_{m, r}^N$ such that $\alpha_l = i$ for some $1 \le l \le m$. By applying the Leibniz rule for the Poisson bracket, the bracket $\{\mu_{i+1, i}(z), \mu_{\alpha_l\beta_l}(q^{\gamma_l}w)\}$ yields the following monomial:
    \begin{equation*}
        \mu_{\alpha_1\beta_1}(q^{\gamma_1}w)\cdots \left( \delta\left(\frac{q^{\gamma_l+1}w}{z}\right) \mu_{i+1, \beta_l+1}(q^{\gamma_l}w)\right) \cdots \mu_{\alpha_{m}\beta_m}(q^{\gamma_m}w).
    \end{equation*}
    On the other hand, if we consider $(\alpha, \beta) \in \Jcal_{m, r}^N$ such that $\alpha_l = i+1$ for some $1 \le l \le m$, we obtain
    \begin{equation*}
        -\mu_{\alpha_1\beta_1}(q^{\gamma_1}w)\cdots \left( \delta\left(\frac{q^{\gamma_l}w}{z}\right) \mu_{i+1, \beta_l}(q^{\gamma_l}w)\right) \cdots \mu_{\alpha_{m}\beta_m}(q^{\gamma_m}w).
    \end{equation*}
    Consequently, after canceling these matching terms, we are left only with the following monomials for $2 \le l \le m$:
    \begin{equation}\label{eq:AppendixC1}
        -\delta\left(\frac{q^{\gamma_l}w}{z}\right) \left(\mu_{\alpha_1\beta_1}(q^{\gamma_1}w) \cdots \mu_{\alpha_{l-1} i}(q^{\gamma_{l-1}}w)\mu_{i+1, \beta_l}(q^{\gamma_l}w)\cdots \mu_{\alpha_m\beta_m}(q^{\gamma_m}w)\right).
    \end{equation}
    
    Next, we consider the terms of degree $m-1$ that give rise to terms of degree $m$ under the Poisson bracket with $\mu_{i+1, i}(z)$, and show that they precisely cancel with \eqref{eq:AppendixC1}. Take any $(\alpha, \beta) \in \Jcal_{m-1, r}^N$ such that $\alpha_l < i$ and $\beta_l \ge i+1$ for some $1 \le l \le m-1$. By applying the Leibniz rule for the Poisson bracket, the bracket $\{\mu_{i+1, i}(z), \mu_{\alpha_l\beta_l}(q^{\gamma_l}w)\}$ yields the following monomial:
    \begin{equation}\label{eq:AppendixC2}
        \mu_{\alpha_1\beta_1}(q^{\gamma_1}w)\cdots \left( \delta\left(\frac{q^{\gamma_l}w}{z}\right) \mu_{\alpha_l i}(z)\mu_{i+1, \beta_l}(q^{\gamma_l}w)\right) \cdots \mu_{\alpha_{m-1}\beta_{m-1}}(q^{\gamma_{m-1}}w).
    \end{equation}
    On the other hand, if we consider $(\alpha, \beta) \in \Jcal_{m-1, r}^{N}$ such that $\alpha_l = i$ and $\beta_l \ge i+1$ for some $1 \le l \le m-1$, we obtain
    \begin{equation}\label{eq:AppendixC3}
    \begin{aligned}
        &\mu_{\alpha_1\beta_1}(q^{\gamma_1}w)\cdots \left( \delta\left(\frac{q^{\gamma_l}w}{z}\right) \mu_{ii}(z)\mu_{i+1, \beta_l}(q^{\gamma_l}w)\right) \cdots \mu_{\alpha_{m-1}\beta_{m-1}}(q^{\gamma_{m-1}}w)\\
        &- \sum_{x=i+1}^{\beta_l} \mu_{\alpha_1\beta_1}(q^{\gamma_1}w)\cdots \left( \delta\left(\frac{q^{\gamma_l+1}w}{z}\right) \mu_{i+1, x}(z)\mu_{x\beta_l}(q^{\gamma_l}w)\right) \cdots \mu_{\alpha_{m-1}\beta_{m-1}}(q^{\gamma_{m-1}}w).
        \end{aligned}
        \end{equation}
    By combining \eqref{eq:AppendixC2} and \eqref{eq:AppendixC3}, we obtain
    \begin{equation}\label{eq:AppendixC4}
        \delta\left(\frac{q^{\gamma_l}w}{z}\right) \mu_{\alpha_1\beta_1}(q^{\gamma_1}w)\cdots \Big( \mu_{\alpha_l i}(z)\mu_{i+1, \beta_l}(q^{\gamma_l}w)\Big) \cdots \mu_{\alpha_{m-1}\beta_{m-1}}(q^{\gamma_{m-1}}w)
    \end{equation}
    for $1 \le l \le m-1$. Consequently, the term corresponding to $l$ in \eqref{eq:AppendixC4} cancels with the term corresponding to $l+1$ in \eqref{eq:AppendixC1}, completing the proof.
\end{proof}

Now we investigate the application of $\pi_{B_+}\{\mu_{i+d, i}(z), \cdot\}$ for $d \in \ZZ_{>1}$.

\begin{lem}\label{lem:AppendixC2}
    Let $d \in \ZZ_{>1}, 1 \le i \le N-d$, and $1 \le k \le l \le N$.
    \begin{enumerate}[(1)]
    \item If $k < i$ and $l = i+d-1$,
    \begin{equation*}
        \pi_{B_+} \{\mu_{i+d, i}(z), \mu_{kl}(w)\} =-\delta\left(\frac{w}{z}\right) \mu_{ki}(z).
    \end{equation*}
    \item If $k < i$ and $l \ge i+d$,
    \begin{equation*}
        \pi_{B_+} \{\mu_{i+d, i}(z), \mu_{kl}(w)\} = \delta\left(\frac{w}{z}\right) \mu_{ki}(z)\mu_{i+1, l}(w).
    \end{equation*}
    \item If $k=i$ and $l = i+d-2$,
    \begin{equation*}
        \pi_{B_+} \{\mu_{i+d, i}(z), \mu_{kl}(w)\} =-\delta\left(\frac{qw}{z}\right).
    \end{equation*}
    \item If $k = i$ and $l =i+d-1$,
    \begin{equation*}
        \pi_{B_+} \{\mu_{i+d, i}(z), \mu_{kl}(w)\} =-\delta\left(\frac{w}{z}\right) \mu_{ii}(z)+\delta\left(\frac{qw}{z}\right)\mu_{i+d-1, i+d-1}(w) + \delta\left(\frac{qw}{z}\right) \mu_{i+d, i+d}(z).
    \end{equation*}
    \item If $k = i$ and $l \ge i+d$,
    \begin{equation*}
    \begin{aligned}
        \pi_{B_+} \{\mu_{i+d, i}(z), \mu_{kl}(w)\}&= \delta\left(\frac{w}{z}\right) \mu_{ii}(z)\mu_{i+d, l}(w)+ \delta\left(\frac{qw}{z}\right)\mu_{i+d-1, l}(w) \\
        &\quad - \sum_{x=i+d}^{l} \delta\left(\frac{qw}{z}\right) \mu_{i+d, x}(z)\mu_{xl}(w) + \delta\left(\frac{qw}{z}\right) \mu_{i+d, l+1}(z).
    \end{aligned}
    \end{equation*}
    \item If $k= i+1$ and $l = i+d-1$,
    \begin{equation*}
        \pi_{B_+} \{\mu_{i+d, i}(z), \mu_{kl}(w)\} = \delta\left(\frac{w}{z}\right).
    \end{equation*}
    \item If $k= i+1$ and $l \ge i+d$,
    \begin{equation*}
        \pi_{B_+} \{\mu_{i+d, i}(z), \mu_{kl}(w)\} = -\delta\left(\frac{w}{z}\right)\mu_{i+d, l}(w).
    \end{equation*}
    \item Otherwise,
    \begin{equation*}
        \pi_{B_+} \{\mu_{i+d, i}(z), \mu_{kl}(w)\} = 0.
    \end{equation*}
    \end{enumerate}
\end{lem}

Observe that $\pi_{B_+}\{\mu_{i+d, i}(z), \mu_{kl}(w)\}$ is a linear combination of degree $0, 1$ and $2$ terms.
\begin{lem}\label{lem:well2}
    For $d \in \ZZ_{>1}, 1 \le i \le N-d$, and $r \in \{1, 2, \dots, N-1\}$, 
    \begin{equation*}
    \pi_{B_+}\{\mu_{i+d, i}(z), T_r(w)\}=0.
    \end{equation*}
\end{lem}
\begin{proof}
    Following the same strategy as before, we show that $\pi_{B_+}\{\mu_{i+d, i}(z), T_r(w)\}$ vanishes degree by degree. As in the previous proof, we omit the notation $\pi_{B_+}$ and proceed by directly setting $\mu_{\alpha\beta}(z) = 0$ for $\alpha > \beta+1$ and $\mu_{\alpha+1, \alpha}(z) = -1$.

    We first track the degree $0$ terms contributed by the following Poisson bracket:
    \begin{equation*}
        \left\{\mu_{i+d, i}(z), \mu_{1r}(w) + \mu_{2, r+1}(qw) + \dots + \mu_{N-r+1, N}(q^{N-r}w) \right\}.
    \end{equation*}
    The degree $0$ terms are explicitly given by
    \begin{equation*}
        \delta_{r, d-1}\{\mu_{i+d, i}(z), \mu_{i, i+d-2}(q^{i-1}w) + \mu_{i+1, i+d-1}(q^iw)\} = \delta_{r, d-1}\left(-\delta\left(\frac{q^iw}{z}\right)+ \delta\left(\frac{q^{i}w}{z}\right)\right) = 0.
    \end{equation*}
    
    Next, we track the linear terms of $\{\mu_{i+d, i}(z), T_r(w)\}$. To do this, it suffices to consider the degree $1$ and $2$ terms of $T_r(w)$. Evaluating the Poisson bracket yields
    \begin{equation*}
        -\delta_{r \ge d} \left(\delta\left(\frac{q^{i+d-r-1}w}{z}\right) \mu_{i+d-r, i}(z) - \delta\left(\frac{q^iw}{z}\right) \mu_{i+d-1, i+r-1}(q^{i-1}w) \right).
    \end{equation*}
    For the degree $2$ terms of $T_r(w)$, computing the Poisson bracket gives
    \begin{equation*}
        \begin{aligned}
            &\quad \delta_{r \ge d}\sum_{\substack{\alpha_2 > i+d-2 \\ \beta_2 = \alpha_2+r-d}} \{\mu_{i+d, i}(z), \mu_{i, i+d-2} (q^{i-1}w) \} \mu_{\alpha_2\beta_2} (q^{\alpha_2-d}w) \\
            &+\delta_{r \ge d}\sum_{\substack{\alpha_2 > i+d-1 \\ \beta_2 = \alpha_2+r-d}} \{\mu_{i+d, i}(z), \mu_{i+1, i+d-1} (q^{i}w) \} \mu_{\alpha_2\beta_2} (q^{\alpha_2-d}w) \\
            &+\delta_{r \ge d}\sum_{\substack{\beta_1 < i \\ \beta_1-\alpha_1 = r-d}} \mu_{\alpha_1\beta_1}(q^{\alpha_1-1}w) \{\mu_{i+d, i}(z), \mu_{i, i+d-2}(q^{i-r+d-2}w)\} \\
            &+\delta_{r \ge d}\sum_{\substack{\beta_1 < i+1 \\ \beta_1-\alpha_1 = r-d}} \mu_{\alpha_1\beta_1}(q^{\alpha_1-1}w) \{\mu_{i+d, i}(z), \mu_{i, i+d-1}(q^{i-r+d-1}w)\} \\
            &= \delta_{r \ge d}\left(\delta\left(\frac{q^{i+d-r-1}w}{z}\right) \mu_{i+d-r, i}(z) - \delta\left(\frac{q^iw}{z}\right) \mu_{i+d-1, i+r-1}(q^{i-1}w) \right).
        \end{aligned}
    \end{equation*}
    To see that the degree $r+1$ and $r$ components vanish, we observe the following equalities from Lemma~\ref{lem:AppendixC2}:
    \begin{equation*}
        \begin{aligned}
            \{\mu_{i+d, i}(z), \mu_{\alpha\alpha}(w)\} &= \delta_{d, 2} \left( \delta_{\alpha, i+1} \delta\left(\frac{w}{z}\right) - \delta_{\alpha i} \delta\left(\frac{qw}{z}\right) \right) \\
            \{\mu_{i+d, i}(z) , \mu_{\alpha, \alpha+1}(w)\} &=\delta_{\alpha i} \left( -\delta_{d, 2} \delta\left(\frac{w}{z}\right) \mu_{ii}(z) + \delta_{d, 2} \delta\left(\frac{qw}{z}\right) \mu_{i+2, i+2}(z) - \delta_{d, 3} \delta\left(\frac{qw}{z}\right) \right) \\
            &\quad + \delta_{\alpha, i+1} \left( - \delta_{d, 2} \delta\left(\frac{w}{z}\right) \mu_{i+2, i+2}(w) + \delta_{d, 3}\left(\frac{w}{z}\right) \right).
        \end{aligned}
    \end{equation*}
    Let $T_{r, r}(w)$ and $T_{r, r-1}(w)$ be the degree $r$ and $r-1$ components of $T_r(w)$, respectively. Note that $T_{r, r}(w)$ is a linear combination of monomials in $\mu_{\alpha \alpha}(w)$, and $T_{r, r-1}(w)$ is a linear combination of monomials consisting of $r-2$ factors of $\mu_{\alpha\alpha}(w)$ and one factor of $\mu_{\alpha, \alpha+1}(w)$. Therefore, $\{\mu_{i+d, i}(z), T_{r, r}(w)\}$ contains no degree $r$ and $r+1$ components, and $\{\mu_{i+d, i}(z), T_{r, r-1}(w)\}$ contains no degree $r$ components. This proves that there are no degree $r$ or $r+1$ components in $\{\mu_{i+d, i}(z), T_r(w)\}$.

    Finally, we prove that the terms of degree $m$ vanish for $2 \le m \le r-1$. The computation is similar to the proof of Lemma~\ref{lem:well1} so we omit the details and only record the resulting terms. For a term of degree $m+1$ to contribute to the degree $m$ part under the Poisson bracket with $\mu_{i+d, i}(z)$, it must contain a factor of the form $\mu_{i, i+d-2}(w)$ or $\mu_{i+1, i+d-1}(w)$. The remaining degree $m$ terms are given by
    \begin{equation}\label{eq:AppendixC5}
    \begin{aligned}
        &-\delta\left( \frac{q^{\gamma_l+1}w}{z}\right) \mu_{\alpha_1\beta_1}(q^{\gamma_1}w) \cdots \mu_{\alpha_{l-1}\beta_{l-1}}(q^{\gamma_{l-1}}w)\mu_{i+d-1,\beta_{l}}(q^{\gamma_{l}}w) \cdots \mu_{\alpha_{m}\beta_{m}}(q^{\gamma_{m}}w)\\
        &+\delta\left(\frac{q^{\gamma_l}w}{z}\right) \mu_{\alpha_1\beta_1}(q^{\gamma_1}w) \cdots \mu_{\alpha_{l-1},i} (q^{\gamma_{l-1}}w)\mu_{\alpha_l\beta_l}(q^{\gamma_l}w) \cdots \mu_{\alpha_m\beta_m}(q^{\gamma_m}w).
    \end{aligned}
    \end{equation}
    The terms of degree $m$ in $T_r(w)$ that contribute to the degree $m$ part of the Poisson bracket with $\mu_{i+d, i}(z)$ are given by
    \begin{equation}\label{eq:AppendixC6}
    \begin{aligned}
        &-\delta\left(\frac{q^{\gamma_l}w}{z}\right)\mu_{\alpha_1\beta_1}(q^{\gamma_1}w) \cdots \mu_{\alpha_{l-1},i}(q^{\gamma_{l-1}}w)\mu_{\alpha_l\beta_l}(q^{\gamma_l}w) \cdots \mu_{\alpha_m\beta_m}(q^{\gamma_m}w)\\
        &+\mu_{\alpha_1\beta_1}(q^{\gamma_1}w) \cdots \mu_{\alpha_{l-1}\beta_{l-1}}(q^{\gamma_{l-1}}w)\times\\
        &\quad \cdots \times \left( \delta\left(\frac{q^{\gamma_l+1}w}{z}\right) \mu_{i+d-1, \beta_l}(q^{\gamma_l}w) + \delta\left(\frac{q^{\gamma_l+1}w}{z}\right) \mu_{i+d, \beta_l +1} (q^{\gamma_l}w)\right)\cdots \mu_{\alpha_m\beta_m}(q^{\gamma_m}w) \\
        &- \delta\left(\frac{q^{\gamma_l}w}{z}\right)\mu_{\alpha_1\beta_1}(q^{\gamma_1}w) \cdots \mu_{\alpha_{l-1},i}(q^{\gamma_{l-1}}w)\mu_{i+d, \beta_l}(q^{\gamma_l}w) \cdots \mu_{\alpha_m\beta_m}(q^{\gamma_m}w).
    \end{aligned}
    \end{equation}
    The terms of degree $m+1$ in $T_r(w)$ that contribute to the degree $m$ part of the Poisson bracket with $\mu_{i+d, i}(z)$ are given by
    \begin{equation}\label{eq:AppendixC7}
    \begin{aligned}
        &-\delta\left(\frac{q^{\gamma_l+1}w}{z}\right)\mu_{\alpha_1\beta_1}(q^{\gamma_1}w) \cdots \mu_{\alpha_{l-1}\beta_{l-1}}(q^{\gamma_{l-1}}w)\mu_{i+d, \beta_l+1}(q^{\gamma_l}w) \cdots \mu_{\alpha_m\beta_m}(q^{\gamma_m}w)\\
        &+ \delta\left(\frac{q^{\gamma_l}w}{z}\right)\mu_{\alpha_1\beta_1}(q^{\gamma_1}w) \cdots \mu_{\alpha_{l-1},i}(q^{\gamma_{l-1}}w)\mu_{i+d, \beta_l}(q^{\gamma_l}w) \cdots \mu_{\alpha_m\beta_m}(q^{\gamma_m}w).
    \end{aligned}
    \end{equation}
    Therefore, combining \eqref{eq:AppendixC5}, \eqref{eq:AppendixC6}, and \eqref{eq:AppendixC7}, we conclude that the degree $m$ terms of $\{\mu_{i+d, i}(z), T_r(w)\}$ vanish.
    \end{proof}

\newpage









\end{document}